\PassOptionsToPackage{unicode}{hyperref}
\PassOptionsToPackage{hyphens}{url}
\PassOptionsToPackage{dvipsnames,svgnames,x11names}{xcolor}
\documentclass[
  12pt,
]{article}
\usepackage{xcolor}
\usepackage[margin=1in]{geometry}
\usepackage{amsmath,amssymb}
\usepackage{iftex}
\ifPDFTeX
  \usepackage[T1]{fontenc}
  \usepackage[utf8]{inputenc}
  \usepackage{textcomp} 
\else 
  \usepackage{unicode-math} 
  \defaultfontfeatures{Scale=MatchLowercase}
  \defaultfontfeatures[\rmfamily]{Ligatures=TeX,Scale=1}
\fi
\usepackage{lmodern}
\ifPDFTeX\else
\fi
\IfFileExists{upquote.sty}{\usepackage{upquote}}{}
\IfFileExists{microtype.sty}{
  \usepackage[]{microtype}
  \UseMicrotypeSet[protrusion]{basicmath} 
}{}
\usepackage{graphicx}
\makeatletter
\newsavebox\pandoc@box
\newcommand*\pandocbounded[1]{
  \sbox\pandoc@box{#1}%
  \Gscale@div\@tempa{\textheight}{\dimexpr\ht\pandoc@box+\dp\pandoc@box\relax}%
  \Gscale@div\@tempb{\linewidth}{\wd\pandoc@box}%
  \ifdim\@tempb\p@<\@tempa\p@\let\@tempa\@tempb\fi
  \ifdim\@tempa\p@<\p@\scalebox{\@tempa}{\usebox\pandoc@box}%
  \else\usebox{\pandoc@box}%
  \fi%
}
\def\fps@figure{htbp}
\makeatother
\NewDocumentCommand\citeproctext{}{}
\NewDocumentCommand\citeproc{mm}{%
  \begingroup\def\citeproctext{#2}\cite{#1}\endgroup}
\makeatletter
 \let\@cite@ofmt\@firstofone
 \def\@biblabel#1{}
 \def\@cite#1#2{{#1\if@tempswa , #2\fi}}
\makeatother
\newlength{\cslhangindent}
\newlength{\csllabelwidth}
\newenvironment{CSLReferences}[2] 
 {\begin{list}{}{%
  \setlength{\itemindent}{0pt}
  \setlength{\leftmargin}{0pt}
  \setlength{\parsep}{0pt}
  \ifodd #1
   \setlength{\leftmargin}{\cslhangindent}
   \setlength{\itemindent}{-1\cslhangindent}
  \fi
  \setlength{\itemsep}{#2\baselineskip}}}
 {\end{list}}
\usepackage{calc}

\providecommand{\tightlist}{%
  \setlength{\itemsep}{0pt}\setlength{\parskip}{0pt}}
\usepackage{algorithm}
\usepackage{algorithmic}
\usepackage{amsmath}
\usepackage{amsthm}
\usepackage{amssymb}
\usepackage{bm}
\usepackage{booktabs}
\usepackage{braket}
\usepackage{caption}
\usepackage{float}
\usepackage[bottom,flushmargin,hang,multiple]{footmisc}
\usepackage{graphicx}
\usepackage{hyperref}
\usepackage{libertinus}
\usepackage[mathcal]{euscript}
\usepackage{multirow}
\usepackage{pmboxdraw}
\usepackage{quoting}
\usepackage{setspace}
\usepackage{soul}
\usepackage{subcaption}
\usepackage{tabularx}
\usepackage{titlesec}
\usepackage{fvextra}
\DefineVerbatimEnvironment{verbatim}{Verbatim}{breaklines,fontsize=\footnotesize}
\renewcommand{\texttt}[1]{\textit{#1}}
\usepackage{bookmark}
\IfFileExists{xurl.sty}{\usepackage{xurl}}{} 
\hypersetup{
  pdftitle={Beyond Pairwise Comparisons: A Distributional Test of Distinctiveness for Machine-Generated Works in Intellectual Property Law},
  colorlinks=true,
  linkcolor={Maroon},
  filecolor={Maroon},
  citecolor={Blue},
  urlcolor={blue},
  pdfcreator={LaTeX via pandoc}}

\title{Beyond Pairwise Comparisons: A Distributional Test of
Distinctiveness for Machine-Generated Works in Intellectual Property
Law}
\author{}
\date{\vspace{-2.5em}}

\begin{document}
\maketitle

\quotingsetup{font={itshape}, leftmargin=2em, rightmargin=2em, vskip=1ex}

\vspace*{\fill}

\begin{center}
Anirban Mukherjee

Hannah Hanwen Chang

\bigskip

26 January, 2026
\end{center}

\vspace*{\fill}

\noindent \hrulefill

\noindent Anirban Mukherjee
(\href{mailto:anirban@avyayamholdings.com}{\nolinkurl{anirban@avyayamholdings.com}})
is Principal at Avyayam Holdings. Hannah H. Chang
(\href{mailto:hannahchang@smu.edu.sg}{\nolinkurl{hannahchang@smu.edu.sg}};
corresponding author) is Associate Professor of Marketing at the Lee
Kong Chian School of Business, Singapore Management University. This
research was supported by the Ministry of Education (MOE), Singapore,
under its Academic Research Fund (AcRF) Tier 2 Grant,
No.~MOE-T2EP40124-0005.

\newpage

\newtheorem{proposition}{Proposition} 

\begin{center} 
\noindent \textbf{Abstract}
\end{center}

\noindent Key doctrines, including novelty (patent), originality
(copyright), and distinctiveness (trademark), turn on a shared empirical
question: whether a body of work is meaningfully distinct from a
relevant reference class. Yet analyses typically operationalize this
set-level inquiry using item-level evidence: pairwise comparisons among
exemplars. That unit-of-analysis mismatch may be manageable for finite
corpora of human-created works, where it can be bridged by \emph{ad hoc}
aggregations. However, it becomes acute for machine-generated works,
where the object of evaluation is not a fixed set of works but a
generative process with an effectively unbounded output space.

We propose a distributional alternative: a two-sample test based on
maximum mean discrepancy computed on semantic embeddings to determine if
two creative processes---whether human or machine---produce
statistically distinguishable output distributions. The test requires no
task-specific training---obviating the need for discovery of proprietary
training data to characterize the generative process---and is
sample-efficient, often detecting differences with as few as 5--10
images and 7--20 texts.

We validate the framework across three domains: handwritten digits
(controlled images), patent abstracts (text), and AI-generated art
(real-world images). Substantively, we reveal a perceptual paradox: even
when human evaluators distinguish AI outputs from human-created art with
only \textasciitilde58\% accuracy, our method detects robust
distributional distinctiveness. Our results present evidence contrary to
the view that generative models act as mere regurgitators of training
data. Rather than producing outputs statistically indistinguishable from
a human baseline---as simple regurgitation would predict---they produce
outputs that are semantically human-like yet stochastically distinct,
suggesting their dominant function is as a semantic interpolator within
a learned latent space.

\begin{center}\rule{0.5\linewidth}{0.5pt}\end{center}

\noindent Keywords: Novelty, Originality, Distinctiveness, Copyright,
Patent, Trademark, Intellectual Property Law, Artificial Intelligence.


\smallskip

\newpage
\doublespacing

\section{Introduction}\label{introduction}

Intellectual property law is built on a shared empirical premise: to
warrant protection, a work must be meaningfully distinct from a relevant
reference class. Whether assessing ``novelty'' in patent law,
``originality'' in copyright, or ``distinctiveness'' in trademark, the
inquiry turns on comparing a candidate work against a corpus of existing
material---be it prior art, the public domain, or registered
marks.\footnote{Foundational cases establishing this comparative
  framework include Graham v. John Deere Co., 383 U.S. 1 (1966) (patent
  nonobviousness); Feist Publ'ns, Inc.~v. Rural Tel. Serv. Co., 499 U.S.
  340 (1991) (copyright originality); and Abercrombie \& Fitch Co.~v.
  Hunting World, Inc., 537 F.2d 4 (2d Cir. 1976) (trademark
  distinctiveness).} While these legal standards differ in threshold and
scope, the underlying quantitative question is structurally analogous:
does the new work diverge sufficiently from the distribution of existing
works to justify an exclusive right?\footnote{\emph{See} 35 U.S.C. §§
  102, 103 (2018) (patent); 17 U.S.C. § 102(a) (2018) (copyright);
  Lanham Act, 15 U.S.C. §§ 1051--1141n (2018) (trademark).}

Yet a fundamental unit-of-analysis mismatch pervades current practice.
While inquiries often concern the distinctiveness of a work from a
portfolio of works, the available metrics are almost exclusively
item-wise and pairwise (\citeproc{ref-hain2022text}{Hain et al. 2022};
\citeproc{ref-vsavelka2022legal}{Šavelka and Ashley 2022}). For
instance, while a court can assess a single painting's features or
compare two specific paintings, neither approach directly reveals
whether that painting is systematically distinct from an entire
portfolio of paintings---or whether one portfolio, such as an artist's
body of work, is as a whole systematically distinct from another.
Attempts to capture such portfolio-level distinctiveness using existing
methods inevitably depend on either qualitative gestalt judgments or
\emph{ad hoc} aggregations of pairwise distance metrics---such as mean
or maximum similarity---that lack a principled statistical basis
(\citeproc{ref-helmers2019automating}{Helmers et al. 2019};
\citeproc{ref-lin2023measuring}{Lin et al. 2023}).

The emergence of generative AI renders this mismatch intractable.
Generative models are stochastic processes with effectively unbounded
output spaces (\citeproc{ref-chesterman2025good}{Chesterman 2025}).
Consequently, comparing only a finite sample of their outputs to a
reference class yields an inherently incomplete assessment---whether a
sample appears unduly similar or dissimilar to the reference class may
be a matter of happenstance rather than a systematic occurrence.
Furthermore, the generative process itself is often obscured behind
proprietary firewalls. In adversarial litigation, access to training
data or model weights is frequently unavailable or contested, preventing
direct audits of whether a system is memorizing or
interpolating.\footnote{\emph{See, e.g.}, Andersen v. Stability AI Ltd.,
  No.~3:23-cv-00201, 2023 WL 7132064 (N.D. Cal. Oct.~30, 2023) (granting
  motions to dismiss in substantial part while allowing direct
  infringement claim based on training-stage copying to proceed;
  requiring clearer allegations for output-based and derivative work
  theories); Getty Images (US), Inc.~v. Stability AI Ltd., {[}2025{]}
  EWHC 2863 (Ch) (allowing trademark claims to proceed where plaintiff's
  watermarks appeared in model outputs; dismissing most copyright claims
  for lack of UK jurisdiction).} Finally, the challenge is compounded
when AI outputs themselves enter the reference class. When both the
candidate source and the reference class possess effectively infinite
cardinality, pairwise exhaustion becomes impracticable.

Compounding this structural failure is a collapse of the law's primary
instrument of measurement: the human proxy. Courts typically rely on the
``ordinary observer'' or ``person having ordinary skill in the art'' to
intuit distance between creative works. Yet this reliance has become
untenable. Recent empirical work finds that human evaluators distinguish
AI-generated images from human creations only marginally better than
chance---about 58\% for AI-generated art
(\citeproc{ref-silva2024artbrain}{Silva et al. 2024}) and about 62\% for
photorealistic images (\citeproc{ref-roca2025good}{Roca et al. 2025});
given recent advances in generative AI, even these figures likely
\emph{overstate} current human evaluator performance
(\citeproc{ref-roca2025good}{Roca et al. 2025, p. 8}). Moreover, if
ordinary observers struggle, one might expect experts to fare better.
They do not. In 2022, Jason Allen's \emph{Théâtre D'opéra Spatial},
created with Midjourney, won the Colorado State Fair's digital arts
category; one judge later acknowledged he ``didn't realize that it was
generated by AI when judging it.''\footnote{\emph{See} Kevin Roose,
  \emph{An A.I.-Generated Picture Won an Art Prize. Artists Aren't
  Happy}, \textsc{N.Y. Times} (Sept.~2, 2022),
  \url{https://www.nytimes.com/2022/09/02/technology/ai-artificial-intelligence-artists.html}.}
In 2025, an entry in Clip Studio Paint's International Illustration
Contest won a prize before being withdrawn for using generative
AI---despite a multi-stage review that included both detection tools and
visual inspection.\footnote{\emph{See} Press Release, Celsys, Inc.,
  Announcing the Winners of the 44th International Illustration Contest
  (Aug.~19, 2025) (noting entry withdrawal).} When even trained
evaluators with professional stakes cannot reliably distinguish AI
outputs from human creations, the law is left without a stable
perceptual yardstick for adjudicating the distinctiveness of
machine-generated content.

This paper proposes a distributional alternative: a two-sample test
based on maximum mean discrepancy (MMD) computed on semantic
embeddings.\footnote{Our framework operationalizes the \emph{empirical}
  inquiry, not the \emph{normative} legal conclusion. The selection of
  the relevant reference class---what constitutes the ``prior art'' or
  ``market''---remains a matter of legal argumentation. Legal doctrines
  also often incorporate additional factors (e.g., market overlap,
  commercial purpose, or consumer sophistication). These lie beyond the
  purely semantic comparison that our metric provides.} MMD is a
kernel-based statistical metric that evaluates samples collectively to
determine if they are drawn from the same underlying distribution
(\citeproc{ref-gretton2012kernel}{Gretton et al. 2012}). To capture
\emph{meaning} rather than mere surface form, we pair MMD with semantic
embeddings---mappings of text or images into high-dimensional vector
spaces such that semantic relationships are preserved
(\citeproc{ref-mikolov2013efficient}{Mikolov et al. 2013};
\citeproc{ref-chalkidis2019deep}{Chalkidis and Kampas 2019};
\citeproc{ref-radford2021learning}{Radford et al. 2021}). This approach
resolves the tripartite challenge identified above: it is
\emph{distributional} (solving the unit-of-analysis mismatch),
\emph{training-free} (solving the opacity problem), and \emph{objective}
(solving the perceptual failure).

Unlike standard machine learning performance metrics (such as Fréchet
Inception Distance) that output raw scores
(\citeproc{ref-heusel2017gans}{Heusel et al. 2017}), we implement MMD as
a hypothesis test. By employing permutation testing, we convert the
distributional distance into a \emph{p}-value, allowing a fact-finder to
determine---at a chosen standard of proof---whether the generative
process is statistically distinct from the reference class. This
provides a principled way to evaluate whether generative models operate
as mere regurgitators of their training data
(\citeproc{ref-bender2021dangers}{Bender et al. 2021}). If an AI model
simply resamples patterns from its training corpus, its output
distribution should be statistically indistinguishable from a human
reference class.\footnote{If models merely ``stitch together'' their
  training patterns, their output distribution should converge to the
  training distribution in expectation. Support for this line of
  reasoning can be found in prior investigations. \emph{See, e.g.},
  Somepalli et al. (\citeproc{ref-somepalli2023diffusion}{2023})
  (investigating whether diffusion models merely replicate training data
  or generate novel compositions, finding that models can regurgitate
  training examples); \emph{cf.} Goodfellow et al.
  (\citeproc{ref-goodfellow2014generative}{2014}) (defining the
  theoretical objective of generative modeling as the minimization of
  divergence between the model distribution and the data distribution,
  implying that a model optimizing this objective tends toward
  statistical indistinguishability from its training set).} Conversely,
if this null hypothesis is rejected, it provides statistical evidence of
``interpolative distinctiveness'': the model is generating outputs that,
while semantically coherent, occupy a distinct topological region of the
creative space. While it may retain the capacity for rare memorization,
its dominant mode of operation is to function as a semantic interpolator
within a learned latent space.

We validate this methodology in three stages. First, we establish
statistical robustness using the MNIST dataset, confirming the method's
sensitivity to known ground-truth differences. Second, we demonstrate
domain versatility using patent abstracts, confirming the method's
ability to distinguish between technical fields based on textual
semantic embeddings. Third, we apply the framework to the AI-ArtBench
dataset (\citeproc{ref-silva2024artbrain}{Silva et al. 2024}), where we
reveal a perceptual paradox: even where human evaluators fail to
distinguish AI outputs from human art, our test detects robust
distributional distinctiveness. We then extend the analysis to recent
generative AI models, tracing how the distributional distinctiveness of
outputs evolves as generative AI advances. Our finding suggests that
these models are not merely replaying prior art but are engaging in
interpolative creativity. Practically, we demonstrate that this
distinctiveness can be detected with notable sample efficiency---often
requiring as few as 5 to 10 samples per group for images, and 7 to 20
for text---offering a scalable evidentiary tool for courts that must
evaluate novelty based on limited portfolios.

\subsection{Contribution}\label{contribution}

Our contribution is threefold. First, we introduce a framework that
shifts the legal unit of analysis from the \emph{item} to the
\emph{process}, resolving the infinite cardinality problem inherent in
machine creativity. Second, we provide a sample-efficient implementation
that requires no task-specific training, allowing courts and IP offices
to assess distinctiveness without requiring access to proprietary model
weights or massive training datasets. Third, we provide substantive
evidence that the dominant mode of generative models is not mere
regurgitation, identifying the phenomenon of ``interpolative
distinctiveness'' where AI outputs are semantically human-like yet
stochastically distinct. The remainder of the paper is organized as
follows: Section 2 reviews the limitations of current legal and
technical metrics; Section 3 details the mathematical derivation of our
MMD-based framework; Section 4 validates the methodology using the MNIST
dataset; Section 5 extends the validation to the textual domain using
patent abstracts; Section 6 applies the framework to the AI-ArtBench
dataset to reveal the perceptual paradox; and Section 7 discusses the
doctrinal and evidentiary implications for copyright, patent, and
trademark law.

\section{Distinctiveness in Literature and
Practice}\label{distinctiveness-in-literature-and-practice}

Distinctiveness requires a metric that matches the ontology of the
subject. Yet, whether in computer science or law, the prevailing tools
for measuring distinctiveness remain structurally misaligned with
generative AI. These tools attempt to adjudicate \emph{process-level}
phenomena using \emph{item-level} evidence, relying on pairwise
comparisons and subjective proxies that falter when applied to the
effectively unbounded output spaces of machine creativity, while
data-dependent metrics are impracticable in litigation where discovery
is limited.

In this section, we map this methodological gap. We first critique
existing academic metrics---from verbatim memorization checks to
pairwise semantic similarity and supervised distinctiveness
classifiers---demonstrating how they fail to capture the systemic nature
of generative models or require data inaccessible in litigation. We then
examine how patent, copyright, and trademark law currently
operationalize distinctiveness, showing why human-centric standards
cannot scale to the infinite cardinality of machine creativity. Finally,
we review the theoretical impasse between regurgitation-based critiques
and probabilistic notions of novelty, arguing that this debate remains
unresolved precisely because empirical researchers lack a distributional
framework capable of distinguishing between mechanical regurgitation and
creative interpolation.

\subsection{Distinctiveness in the Academic
Literature}\label{distinctiveness-in-the-academic-literature}

Research on quantifying the distinctiveness of creative works spans
legal scholarship, natural language processing (NLP), and cultural
evolution. Prevailing methodologies fall into three primary streams: (1)
verbatim memorization metrics, (2) pairwise semantic similarity
measures, and (3) supervised proxies for human judgment. These
approaches operate fundamentally at the \emph{item} level---scoring
specific sentences, images, or marks---rather than at the level of
creative processes. A fourth emerging stream---distributional metrics
from machine learning---adopts the right structural approach but remains
misaligned with legal objectives, prioritizing fidelity over
distinctiveness. Consequently, the literature currently lacks a
framework capable of quantifying the divergence of a generative process
or the distinctiveness of a set of works from its reference class.

\subsubsection{Verbatim Memorization and Exact
Matching}\label{verbatim-memorization-and-exact-matching}

Research in AI safety has developed rigorous metrics to detect when
models reproduce training data verbatim. For instance, Carlini et al.
(\citeproc{ref-carlini2023quantifying}{2023}) quantify memorization by
prompting models with prefixes taken from the training set and counting
how often the model reproduces long suffixes \emph{exactly}; they also
check whether generated continuations appear as verbatim substrings
elsewhere in the training corpus. Chang et al.
(\citeproc{ref-chang2023speak}{2023}) instead employ
membership-inference style tests: carefully constructed cloze prompts
over copyrighted books probe whether a model reliably fills in missing
passages, revealing when it has effectively memorized the underlying
text. The RAVEN framework by McCoy et al.
(\citeproc{ref-mccoy2023much}{2023}) assesses novelty by measuring, for
each generated text, the proportion of structural patterns (e.g.,
\(N\)-grams, syntactic configurations) that do not appear in the
training corpus.

While vital for detecting data leakages, these metrics are insufficient
for assessing distinctiveness. They are inherently brittle; a model that
paraphrases a copyrighted work or mimics an artistic style without exact
word-level or pixel-level replication would score as ``novel'' under
these metrics, despite potentially lacking independent creation.
Furthermore, they typically require access to the full training
corpus---which is often unavailable for proprietary foundation models in
litigation or administrative review. Ultimately, they measure
\emph{replication}, not \emph{distinctiveness}, and they do so at the
level of individual samples rather than the underlying creative process.
They cannot determine whether any replication is the system's governing
dynamic or merely a statistical anomaly.

\subsubsection{Pairwise Semantic
Similarity}\label{pairwise-semantic-similarity}

To capture meaning rather than mere surface form, the dominant
methodological paradigm is \emph{pairwise semantic similarity} using
semantic embeddings. In this framework, documents or images are mapped
to high-dimensional vectors (using models like BERT for text or
ResNet/VGG for images), and distinctiveness is measured as the geometric
distance (typically cosine or Euclidean) between a specific candidate
work and specific prior art references.

This approach is prevalent in legal informatics and patent analysis. For
example, Westermann et al. (\citeproc{ref-westermann2020sentence}{2020})
use sentence embeddings and approximate nearest-neighbor search to
retrieve case-law paragraphs that are semantically similar to a query
sentence. In patent and scientometrics, Hain et al.
(\citeproc{ref-hain2022text}{2022}) describe a text-embedding-based
method for computing cosine similarity between patents at scale, using
these pairwise distances to study technological relatedness and
knowledge flows, while Shibayama et al.
(\citeproc{ref-shibayama2021measuring}{2021}) use word embeddings to
define recombinant novelty scores based on how unusually distant a
paper's cited references are from one another. More recently, Lin et al.
(\citeproc{ref-lin2023measuring}{2023}) use multimodal embeddings to
measure patent similarity based on both text and image recognition, and
Chiba-Okabe and Su (\citeproc{ref-chibaokabe2024tackling}{2024}) measure
originality as the expected distance between a given image and samples
drawn from a baseline distribution of context-conditioned or ``generic''
images.

However, pairwise methods face a fundamental scaling limitation when
applied to generative AI. Assessing an AI model's distinctiveness
requires comparing its \emph{generative potential} against potentially
infinite sets (e.g., the entirety of prior art). One cannot feasibly
compute the distance between every potential AI output and every
existing human work. In practice, analysts must rely on finite samples,
treating observed distances as deterministic rather than as outcomes of
a probabilistic process. Moreover, aggregating pairwise distances---such
as taking the mean or minimum distance---yields a statistically
impoverished metric as it collapses many pairwise distances into
scalars, while also forcing an arbitrary choice among scalars (e.g.,
mean, median, or variance) that each characterize the distances
differently. As Lin et al. (\citeproc{ref-lin2023measuring}{2023})
argue, such ad hoc aggregations often depend on arbitrary choices that
lack a principled statistical basis (p.~2). This reduction also discards
the full distributional structure; without this context, the aggregates
lack a null model, leaving courts without a way to determine whether an
observed distance is statistically significant (i.e., unlikely to occur
purely by chance).

\subsubsection{Subjective Proxies and Supervised
Classification}\label{subjective-proxies-and-supervised-classification}

Studies have sought to automate distinctiveness judgments by training
models to mimic legal decision-makers. In trademark, Adarsh et al.
(\citeproc{ref-adarsh2024automating}{2024}) build a classifier over
roughly 1.5 million USPTO applications, using examiner office actions to
infer whether a mark was treated as inherently distinctive under the
Abercrombie spectrum, and then fine-tune transformer-based language
models to predict that label. Similarly, Xu and Ashley
(\citeproc{ref-xu2025labelfree}{2025}) generate synthetic ``anchor''
marks spanning Abercrombie categories, obtain pairwise distinctiveness
judgments for anchor--anchor and anchor--real-mark comparisons from
large language models (LLMs), and fit a Bradley--Terry model to these
comparisons to derive a continuous ``distinctiveness score.''

While these methods align closely with legal doctrine, they are
inherently limited by their reliance on labeled or pseudo-labeled
judgments. First, they are predictive rather than metric: they forecast
a \emph{legal conclusion} (e.g., ``inherently distinctive'' versus
``descriptive'') rather than quantifying the \emph{empirical distance}
between creative processes. Second, they ultimately inherit the limits
of human (or human-simulating) perception: Adarsh et al.
(\citeproc{ref-adarsh2024automating}{2024}) use examiner decisions as
ground truth, and Xu and Ashley (\citeproc{ref-xu2025labelfree}{2025})
anchor their scale in Abercrombie-style categories and LLMs trained to
emulate legal intuition. This is especially problematic in the context
of AI-generated art, where recent empirical work demonstrates that human
evaluators distinguish AI-generated images from human art with only
approximately 58\% accuracy (\citeproc{ref-silva2024artbrain}{Silva et
al. 2024}). If human perception collapses at the item level, supervised
models trained to replicate those judgments lose their utility as
ground-truth metrics for distinctiveness.

\subsubsection{The Distributional Gap}\label{the-distributional-gap}

Theoretical legal scholarship has long recognized that distinctiveness
should ideally be conceptualized probabilistically. Vermont
(\citeproc{ref-vermont2012sine}{2012}) argues that copyright's proper
domain is \emph{unique} works---those that are ``novel and
unrepeatable,'' in the sense that no other creator is likely to
independently produce the same work. Byron
(\citeproc{ref-byron2006tying}{2006}) develops a probability theory of
``copyrightable creativity,'' under which a work's protectability turns
on how unlikely it is to be created given the constraints of the
expressive field. Both accounts treat distinctiveness as a property of
an underlying \emph{distribution}: originality becomes ``probabilistic
uniqueness,'' the chance that a work would \emph{not} arise again from
another draw on the same cultural and technological resources.

However, while legal theory anticipates a distributional framework,
existing distributional metrics in machine learning were developed with
an inverted objective: establishing fidelity rather than
distinctiveness. Metrics such as FID and KID compare the distributions
of deep features for real and generated images. They are used as
\emph{performance metrics}: lower FID or KID indicates that a generative
model more faithfully reproduces the empirical image distribution
(\citeproc{ref-heusel2017gans}{Heusel et al. 2017};
\citeproc{ref-binkowski2018demystifying}{Bińkowski et al. 2018};
\citeproc{ref-naeem2020reliable}{Naeem et al. 2020};
\citeproc{ref-wang2025distributed}{Wang et al. 2025}). While these
metrics are undeniably distributional, they function as uncalibrated
\emph{scores}: a smaller number is ``better,'' yet there is no
principled answer to questions like whether a FID difference of 10
versus 20 corresponds to a meaningful degree of distinctiveness.
Crucially, they do not present hypothesis tests; there is no associated
null distribution, \emph{p}-value, or confidence interval, and thus no
way to translate a particular score into a statement such as ``with 99\%
confidence, these two generative processes differ.''

Recent work on ``novelty'' pushes closer to the legal question but still
falls short of evidentiary needs. Zhang et al.
(\citeproc{ref-zhang2024interpretable}{2024}) propose Kernel-based
Entropic Novelty (KEN) to quantify \emph{mode-based novelty},
identifying types that occur more frequently in a generative model than
in a reference model. Kim et al. (\citeproc{ref-kim2022mutual}{2022})
introduce mutual information divergence (MID) to evaluate text--image
generative models. While both metrics directly target distributional
differences and correlate well with human judgments, they are both
scalar metrics. Neither yields a probability or level of confidence that
a distribution is distinct from another.

Consequently, a methodological gap remains. There is currently no
established framework, particularly in legal or forensic contexts, that
combines the \emph{semantic richness} of modern embeddings with the
\emph{rigorous hypothesis testing} of statistics. Our proposed framework
fills this gap by repurposing kernel methods from performance evaluation
to forensic analysis. By wrapping kernel mean embeddings in a two-sample
testing procedure, we convert raw distances into p‑values, providing a
way to determine---at an explicitly chosen standard of proof---whether
two creative processes are statistically distinct.

Thus, we offer a complementary analytical pathway to traditional
verbatim memorization checks. While these metrics assess whether
\emph{any} specific example of a candidate class or output of a
candidate process is similar to a member of a reference
class---detecting even rare instances of potential infringement---we
assess the \emph{systematic} tendencies of a candidate class or process.
A process that is predominantly distributionally distinct may yet
produce regurgitative content; conversely, a process that avoids exact
verbatim matches may yet produce outputs with substantial similarity to
a reference class. Legal adjudication therefore is likely to benefit
from both forms of scrutiny: memorization tests to detect specific data
leakage, and distributional tests to evaluate the creative independence
of the generative process.

\subsection{Distinctiveness in
Practice}\label{distinctiveness-in-practice}

Legal inquiries of novelty, originality, and distinctiveness share a
common structure: they are comparative and anthropocentric. In each
domain, the law relies on a human proxy---the ``person having ordinary
skill in the art'' (patent), the ``ordinary observer'' (copyright), or
the ``reasonable consumer'' (trademark)---to perform a pairwise
comparison between a specific new work and specific prior art. As we
show below, this reliance on pairwise comparisons and human intuition
creates a methodological gap when applied to the effectively infinite
and non-human output distributions of machines.

\subsubsection{Patent}\label{patent}

Patent law demands a rigorous, structured comparison between a claimed
invention and the ``prior art''---a standard codified as novelty and
non-obviousness. The Supreme Court formalized this process in
\emph{Graham v. John Deere Co.}, establishing a framework that requires
the fact-finder to determine ``the scope and content of the prior art''
and to ascertain the ``differences between the prior art and the claims
at issue.''\footnote{\label{graham}383 U.S. 1, 17 (1966) (establishing
  the framework for non-obviousness).} This is fundamentally a pairwise
exercise: an examiner or court places the claimed invention side-by-side
with specific prior art references to determine if the new claim is
anticipated or rendered obvious by the old. The benchmark for this
measurement is not a quantitative metric, but the PHOSITA---a
hypothetical construct used to determine if the difference between the
two items would have been ``obvious to a person having ordinary skill in
the art.''\footnote{35 U.S.C. § 103 (2018).} This framework has evolved
to prioritize functional predictability over rigid rules, yet it remains
tethered to human cognition.\footnote{In \emph{KSR International Co.~v.
  Teleflex Inc.}, the Supreme Court rejected rigid tests for
  obviousness, holding instead that if a combination of known elements
  yields ``predictable results,'' it lacks the requisite
  distinctiveness. \emph{See} 550 U.S. 398, 416 (2007).}

In the context of AI, the inquiry has shifted from the distinctiveness
of the output to the distinctiveness of the human contribution. In
\emph{Thaler v. Vidal}, the Federal Circuit held that an ``inventor''
under the Patent Act must be a natural person, affirming the rejection
of AI systems as inventors.\footnote{43 F.4th 1207, 1213 (Fed. Cir.
  2022).} Under current USPTO guidance, inventorship for AI-assisted
inventions turns on traditional conception standards: whether a natural
person formed a ``definite and permanent idea of the complete and
operative invention,'' with AI systems treated as tools analogous to
laboratory equipment, effectively rendering the probabilistic novelty
(or lack thereof) introduced by AI outputs irrelevant absent significant
human contribution.\footnote{\emph{See} Revised Inventorship Guidance
  for AI-Assisted Inventions, 90 \textsc{Fed. Reg.} 54,636 (Nov.~28,
  2025) (rescinding prior guidance and clarifying that traditional
  conception standards govern all inventions regardless of AI
  involvement).} Furthermore, for visual innovations---the domain most
relevant to generative art---design patent law employs the ``ordinary
observer'' test established in \emph{Egyptian Goddess, Inc.~v. Swisa,
Inc.}\footnote{543 F.3d 665, 678 (Fed. Cir. 2008) (en banc)
  (establishing the ``ordinary observer'' test for design patent
  infringement).} This test asks whether an ordinary observer, giving
such attention as a purchaser usually gives, would find the two designs
substantially the same. Like the utility patent framework, this test
relies on a gestalt human impression of pairwise similarity rather than
a distributional analysis of the design space.

\subsubsection{Copyright}\label{copyright}

In copyright jurisprudence, the distinctiveness inquiry operates at two
separate stages: the threshold determination of ``originality''
(validity) and the assessment of ``substantial similarity''
(infringement). For validity, the Supreme Court in \emph{Feist
Publications, Inc.~v. Rural Telephone Service Co.} established that
originality requires only independent creation and a ``modicum of
creativity.''\footnote{\label{feist}499 U.S. 340, 346 (1991).} This
threshold is operationally binary, offering no metric for the degree of
divergence. By contrast, the standard for \emph{transformative}
distinctiveness---relevant to fair use---is substantially more
stringent. In \emph{Andy Warhol Foundation v. Goldsmith}, the Supreme
Court held that even a work with a distinct aesthetic (e.g., a
silkscreen treatment of a photograph) fails to be transformative if it
shares the same commercial purpose as the original.\footnote{598 U.S.
  508 (2023).} This suggests that distinctiveness is not merely a
measure of visual difference but also of market function---a nuance that
purely geometric metrics may miss.

When measuring infringement, courts employ a bifurcated approach to
assess the distance between two specific works. The Ninth Circuit's
``extrinsic/intrinsic'' test, articulated in \emph{Sid \& Marty Krofft
Television Productions v. McDonald's Corp.}, exemplifies this pairwise
methodology.\footnote{562 F.2d 1157, 1164 (9th Cir. 1977).} The
``extrinsic'' test employs expert testimony to analytically dissect the
works and compare objective elements (e.g., plot, themes, dialogue),
effectively filtering out unprotectable ideas. If similarity survives
this dissection, the ``intrinsic'' test asks whether an ``ordinary
reasonable person'' would perceive the ``total concept and feel'' of the
works as substantially similar.\footnote{\emph{Id.} at 1167.} The Second
Circuit's ``abstraction-filtration-comparison'' test in \emph{Computer
Associates International, Inc.~v. Altai, Inc.} systematizes this
process, filtering out non-distinctive elements before comparing the
core of protectable expression.\footnote{\emph{See} 982 F.2d 693, 706
  (2d Cir. 1992).} In both frameworks, distinctiveness is measured by
the subjective impression of similarity between two works, rather than
by any objective measure of a work's standing within a broader portfolio
or catalog.

Recent cases involving generative AI expose the limitations of these
metrics. In \emph{Thaler v. Perlmutter}, affirmed by the D.C. Circuit in
2025, courts applied a strict ``human authorship'' requirement, refusing
to consider the distinctiveness of AI outputs---regardless of their
novelty---because they lack a human origin.\footnote{\emph{See} Thaler
  v. Perlmutter, 130 F.4th 1039 (D.C. Cir. 2025), \emph{aff'g} 687 F.
  Supp. 3d 140 (D.D.C. 2023).} Conversely, in infringement litigation
such as \emph{Andersen v. Stability AI Ltd.}, plaintiffs have struggled
to plausibly plead certain theories---such as that AI models contain
``compressed copies'' of training data or that outputs are infringing
derivative works---without showing substantial similarity between
specific outputs and specific copyrighted works.\footnote{\emph{See}
  No.~3:23-cv-00201, 2023 WL 7132064, at *7--8 (N.D. Cal. Oct.~30, 2023)
  (requiring plaintiffs to clarify ``compressed copies'' theory and
  noting that derivative work claims require substantial similarity
  allegations).} Where plaintiffs have succeeded in demonstrating a lack
of distinctiveness, it has been through allegations of verbatim or
near‑verbatim memorization, such as the evidence of text reproduction
cited in \emph{The New York Times Co.~v. Microsoft Corp.} to demonstrate
the model's ``regurgitation'' of articles (e.g., Exhibit J's red‑marked
overlaps).\footnote{\emph{See} Complaint at ¶¶ 105--07, Ex. J, N.Y.
  Times Co.~v. Microsoft Corp., No.~1:23-cv-11195 (S.D.N.Y. Dec.~27,
  2023).} This reliance on verbatim metrics confirms that the law
currently lacks a tool to measure distinctiveness absent exact
replication.

\subsubsection{Trademark}\label{trademark}

Unlike the creativity-focused standards of patent and copyright,
trademark distinctiveness is explicitly market-facing: it asks whether a
mark can differentiate goods or services in the minds of consumers.
Courts operationalize this through a semantic taxonomy known as the
``Abercrombie spectrum.'' Established in \emph{Abercrombie \& Fitch
Co.~v. Hunting World, Inc.}, this framework categorizes marks based on
their linguistic relationship to the product: generic, descriptive,
suggestive, arbitrary, or fanciful.\footnote{537 F.2d 4, 9 (2d Cir.
  1976).} A mark's placement on this spectrum determines its legal
strength. For example, ``arbitrary'' marks (e.g., ``Apple'' for
computers) are deemed inherently distinctive, while ``descriptive''
marks require proof of ``secondary meaning''---an empirical
demonstration that the public associates the term with a specific
source.\footnote{\emph{Id.} at 10.} Courts operationalize this inquiry
through multi-factor tests, such as the \emph{Polaroid} factors, which
explicitly require a pairwise assessment of the ``similarity of the
marks'' against specific prior uses.\footnote{Polaroid Corp.~v. Polarad
  Elecs. Corp., 287 F.2d 492 (2d Cir. 1961).} The measurement is
primarily conceptual and linguistic, relying on judicial intuition
regarding consumer understanding rather than a systematic comparison of
the mark against the distribution of existing commercial symbols.

This market-facing inquiry bifurcates when applied to visual trade
dress. In \emph{Two Pesos, Inc.~v. Taco Cabana, Inc.}, the Supreme Court
held that the décor of a Mexican restaurant could be ``inherently
distinctive'' based on a jury's gestalt judgment of its overall look and
feel.\footnote{\emph{See} 505 U.S. 763, 773--74 (1992).} In contrast,
the Court in \emph{Wal-Mart Stores, Inc.~v. Samara Brothers, Inc.} ruled
that product design (e.g., the cut of children's clothing) can
\emph{never} be inherently distinctive and always requires empirical
proof of secondary meaning, typically via consumer surveys.\footnote{529
  U.S. 205, 212--13 (2000).} These arguments have recently been extended
to digital assets. In \emph{Hermès International v. Rothschild} (the
``MetaBirkins'' case), the court denied summary judgment, and a jury
subsequently found that NFT iterations of physical handbags were not
distinct enough to avoid confusion, prioritizing consumer perception
over the technical novelty of the digital medium.\footnote{Hermès Int'l
  v. Rothschild, 654 F. Supp. 3d 268 (S.D.N.Y. 2023) (denying
  cross-motions for summary judgment); \emph{see also} Verdict Form at
  1, Hermès Int'l v. Rothschild, No.~1:22-cv-00384 (S.D.N.Y. Feb.~8,
  2023), ECF No.~146 (finding for Hermès on trademark infringement,
  dilution, and cybersquatting claims).} Thus, the law oscillates
between subjective ``gestalt'' impressions and ad hoc empirical data,
lacking a unified metric for visual distinctiveness.

Paradoxically, this reliance on cognitive association may enable
trademark claims based on AI ``hallucinated'' marks. In \emph{Getty
Images (US), Inc.~v. Stability AI, Inc.}, the plaintiff alleged that
AI-generated images containing distorted, illegible versions of the
Getty watermark nonetheless infringed its trademark.\footnote{Complaint,
  Getty Images (US), Inc.~v. Stability AI, Inc., No.~1:23-cv-00135 (D.
  Del. Feb.~3, 2023).} A recent decision by the English High Court in
the parallel UK litigation found limited trademark infringement where
early versions of Stable Diffusion generated outputs displaying Getty's
watermarks, though the court rejected Getty's secondary copyright claims
and characterized its trademark findings as ``extremely
limited.''\footnote{Getty Images (US) Inc.~v. Stability AI
  Ltd.~{[}2025{]} EWHC 2863 (Ch) (Eng.) (finding trademark infringement
  under ss. 10(1) and 10(2) of the Trade Marks Act 1994 for certain
  watermark-bearing outputs, while dismissing the s. 10(3) claim and
  rejecting the secondary copyright infringement theory).} Under
traditional confusion analysis, the distinctiveness of the watermark is
measured not by its fidelity to the original, but by its capacity to
trigger a cognitive association with the source in the mind of the
consumer. If the ``hallucination'' retains the essential visual
characteristics of the mark, it is legally indistinguishable from the
genuine article, regardless of its origins.

\subsection{Regurgitation vs.~Interpolation: Machine Creativity and
Distinctiveness}\label{regurgitation-vs.-interpolation-machine-creativity-and-distinctiveness}

Current intellectual property frameworks operate on a strong presumption
against AI distinctiveness, viewing generative systems as tools rather
than independent creators. This position relies on the ``human
authorship'' requirement, which posits that AI systems lack the
requisite mental conception to produce protectable work. As Ginsburg and
Budiardjo (\citeproc{ref-ginsburg2019authors}{2019}) argue, because
machines cannot formulate creative plans, they lack the ``initiative
that characterizes human authorship'' and are ``closer to amanuenses
than to true `authors'\,'' (p.~349). Under this view, the human
programmer or prompter provides the creative rules, and the machine
merely executes them (\citeproc{ref-bridy2012coding}{Bridy 2012};
\citeproc{ref-lemley2023generative}{Lemley 2023}). This doctrinal stance
has been formalized in recent administrative and judicial decisions. The
U.S. Copyright Office has declined to extend copyright protection to
AI-generated images on the grounds that they are not the product of
human creative control,\footnote{\emph{See} Letter from U.S. Copyright
  Office to Van Lindberg, Counsel for Kristina Kashtanova, re: Zarya of
  the Dawn (Registration \# VAu001480196) (Feb.~21, 2023) (canceling
  original registration and reissuing certificate excluding AI-generated
  images).} and federal courts have affirmed that works generated solely
by AI are ineligible for protection.\footnote{\emph{See} Thaler v.
  Perlmutter, No.~23-5233 (D.C. Cir. Mar.~18, 2025).} Similarly, patent
law now requires that a natural person provide a ``significant
contribution'' to the conception of the invention, effectively treating
the AI's output as non-distinctive absent human intervention.\footnote{\emph{See}
  Inventorship Guidance for AI-Assisted Inventions, 89
  \textsc{Fed. Reg.} 10,043 (Feb.~13, 2024).}

The theoretical underpinning of this legal exclusion is the
regurgitation critique. This view posits that LLMs and diffusion models
merely stitch together patterns observed in training data without
genuine understanding or intent (\citeproc{ref-bender2021dangers}{Bender
et al. 2021}). Viewed through this lens, AI outputs are ``functionally
derivative''---operational recombinations of prior art that rely on
statistical correlations rather than original expression. This theory is
central to current infringement litigation, where plaintiffs allege that
AI models contain ``compressed copies'' of their training corpora and
``regurgitate'' memorized content.\footnote{\emph{See, e.g.}, Complaint,
  \emph{supra} note 21, at ¶¶ 105--07; \emph{Andersen}, 2023 WL 7132064,
  at *7--8.} If the regurgitation critique holds, AI outputs should be
statistically indistinguishable from the distribution of works upon
which they were trained, representing no meaningful departure from the
prior art.

However, a competing theoretical perspective suggests that AI generative
processes are fundamentally interpolative, potentially yielding
``probabilistic novelty.'' Because generative models map
high-dimensional latent spaces to output spaces, they do not simply
retrieve existing data points but interpolate between them, often
building internal representations of the underlying concepts
(\citeproc{ref-bubeck2023sparks}{Bubeck et al. 2023}). Consequently, the
resulting outputs are almost always structurally distinct from any
single training example. Indeed, the phenomenon of ``hallucination'' or
``confabulation''---where models generate plausible but non-factual
content---serves as evidence that the system is diverging from its
training distribution rather than merely reproducing it
(\citeproc{ref-Ji2023}{Ji et al. 2023};
\citeproc{ref-mukherjee2023managing}{Mukherjee and Chang 2023}). This
argument has been deployed in defense of AI training, with developers
arguing that because models learn statistical relationships to generate
new expression, the outputs are transformative rather than
derivative.\footnote{\emph{See} Defendant Meta Platforms, Inc.'s Motion
  to Dismiss, Kadrey v. Meta Platforms, Inc., No.~3:23-cv-03417 (N.D.
  Cal. Sept.~18, 2023), ECF No.~23.}

Empirical attempts to resolve this tension have reached an impasse. One
strand of research supports the regurgitation view, documenting
instances where models memorize and reproduce training data verbatim,
particularly when prompted with specific prefixes
(\citeproc{ref-chang2023speak}{Chang et al. 2023};
\citeproc{ref-carlini2023extracting}{Nasr et al. 2023};
\citeproc{ref-diakopoulos2023memorized}{Diakopoulos 2023}). Conversely,
other studies utilizing semantic analysis suggest that AI systems can
achieve high degrees of structural novelty and systematic generalization
that go beyond mere memorization (\citeproc{ref-lake2023human}{Lake and
Baroni 2023}; \citeproc{ref-mccoy2023much}{McCoy et al. 2023};
\citeproc{ref-lin2024evaluating}{Lin et al. 2024}). This apparent
contradiction arises in part from methodological limitations: existing
studies typically rely on \emph{item-level} metrics---checking for exact
string matches or pairwise similarity between specific outputs and
specific training examples. These metrics cannot adjudicate the
regurgitation-versus-interpolation debate because they fail to capture
the \emph{distributional} nature of the claim: if a model simply
regurgitates, its output distribution should converge to the training
distribution. If it interpolates, the distribution should shift.
Item-level analysis misses this distinction entirely, as it conflates
the occasional novel generation of an output similar to a training
example with the regurgitative replication of the training example.

\subsection{The Case for a Distributional
Framework}\label{the-case-for-a-distributional-framework}

The foregoing review reveals a tripartite challenge. Academic metrics
face a scaling crisis: verbatim checks are too narrow to capture
stylistic mimicry, while pairwise semantic comparisons succumb to the
``infinite cardinality'' of generative output spaces. Legal tests face
an adjudicability crisis: doctrines relying on human proxies (the
PHOSITA or ordinary observer) are strained as AI outputs become
increasingly difficult for humans to distinguish from human work, yet
legal institutions lack the resources to audit massive training sets.
Theoretical debates face an empirical impasse: the regurgitation
critique cannot be adjudicated by item-level analysis, as even a pure
regurgitator can occasionally produce a unique sentence by chance.
Resolving these challenges requires a metric that is
\emph{distributional} (to solve the scaling problem), \emph{semantically
aware} (to capture meaning over form), and \emph{sample-efficient} (to
be usable in court).

We propose a framework based on MMD as the specific quantitative remedy.
MMD evaluates samples collectively to determine if they are drawn from
the same underlying distribution, shifting the unit of analysis from the
\emph{item} to the \emph{process}. By pairing MMD with semantic
embeddings and permutation-based hypothesis testing, we convert raw
distributional distances into \emph{p}-values, allowing fact-finders to
assess---at a chosen standard of proof---whether two creative processes
are statistically distinct.

This formulation clarifies a critical limitation. Formally, we can view
the generative output distribution \(Q\) as a mixture:
\(Q = (1-\epsilon)Q_{novel} + \epsilon P_{training}\), where
\(\epsilon\) represents the rate at which the model regurgitates
training examples rather than generating novel outputs. If \(\epsilon\)
is small, the MMD between \(Q\) and \(P\) will remain high (indicating
distributional distinctiveness), effectively masking the infringing
tail. Thus, MMD measures the dominant creative mode (\(1-\epsilon\)),
while memorization audits are required to detect and quantify the
regurgitative component (\(\epsilon\)). This mixture formulation
explains why both tests are necessary: MMD evaluates whether
\(Q \neq P\) (process distinctiveness), while memorization audits detect
the presence and magnitude of the \(\epsilon P_{training}\) component
(item-level infringement).

We position this test not as a replacement for memorization audits, but
as a necessary counterpart. Because a generative process may be
distributionally distinct yet still produce rare instances of verbatim
regurgitation, we advocate for a bifurcated inquiry: using MMD to assess
process-level distinctiveness, while employing nearest-neighbor audits
to detect specific outliers or ``needles in the haystack.'' The next
section develops the mathematical foundations of our approach.

\section{Method Development}\label{method-development}

To resolve the unit-of-analysis mismatch identified in Section 2, we
must re-conceptualize the object of inquiry. Rather than focusing on
individual items, we represent creative sources---whether a human artist
or a generative AI model---as ``creative processes,'' treating specific
works as samples drawn from the high-dimensional probability
distribution characterizing the process. For instance, the corpus of
William Shakespeare's works can be understood as realizations of a
specific stochastic process, providing an empirical outline of the
random variable describing his creativity.

``Distinctiveness,'' in turn, is not a measure of the geometric distance
between two specific works (such as a new work and a specific prior art
reference). Instead, it is the statistical divergence between the
distributions of the corresponding stochastic processes, \(P\) and
\(Q\). If the processes are identical, \(P = Q\); if they are distinct,
\(P \neq Q\).

This conceptual shift accounts for the probabilistic nature of
creativity. As the ``infinite monkey theorem'' suggests, distinct
processes may occasionally produce the same output: given infinite time,
the creative process of Shakespeare and that of a monkey on a typewriter
might both output the same sonnet. However, such overlaps become
increasingly unlikely as the processes diverge---on average and in
finite time, we would expect a monkey to only produce output that vastly
differs from a sonnet. Conversely, if an AI model regurgitates
Shakespeare, we would expect high-probability regions of their
distributions to overlap, and for the AI's outputs to more frequently
and more closely resemble Shakespeare's works. Thus, given finite
samples of works from two processes, it is the aggregate pattern of
pairwise distances---and not any incidental pairwise distances---that
reveals distinctiveness; systemic distinctiveness manifests at the
distributional level and not at the item level.

To operationalize this inquiry, we propose a statistical framework based
on KMEs (for detailed technical derivations and properties, see
\citeproc{ref-gretton2012kernel}{Gretton et al. 2012};
\citeproc{ref-muandet2017kernel}{Muandet et al. 2017})\footnote{The
  mathematics underlying KMEs is complex. We provide a discussion
  tailored to our specific use; additional details can be found in the
  referenced works, with an exhaustive presentation in Berlinet and
  Thomas-Agnan (\citeproc{ref-berlinet2011reproducing}{2004}).}, MMD,
and semantic embeddings. The methodology integrates two complementary
strands of research on embeddings. The first discusses more abstract
notions of embeddings and establishes formal properties useful for
theoretical analysis
(\citeproc{ref-sriperumbudur2010hilbert}{Sriperumbudur et al. 2010}).
The second develops effective semantic embeddings for non-numerical
objects, such as text and images
(\citeproc{ref-mikolov2013efficient}{Mikolov et al. 2013}). We combine
these approaches to create a unified framework for distributional
distinctiveness analysis.

Our approach comprises three steps. First, we employ a semantic
embedding to map non-numerical data (such as prior art and AI-generated
images) into a numerical vector space where distances reflect semantic
relationships (\citeproc{ref-stammbach2021docscan}{Stammbach and Ash
2021}). Second, we use the vector representations to construct KMEs,
which map the probability distributions of the works into a reproducing
kernel Hilbert space (RKHS) of functions. Third, we compute the MMD
between these KMEs---a type of integral probability metric (IPM)---to
quantify the statistical distance between the creative processes
themselves.

\subsection{Definitions and
Background}\label{definitions-and-background}

Let \(X = \{x_1, x_2, \dots, x_m\}\) be a sample of works drawn from a
process with unknown probability distribution \(P\), and
\(Y = \{y_1, y_2, \dots, y_n\}\) be a sample of works from another
process with unknown probability distribution \(Q\). Our goal is to test
the null hypothesis \(H_0: P = Q\) (the distributions of the processes
are identical) against the alternative hypothesis \(H_1: P \neq Q\) (the
distributions differ).

In the context of intellectual property adjudication, these variables
map directly to the evidentiary record. For instance, \(P\) may
represent the distribution of a relevant Reference Class (e.g., the
corpus of prior art in patent, or the market of registered marks in
trademark), while \(Q\) may represent the Candidate Process (e.g., a
specific AI model). Accordingly, \(X\) and \(Y\) constitute the observed
portfolios presented to the fact-finder. The kernel function \(k\),
defined below, acts as the formal proxy for the Ordinary Observer or
PHOSITA, providing a consistent metric for the semantic similarity
between any two individual works.

An RKHS \(\mathcal{H}\) is a Hilbert space of functions defined by a
positive definite kernel function
\(k: \mathcal{X} \times \mathcal{X} \rightarrow \mathbb{R}\), where
\(\mathcal{X}\) is the input space (e.g., the space of possible creative
outputs). The RKHS is distinguished by the reproducing property: for
every function \(f\) in the RKHS and every point \(x \in \mathcal{X}\),
the value of \(f\) at \(x\), \(f(x)\), is reproduced by the inner
product of \(f\) with the kernel evaluation function, which is the
kernel function centered at \(x\), \(k(\cdot, x)\): \[
f(x) = \langle f, k(\cdot, x) \rangle_{\mathcal{H}}.
\] \(k(x, \cdot)\) denotes the kernel evaluation function: a function in
the RKHS defined by fixing one argument of the kernel,
\(y \mapsto k(x, y)\). The kernel function provides a way to ``probe''
the function \(f\) at any point \(x\) through the inner product. It
allows us to represent high-dimensional or even infinite-dimensional
feature spaces implicitly, which is a cornerstone of kernel methods in
machine learning (\citeproc{ref-shawe2004kernel}{Shawe-Taylor and
Cristianini 2004}; \citeproc{ref-steinwart2008support}{Steinwart and
Christmann 2008}).

KME leverages the machinery of RKHS to embed probability distributions
into a Hilbert space. Given a probability distribution \(P\) and a
reproducing kernel \(k\) that induces the RKHS \(\mathcal{H}\), the KME
of \(P\) into \(\mathcal{H}\), denoted \(\mu_P\), is the expected value
of the kernel evaluation function over \(P\): \[
\mu_P = \mathbb{E}_{x \sim P}[k(x, \cdot)] = \int_{\mathcal{X}} k(x, \cdot) \, dP(x),
\] where the integral is a Bochner integral.

A KME maps \(P\) to \(\mu_P\), a function in the RKHS \(\mathcal{H}\).
If the kernel \(k\) is \emph{characteristic}, then this mapping is
\emph{injective} (one-to-one). That is, if two probability distributions
differ, their KMEs also differ (\(P \neq Q \implies \mu_P \neq \mu_Q\)).
Equivalently, if the mean embeddings are identical, the underlying
distributions must be identical (\(\mu_P = \mu_Q \implies P = Q\)). This
builds on the notion that a kernel function measures the similarity
between two points in the input space; if the kernel is sufficiently
`rich' (formally, characteristic), the distance between two KMEs
strictly corresponds to the distance between the distributions
themselves.

An IPM between distributions \(P\) and \(Q\) is defined as: \[
\text{IPM}(P, Q) = \sup_{f \in \mathcal{F}} \left| \int_{\mathcal{X}} f(x) \, dP(x) - \int_{\mathcal{X}} f(x) \, dQ(x) \right|,
\] where \(\mathcal{F}\) is a class of functions.

MMD is a type of IPM where the class of functions \(\mathcal{F}\) is the
unit ball in the RKHS. The MMD quantifies the distance between \(P\) and
\(Q\) as the distance between their respective KMEs in the RKHS
(\citeproc{ref-gretton2012kernel}{Gretton et al. 2012}): \[
\text{MMD}^2(P, Q) = \| \mu_P - \mu_Q \|_{\mathcal{H}}^2,
\] where \(\|\cdot\|_{\mathcal{H}}\) denotes the norm in the RKHS.

\subsection{Employing MMD to Measure Distributional
Distinctiveness}\label{employing-mmd-to-measure-distributional-distinctiveness}

Suppose two creative processes produce only numerical data. Given
samples \(X\) and \(Y\) from their distributions \(P\) and \(Q\)
respectively, we can compute an \emph{unbiased} empirical estimator of
\(\text{MMD}^2\): \begin{align}
\widehat{\text{MMD}}^2_u(X, Y) = \frac{1}{m(m-1)} \sum_{i=1}^{m} \sum_{\substack{j=1 \\ j \neq i}}^{m} k(x_i, x_j) + \frac{1}{n(n-1)} \sum_{i=1}^{n} \sum_{\substack{j=1 \\ j \neq i}}^{n} k(y_i, y_j) - \frac{2}{mn} \sum_{i=1}^{m} \sum_{j=1}^{n} k(x_i, y_j). \label{eqn:numerical_MMD}
\end{align}

This estimator can be computed efficiently using the \emph{kernel
trick}, avoiding explicit computation of the feature maps
\(k(\cdot, x)\).\footnote{While the population \(\text{MMD}^2\) is
  strictly non-negative, this unbiased estimator can yield negative
  values in finite samples if the cross-sample similarity exceeds the
  within-sample similarity due to sampling noise.} The components of
this equation are interpreted as follows:

\begin{itemize}
\tightlist
\item
  \(k(\cdot, \cdot)\) is the kernel function.
\item
  \(x_i\) and \(x_j\) are samples from distribution \(P\).
\item
  \(y_i\) and \(y_j\) are samples from distribution \(Q\).
\item
  \(m\) and \(n\) are the sizes of the samples from \(P\) and \(Q\),
  respectively.
\item
  The first term,
  \(\frac{1}{m(m-1)} \sum_{i=1}^{m} \sum_{\substack{j=1 \\ j \neq i}}^{m} k(x_i, x_j)\),
  is the average of the kernel evaluations over all distinct ordered
  pairs of samples from \(P\).
\item
  The second term,
  \(\frac{1}{n(n-1)} \sum_{i=1}^{n} \sum_{\substack{j=1 \\ j \neq i}}^{n} k(y_i, y_j)\),
  is the average of the kernel evaluations over all distinct ordered
  pairs of samples from \(Q\).
\item
  The third term,
  \(-\frac{2}{mn} \sum_{i=1}^{m} \sum_{j=1}^{n} k(x_i, y_j)\), subtracts
  twice the average of the kernel evaluations between samples from \(P\)
  and samples from \(Q\).
\end{itemize}

\(\text{MMD}^2\) can be used to establish the distributional
distinctiveness of \(P\) and \(Q\). Most creative processes, however, do
not generate numerical data. Therefore, to apply this framework to
non-numerical data (e.g., text, images), we propose mapping all works
into a numerical vector space using a machine learning embedding.

Let \(\phi_x: \mathcal{X} \rightarrow \mathcal{Z}\) represent such an
embedding, where \(\mathcal{Z}\) is often \(\mathbb{R}^d\), \(d\) being
the dimensionality of the embedding space. The choice of embedding
depends on the specific data type (e.g., a text embedding for text data,
a convolutional neural network (CNN) embedding for images) and the
relevant semantic relationships between the data points (e.g.,
similarity in meaning for text, visual similarity for art).

We propose the compositional structure: \[
\phi(x) = \phi_k(\phi_x(x)), \quad x \in \mathcal{X},
\] where:

\begin{itemize}
\tightlist
\item
  \(\phi_x(x)\) is the machine learning embedding of the `raw' data
  point \(x\).
\item
  \(\phi_k\) is the feature map defined implicitly by the kernel \(k\)
  in the RKHS.
\end{itemize}

Proposition \ref{prop:composed_kernel} establishes the theoretical
validity of this approach. It shows that the kernel defined by the
compositional structure \(k(\phi_x(x_i), \phi_x(x_j))\) is
characteristic if the machine learning embedding \(\phi_x\) is injective
and if the kernel \(k\) in the RKHS is characteristic. That is, if the
machine learning embedding \(\phi_x\) preserves the distinctness of
inputs (injectivity) and the kernel \(k\) is capable of distinguishing
distributions in the vector space (characteristic), then the combined
framework can reliably detect distributional differences in the original
creative works.

Therefore, we can apply our MMD framework to any data type for which a
suitable embedding \(\phi_x\) and a suitable kernel \(k\) can be found.
This estimator is computed by evaluating the kernel function on the
embedded data, replacing \(k(x_i, x_j)\) with
\(k(\phi_x(x_i), \phi_x(x_j))\) in Equation \ref{eqn:numerical_MMD}.

\begin{proposition} \label{prop:composed_kernel}
Let \(\phi_x: \mathcal{X} \to \mathcal{Z}\) be an injective mapping and \(k: \mathcal{Z} \times \mathcal{Z} \to \mathbb{R}\) be a characteristic kernel on \(\mathcal{Z}\). Then the composed kernel \(k_{\phi}: \mathcal{X} \times \mathcal{X} \to \mathbb{R}\), defined as \(k_{\phi}(x, x') = k(\phi_x(x), \phi_x(x'))\), is characteristic on \(\mathcal{X}\).
\end{proposition}

\begin{proof}
A kernel is characteristic if the map from probability measures to kernel mean embeddings is injective. That is, for any two probability measures \(P\) and \(Q\) on \(\mathcal{X}\), \(\|\mu_P - \mu_Q\|_{\mathcal{H}_\phi} = 0\) must imply \(P = Q\).

Let \(P_x\) and \(Q_x\) be probability measures on \(\mathcal{X}\). We define their pushforward measures on \(\mathcal{Z}\) as \(P_z = P \circ \phi_x^{-1}\) and \(Q_z = Q \circ \phi_x^{-1}\).

First, we establish that the mapping from measures on \(\mathcal{X}\) to measures on \(\mathcal{Z}\) is injective. By definition of the pushforward, \(P_z(B) = P_x(\phi_x^{-1}(B))\) for any measurable set \(B \subseteq \mathcal{Z}\). Applying this to \(B = \phi_x(A)\) for some \(A \subseteq \mathcal{X}\), we obtain \(P_z(\phi_x(A)) = P_x(\phi_x^{-1}(\phi_x(A)))\). Because \(\phi_x\) is injective, \(\phi_x^{-1}(\phi_x(A)) = A\), and thus \(P_z(\phi_x(A)) = P_x(A)\). If \(P_x \neq Q_x\), there exists a set \(A\) such that \(P_x(A) \neq Q_x(A)\), which implies \(P_z(\phi_x(A)) \neq Q_z(\phi_x(A))\), and thus \(P_z \neq Q_z\). Equivalently, \(P_z = Q_z\) implies \(P_x = Q_x\).

Next, we examine the distance between the kernel mean embeddings. The squared MMD distance in the RKHS \(\mathcal{H}_{\phi}\) associated with the composed kernel \(k_\phi\) is:
\[
\| \mu_{P_x} - \mu_{Q_x} \|_{\mathcal{H}_{\phi}}^2 = \mathbb{E}_{x,x'}[k_\phi(x,x')] - 2\mathbb{E}_{x,y}[k_\phi(x,y)] + \mathbb{E}_{y,y'}[k_\phi(y,y')],
\]
where \(x, x' \sim P_x\) and \(y, y' \sim Q_x\).

By substituting the definition \(k_\phi(a, b) = k(\phi_x(a), \phi_x(b))\) and applying the change of variables (where \(z = \phi_x(x) \sim P_z\) and \(w = \phi_x(y) \sim Q_z\)), this equates to:
\[
\mathbb{E}_{z,z'}[k(z,z')] - 2\mathbb{E}_{z,w}[k(z,w)] + \mathbb{E}_{w,w'}[k(w,w')] = \| \mu_{P_z} - \mu_{Q_z} \|_{\mathcal{H}_z}^2.
\]

This establishes an isometry between the embeddings in the two spaces:
\[
\| \mu_{P_x} - \mu_{Q_x} \|_{\mathcal{H}_{\phi}} = \| \mu_{P_z} - \mu_{Q_z} \|_{\mathcal{H}_z}.
\]

\(\| \mu_{P_x} - \mu_{Q_x} \|_{\mathcal{H}_{\phi}} = 0\) implies \(\| \mu_{P_z} - \mu_{Q_z} \|_{\mathcal{H}_z} = 0\), and therefore \(\mu_{P_z} = \mu_{Q_z}\). Since kernel \(k\) is characteristic on \(\mathcal{Z}\), the equality of mean embeddings \(\mu_{P_z} = \mu_{Q_z}\) implies the equality of distributions \(P_z = Q_z\). As established, the injectivity of \(\phi_x\) ensures that \(P_z = Q_z \implies P_x = Q_x\).

Therefore, \(\| \mu_{P_x} - \mu_{Q_x} \|_{\mathcal{H}_{\phi}} = 0\) implies \(P_x = Q_x\), proving that the mapping \(P_x \mapsto \mu_{P_x}\) is injective and the composed kernel \(k_{\phi}\) is characteristic on \(\mathcal{X}\).
\end{proof}

\subsection{Hypothesis Testing}\label{hypothesis-testing}

In forensic contexts, a raw distance metric is of limited evidentiary
value; a score of ``0.1'' or ``100'' is meaningless without a baseline
for stochastic variation. To render the MMD metric legally actionable,
we must determine whether an observed divergence exceeds what would be
expected by chance between two samples drawn from the same process.

We therefore frame the inquiry not as a measurement of magnitude, but as
a hypothesis test. This approach allows the fact-finder to explicitly
set the significance level (\(\alpha\))---the statistical analogue to
the legal standard of proof---thereby controlling the risk of a False
Positive (Type I error) where a process is incorrectly deemed
distinctive.

To test the null hypothesis \(H_0: P = Q\), we use the empirical
estimator \(\widehat{\text{MMD}}^2_u\) as our test statistic. Because
MMD is a measure of distance, the alternative hypothesis
\(H_1: P \neq Q\) implies a strictly positive deviation. Therefore, we
employ a one-sided (right-tailed) permutation test: we reject the null
hypothesis only if the observed statistic is significantly larger than
what would be expected by chance. Negative values of the unbiased
estimator, which arise solely from finite sampling noise, are consistent
with the null hypothesis and do not constitute evidence of
distinctiveness.

\begin{algorithm}[htbp]
\caption{Permutation-Based Hypothesis Test for MMD}
\label{alg:mmd_permutation_based}
\begin{algorithmic}[1]
\REQUIRE Samples \(X = \{x_1, \dots, x_m\}\) from distribution \(P\), samples \(Y = \{y_1, \dots, y_n\}\) from distribution \(Q\), kernel function \(k(\cdot,\cdot)\), number of permutation iterations \(R\), significance level \(\alpha\).
\STATE Compute:
    \[
    \Delta_{\text{obs}} \gets \widehat{\text{MMD}}^2_u(X, Y),
    \]
    \[
    \widehat{\text{MMD}}^2_u(X, Y) = \frac{1}{m(m-1)} \sum_{i=1}^{m}\sum_{\substack{j=1 \\ j\neq i}}^{m} k(x_i,x_j) + \frac{1}{n(n-1)} \sum_{i=1}^{n}\sum_{\substack{j=1 \\ j\neq i}}^{n} k(y_i,y_j) - \frac{2}{mn}\sum_{i=1}^{m}\sum_{j=1}^{n} k(x_i,y_j).
    \]
\STATE Pool samples into a single dataset of size \(m+n\):
    \[
    Z \gets X \cup Y.
    \]
\FOR{\(r = 1\) \TO \(R\)}
    \STATE Randomly permute the pooled sample \(Z\). Let the permuted sample be \(Z^*\).
    \STATE Partition \(Z^*\) into two sets: \(X^*_r\) containing the first \(m\) elements, and \(Y^*_r\) containing the remaining \(n\) elements.
    \STATE Compute the statistic on the permuted partition:
        \[
        \Delta^*_r \gets \widehat{\text{MMD}}^2_u(X^*_r, Y^*_r).
        \]
\ENDFOR
\STATE Calculate the *p*-value:
    \[
    p \gets \frac{1 + \sum_{r=1}^{R} \mathbf{1}\{\Delta^*_r \ge \Delta_{\text{obs}}\}}{R + 1},
    \]
    where
    \[
    \mathbf{1}\{A\} = \begin{cases}
    1 & \text{if } A \text{ is true}\\[6pt]
    0 & \text{otherwise}
    \end{cases}.
    \]
\STATE Determine the critical value \(c_\alpha\) from the permutation distribution:
    \[
    c_\alpha \gets Q_{1-\alpha}\left(\{\Delta^*_r\}_{r=1}^R\right),
    \]
    where \(Q_{\gamma}(\cdot)\) denotes the \(\gamma\)-quantile of the permutation-based statistics.
\IF{\(\Delta_{\text{obs}} > c_\alpha\) (or equivalently, if \(p < \alpha\))}
    \STATE Reject \(H_0: P=Q\).
\ELSE
    \STATE Do not reject \(H_0\).
\ENDIF
\end{algorithmic}
\end{algorithm}

Algorithm \ref{alg:mmd_permutation_based} constructs the null
distribution of the MMD statistic by repeatedly resampling the pooled
data. This process destroys any systematic distributional differences
between the two groups while preserving the marginal distribution of the
combined data. The ``\(+1\)'' terms in the numerator and denominator of
the \emph{p}-value formula (step 8) ensure exact Type I error control by
treating the observed statistic as an additional permutation under the
null hypothesis (\citeproc{ref-phipson2010permutation}{Phipson and Smyth
2010}). We calculate the critical value \(c_\alpha\) as the
\((1-\alpha)\)-quantile of these permutation statistics. If the observed
statistic \(\widehat{\text{MMD}}^2_u(X, Y)\) exceeds this threshold---or
equivalently, if the \emph{p}-value is less than the significance level
\(\alpha\)---we reject the null hypothesis and conclude that the
creative processes exhibit statistically significant distributional
distinctiveness.

\subsection{Stability of MMD Under Approximate
Embeddings}\label{stability-of-mmd-under-approximate-embeddings}

Proposition \ref{prop:composed_kernel} establishes that if the machine
learning embedding \(\phi_x\) is injective, the resulting MMD is a valid
metric for distributional distinctiveness. Dimensionality reduction
necessarily involves information loss, and one might worry that the
pigeonhole principle implies distinct inputs must collide when mapped to
a lower-dimensional space.

However, in the context of real-world data, the injectivity condition is
not a binding constraint. The pigeonhole principle requires more items
(inputs) than slots (outputs). In the context of semantic
distinctiveness, this condition is rarely met. Digital images occupy a
discrete subset of the input space: pixels take integer values
(typically 0--255), and the set of images constituting meaningful
creative works is a sparse manifold within that space. Meanwhile, the
embedding output space, though discrete at the floating-point level, has
vast capacity. For a \(d\)-dimensional embedding using standard
representations, the number of distinct representable vectors is
approximately \(C^d\), where \(C\) is the capacity of the floating-point
format. For a 1024-dimensional embedding, this yields a state space
exceeding \(10^{5000}\) distinct vectors---a number exceeding any
plausible enumeration of semantically distinct visual configurations.
Moreover, contrastive embeddings are trained precisely to preserve
distinctions that matter: two works map to similar vectors if and only
if they are semantically similar. Consequently, collisions would require
two perceptually distinct works to map to bit-identical vectors, an
event with negligible probability.

Nevertheless, a practical concern arises from \emph{approximation
error}. In forensic contexts, analysts may rely on off-the-shelf
foundation models or dimensionally reduced representations to ensure
computational tractability. These networks may not exploit their
complete theoretical capacity and may map distinct inputs to proximal
representations. It therefore becomes necessary to characterize how such
item-level approximation errors propagate to the process-level MMD
metric.

We treat approximate embeddings as inducing a bounded perturbation of
the kernel values. Proposition \ref{prop:mmd_stability} demonstrates
that the squared MMD is Lipschitz-continuous with respect to this
perturbation: if the embedding error is bounded, the resulting error in
the distinctiveness metric is also strictly bounded.

\begin{proposition}[Stability under Kernel Perturbations]
\label{prop:mmd_stability}
Let \(k_\phi\) be the "ideal" composed kernel derived from a high-fidelity embedding, and \(k_d\) be an approximate kernel derived from a lower-dimensional or approximate embedding. Suppose the approximation error in the kernel evaluation is uniformly bounded by \(\varepsilon\), such that:
\[
\sup_{x,x' \in \mathcal{X}} \bigl|k_d(x,x') - k_\phi(x,x')\bigr| \le \varepsilon.
\]
Then for any probability measures \(P\) and \(Q\), the absolute error in the squared MMD is bounded by \(4\varepsilon\):
\[
\bigl|\mathrm{MMD}^2_{d}(P,Q) - \mathrm{MMD}^2_{\phi}(P,Q)\bigr| \le 4\varepsilon.
\]
\end{proposition}

\begin{proof}
Recall the expansion of the squared MMD:
\[
\mathrm{MMD}^2(P,Q) = \mathbb{E}_{x,x'}[k(x,x')] - 2\mathbb{E}_{x,y}[k(x,y)] + \mathbb{E}_{y,y'}[k(y,y')].
\]
Let \(\delta(x,x') = k_d(x,x') - k_\phi(x,x')\). The difference between the two MMD estimators is:
\[
\Delta_{\text{err}} = \mathbb{E}_{x,x'}[\delta(x,x')] - 2\mathbb{E}_{x,y}[\delta(x,y)] + \mathbb{E}_{y,y'}[\delta(y,y')].
\]
Applying the triangle inequality and the uniform bound \(|\delta(\cdot, \cdot)| \le \varepsilon\):
\[
|\Delta_{\text{err}}| \le \bigl|\mathbb{E}[\delta(x,x')]\bigr| + 2\bigl|\mathbb{E}[\delta(x,y)]\bigr| + \bigl|\mathbb{E}[\delta(y,y')]\bigr| \le \varepsilon + 2\varepsilon + \varepsilon = 4\varepsilon.
\]
The factor of four arises because the MMD expansion involves four expectation terms—two within-group terms and one cross-group term (counted twice)—each of which can accumulate error up to \(\varepsilon\).
\end{proof}

This result has specific implications for the Gaussian RBF kernel
employed in our empirical analysis, defined as
\(k(z,z') = \exp(-\|z-z'\|^2 / 2\sigma^2)\). The function
\(f(u) = e^{-u}\) is 1-Lipschitz on \([0, \infty)\). Therefore, if an
approximate embedding distorts the squared Euclidean distance between
any two works by at most \(\eta\), the kernel value is distorted by at
most \(\eta / 2\sigma^2\).

Combining this with Proposition \ref{prop:mmd_stability}, we obtain a
specific bound for Gaussian kernels: \[
\bigl|\mathrm{MMD}^2_{d}(P,Q) - \mathrm{MMD}^2_{\phi}(P,Q)\bigr| \le \frac{2\eta}{\sigma^2}.
\]

This provides a ``safety guarantee'' for the metric. For kernels bounded
in \([0,1]\) (such as the Gaussian RBF), the squared MMD always lies in
the range \([0,2]\). If the embedding approximation changes pairwise
similarities by at most 1\% (\(\varepsilon=0.01\)), the final
distinctiveness score changes by at most 0.04. This ensures the method
remains robust even when operating under computational constraints.

Moreover, it is important to note that the permutation-based hypothesis
test described in Algorithm \ref{alg:mmd_permutation_based} retains
exact Type I error control regardless of the embedding employed---a
critical feature for admissibility. Under the null hypothesis
\(H_0: P = Q\), the permutation distribution is constructed from the
same kernel \(k_d\) used for the test statistic, ensuring that the
critical value \(c_\alpha\) correctly bounds the false positive rate at
level \(\alpha\). The primary consequence of embedding approximation is
a potential loss of statistical power: if \(k_d\) introduces systematic
distortions, the test may require larger samples to detect true
distributional differences. The bounds above quantify this trade-off,
allowing practitioners to select embeddings that balance computational
efficiency against measurement precision.

Beyond computational tractability, there are substantive reasons to
consider dimensionality reduction when applying MMD to machine learning
embeddings. Conventional embeddings from foundation models---whether
text transformers producing 384- or 768-dimensional vectors or vision
models producing 1024-dimensional representations---present two
practical challenges. First, these embeddings can be sensitive to input
noise: minor perturbations in the original data (compression artifacts,
watermarks, or sensor noise) may induce disproportionate shifts in the
high-dimensional representation. Second, in very high-dimensional
spaces, data points become increasingly sparse---the well-known ``curse
of dimensionality''---which can degrade the reliability of
distance-based statistics. Dimensionality reduction algorithms such as
UMAP (\citeproc{ref-mcinnes2018umap}{McInnes et al. 2018}) address both
concerns: by projecting embeddings onto a lower-dimensional manifold,
they can attenuate the influence of noise while increasing the effective
density of observations in the reduced space. In our empirical analyses,
we apply UMAP reduction (to 64 dimensions) to the raw embeddings before
computing MMD, finding that this improves robustness to perturbations
while preserving---and in some cases enhancing---the test's ability to
detect genuine distributional differences.

\section{Validation: MNIST Handwritten
Digits}\label{validation-mnist-handwritten-digits}

Before applying our methodology to real-world data, we validate its
statistical properties and practical utility in a controlled setting
with known ground truth. We use the MNIST dataset of handwritten digits
(\citeproc{ref-lecun1998gradient}{LeCun et al. 1998}), a widely
recognized benchmark in machine learning. MNIST comprises 70,000
grayscale images (28×28 pixels) of handwritten digits from 0 to 9, split
into a training set of 60,000 images and a test set of 10,000 images
(containing approximately 1000 examples per digit class). Each image
represents a single digit, providing a convenient ground truth: it is
reasonable to treat the class-conditional distributions of the images of
different digits as distinct.

To represent the images in a vector space suitable for MMD calculation,
we employ a convolutional neural network (CNN) embedding designed for
handwritten digit recognition. We use a LeNet-5-style architecture
(\citeproc{ref-lecun1998gradient}{LeCun et al. 1998}), adapted to use
ReLU activations, that consists of two convolutional layers with average
pooling, followed by three fully connected (dense) layers; we include
dropout regularization (rate = 0.1) after the 120-unit layer and after
the 84-unit embedding layer. The architecture is summarized in Table
\ref{tab:lenet5}.

\begin{table}[htbp]
\centering
\begin{tabular}{lll}
\toprule
Layer Type                          & Output Shape & Parameters \\
\midrule
Input                               & 28×28×1      & 0          \\
Conv2D (6 filters, 5×5 kernel, ReLU)& 24×24×6      & 156        \\
AvgPool2D (2×2)                     & 12×12×6      & 0          \\
Conv2D (16 filters, 5×5 kernel, ReLU)& 8×8×16      & 2,416      \\
AvgPool2D (2×2)                     & 4×4×16       & 0          \\
Flatten                             & 256          & 0          \\
Dense (120 units, ReLU)             & 120          & 30,840     \\
Dropout (rate = 0.1)                & 120          & 0          \\
Dense (84 units, ReLU; embedding)   & 84           & 10,164     \\
Dropout (rate = 0.1)                & 84           & 0          \\
Dense (10 units, Softmax)           & 10           & 850        \\
\bottomrule
\end{tabular}
\caption{LeNet-5-style Architecture Details}
\label{tab:lenet5}
\end{table}

We train the model on the MNIST training set using the Adam optimizer
(learning rate = 0.001), categorical cross-entropy loss, and a batch
size of 64. To ensure the learned feature space is robust to input
perturbations, we apply data augmentation during training, including
random affine transformations and additive Gaussian noise
(\(\sigma \sim \text{Uniform}(0, 1.0)\)). Training employs early
stopping (patience = 10 epochs) and model checkpointing (saving the best
model based on validation loss calculated on a 10\% hold-out split). Our
final trained model achieves a test accuracy of 99.1\%, indicating that
the learned embeddings effectively capture the distinguishing visual
features of each digit while remaining relatively invariant to noise.

\subsection{MMD Analysis Procedure and
Setup}\label{mmd-analysis-procedure-and-setup}

Our validation procedure comprises the following steps:

\begin{enumerate}
\def\labelenumi{\arabic{enumi}.}
\item
  \textbf{Embedding Extraction:} We process all MNIST test images,
  extracting 84-dimensional embeddings from the penultimate dense layer
  of the trained model and then applying UMAP reduction to 64
  dimensions. These embeddings represent each digit as a numerical
  vector capturing the high-level visual features identified by the
  network.
\item
  \textbf{Sample Generation:} We perform comparisons for all
  \(10 \times 10\) combinations of digit classes. For each comparison,
  as a negative control, we compare two \emph{disjoint} random
  subsamples drawn from the \emph{same} digit class (e.g., digit 3 vs.~a
  different set of digit 3s). Here, we expect the MMD statistic to be
  near zero and the null hypothesis not to be rejected, providing a
  baseline for evaluating the method's Type I error rate. As a positive
  control, we compare two \emph{disjoint} random subsamples drawn from
  \emph{different} digit classes. To ensure balanced comparisons and
  maintain computational feasibility for the permutation-based
  hypothesis testing, we cap sample sizes at 500 embeddings for each
  distribution.
\item
  \textbf{MMD Calculation and Hypothesis Testing:} For each pair and
  sample size, we compute the unbiased squared MMD statistic using a
  Gaussian radial basis function (RBF) kernel.\footnote{Throughout the
    empirical sections, we report values of the unbiased squared MMD
    estimator (\(\widehat{\text{MMD}}^2_u\); see Equation
    \ref{eqn:numerical_MMD}). For brevity, we denote this quantity as
    \(\text{MMD}^2\) when reporting numerical results.} We opt for the
  Gaussian RBF kernel as it is \emph{characteristic}, ensuring that the
  test statistic is a proper metric where a distance of zero implies
  identical distributions (\citeproc{ref-gretton2012kernel}{Gretton et
  al. 2012}). For the bandwidth parameter (\(\sigma\)), we implement the
  median heuristic, setting \(\sigma\) to the median of all pairwise
  Euclidean distances in the combined sample for that specific
  comparison. This data-adaptive approach scales the kernel
  appropriately to the data's dimensionality; we test the sensitivity of
  our conclusions to this heuristic in ablation studies. We perform the
  permutation-based hypothesis test described in Algorithm 1, with
  \(R=500\) permutation iterations and a significance level of
  \(\alpha=0.01\).
\item
  \textbf{Sample Size Variation and Rejection Rate Estimation:} To
  evaluate the method's sensitivity and data efficiency, we repeat steps
  2 and 3 across smaller sample sizes: 3, 4, 5, 6, 7, 8, 9, and 10. For
  each digit pair and sample size, we perform 500 independent trials,
  each involving fresh random sampling and a full permutation test.
  Averaging the outcomes (reject/fail-to-reject \(H_0\)) across these
  500 trials provides a robust estimate of the rejection rate (empirical
  statistical power) for that scenario. We focus on a representative
  subset of digit pairs that vary in visual similarity---(0 vs.~1), (1
  vs.~7), (2 vs.~8), (3 vs.~5), and (4 vs.~9)---to evaluate performance
  across both easy and challenging comparisons. This systematic
  exploration helps characterize the minimum data requirements for
  reliable distribution discrimination.
\end{enumerate}

\subsection{Core Results}\label{core-results}

Figure \ref{fig:mnist_rejection_rate} illustrates the sensitivity of our
approach, describing the estimated rejection rate of the null hypothesis
(\(H_0: P=Q\)) at a significance level of \(\alpha=0.01\) as the sample
size per distribution increases. Our results establish high data
efficiency. For both visually distinct examples (e.g., 0 vs.~1, 1 vs.~7)
and those exhibiting greater visual similarity (e.g., 3 vs.~5, 4 vs.~9),
the rejection rate rapidly surpasses the 95\% threshold at just \(n=6\)
samples per distribution. As the sample size increases from \(n=3\) to
\(n=10\), the rejection rates consistently approach 100\% for all tested
pairs, confirming the method's statistical convergence.

\begin{figure}[htbp]
\centering
\includegraphics[width=\linewidth]{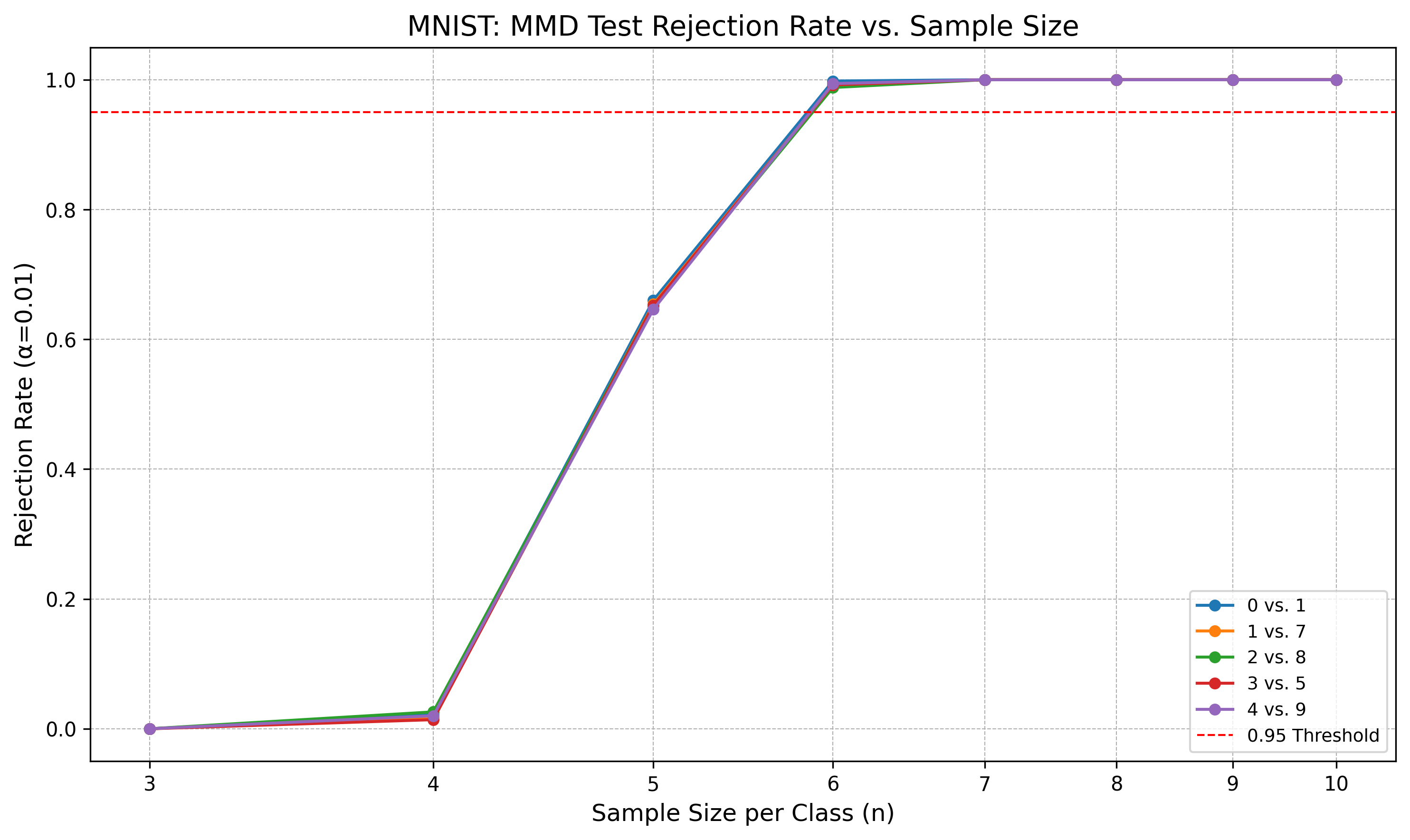}
\caption{Rejection Rate vs. Sample Size for Selected MNIST Digit Pairs}
\label{fig:mnist_rejection_rate}
\begin{minipage}{0.95\linewidth}
\footnotesize
Note: Each line represents the proportion of null hypothesis rejections (\(H_0: P=Q\)) at \(\alpha=0.01\), estimated by averaging results over 500 independent random sampling trials for each sample size and digit pair. The dashed line at 0.95 highlights rapid achievement of high statistical power with very small sample sizes (n=6 for the representative pairs shown).
\end{minipage}
\end{figure}

Figure \ref{fig:mnist_mmd_heatmap} depicts unbiased squared MMD
statistics (\(\widehat{\text{MMD}}^2_u\)) across all digit comparisons
at a sample size of \(n=500\). Diagonal comparisons (negative controls,
comparing disjoint samples of the same digit) yield statistics reliably
close to zero, consistent with the behavior of the unbiased estimator
under the null. These comparisons yield no rejections of the null
hypothesis. In contrast, we reject the null in all 90 off-diagonal
comparisons (representing the 45 distinct unordered digit pairs) for
\(p < 0.01\).

\begin{figure}[htbp]
\centering
\includegraphics[width=\linewidth]{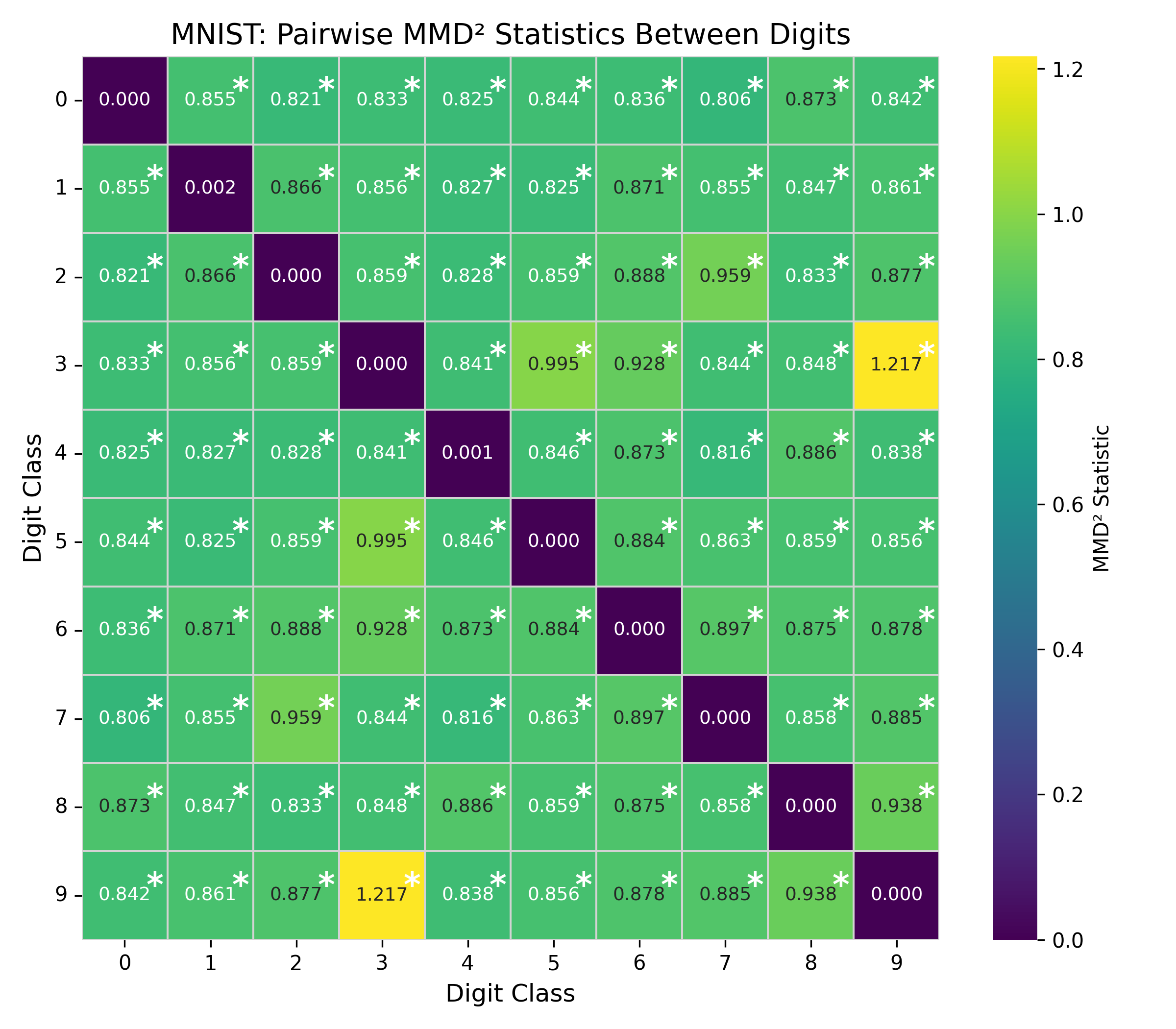}
\caption{Heatmap of Squared MMD Statistics for All MNIST Digit Pairs (Sample Size \(n=500\)).}
\label{fig:mnist_mmd_heatmap}
\begin{minipage}{0.95\linewidth}
\footnotesize
Note: Diagonal cells (negative controls) show near-zero, non-significant MMD values. All off-diagonal cells show statistically significant differences (p \(< 0.01\), marked with *), with MMD magnitudes reflecting the degree of distributional dissimilarity.
\end{minipage}
\end{figure}

This alignment between the quantitative MMD measure and visual
dissimilarity indicates that our methodology can reliably and
efficiently distinguish between different digit distributions, achieving
statistically significant differentiation (\(p < 0.01\)) with as few as
6 samples per digit class. This level of data efficiency is likely to be
particularly valuable in contexts where comprehensive datasets may be
unavailable---such as when evaluating the novelty of a small set of
AI-generated works or comparing a new trademark to a limited set of
existing marks.

\subsection{Metric Stability}\label{metric-stability}

While the rejection rates presented above confirm the method's utility
in a binary decision context (i.e., distinguishing distinct
distributions), forensic applications also require stability in the
\emph{magnitude} of the measured distance. If MMD is to serve as a
quantitative proxy for distinctiveness, its value should change
predictably---rather than erratically---under controlled perturbations
to either the embedding representation or the input data itself.

Proposition \ref{prop:mmd_stability} establishes that the squared MMD is
Lipschitz-continuous with respect to bounded kernel perturbations,
predicting that the distance estimate should degrade gracefully as
information is lost or noise is introduced. We empirically validate this
prediction through two complementary analyses: dimensionality reduction
(compressing the embedding) and input perturbation (adding noise or
watermarks to the raw images).

\subsubsection{Dimensionality Reduction}\label{dimensionality-reduction}

To evaluate how embedding compression affects the MMD metric, we conduct
a separate analysis comparing the full 84-dimensional LeNet embeddings
to UMAP reductions across a grid of target dimensions
\(d \in \{2, 5, 10, 20, 32, 50, 64, 70, 84\}\). For each of 100
independent trials, we draw a fixed pair of samples (\(n=50\) per group)
and compute the MMD in both the full-dimensional space (\(D=84\)) and
the reduced space (\(d\)) using the same observations. At each \(d\), we
recompute the kernel bandwidth \(\sigma\) using the median heuristic. We
then record the absolute deviation:
\(\Delta(d) = | \widehat{\text{MMD}}^2_{84}(X, Y) - \widehat{\text{MMD}}^2_{d}(X, Y) |\).

\begin{figure}[htbp]
\centering
\includegraphics[width=\linewidth]{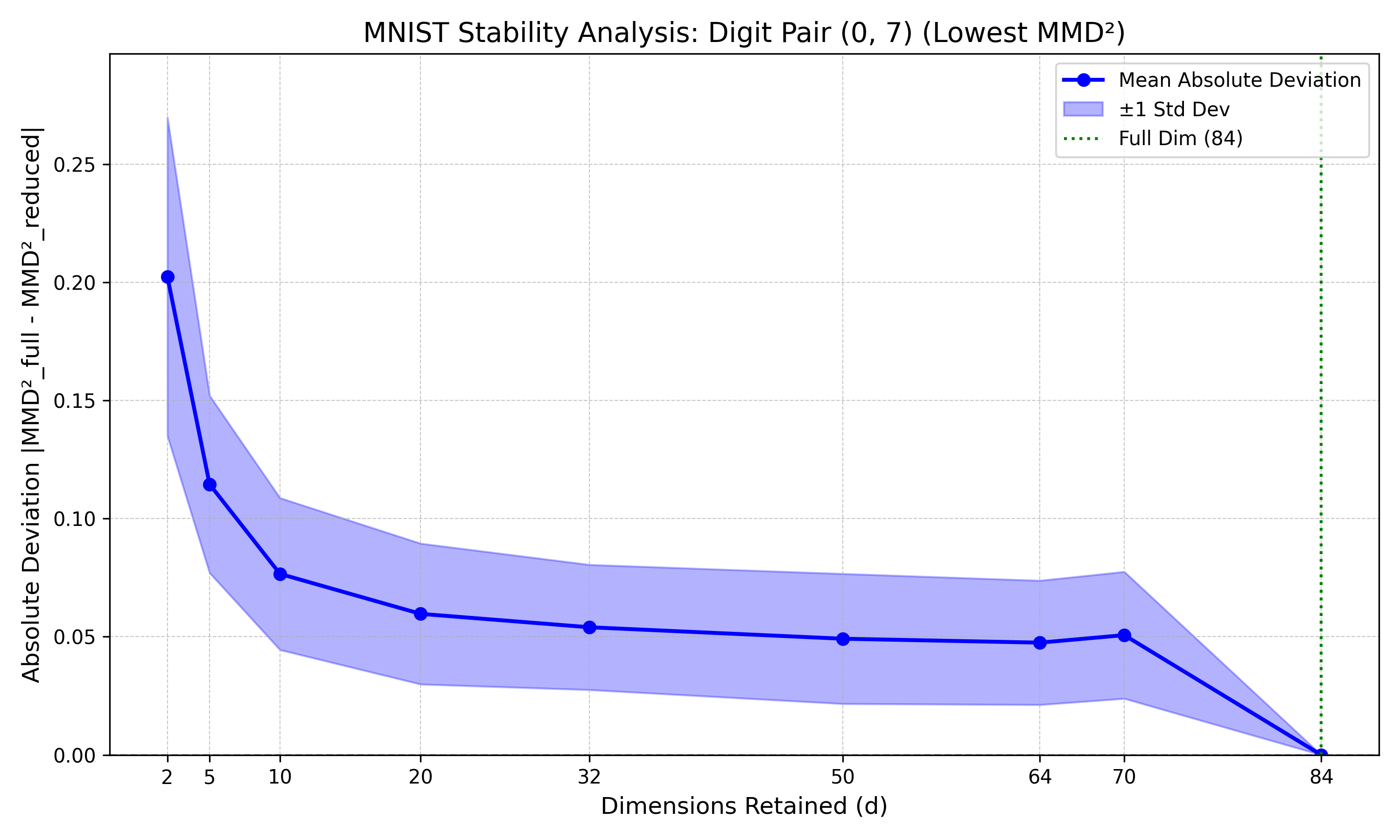}
\caption{Geometric Stability: Absolute Deviation in MMD vs. Retained Dimensions}
\label{fig:mnist_stability}
\begin{minipage}{0.95\linewidth}
\footnotesize
Note: Mean absolute deviation (with standard deviation error bars) between full-dimensional MMD and reduced-dimensional MMD, computed over 100 independent trials. Consistent with Proposition 2, approximation error decays rapidly as dimensions are added.
\end{minipage}
\end{figure}

Consistent with theoretical expectations, the approximation error decays
rapidly as dimensions are added. While the deviation is substantial at
extreme compression (\(d=2\)), it drops sharply by \(d=10\) (retaining
only \textasciitilde12\% of the original dimensions), where the mean
absolute error falls below 0.1. By \(d=20\), the error becomes
negligible. This confirms that the MMD metric provides a highly stable
measure of distributional distinctiveness even under aggressive
dimensionality reduction.

\subsubsection{Input Perturbation}\label{input-perturbation}

To test robustness to data-level corruption, we evaluate how the MMD
\emph{statistic itself} responds to systematic and unsystematic image
perturbations. We apply Gaussian noise (unsystematic) and grid
watermarks (systematic) at varying intensities to images from three
digit classes---digits 1, 3, and 8---selected to span a range of visual
complexity (simple strokes to complex curves). For each digit, we
compare perturbed images against clean images from the same class using
a paired design. We parameterize perturbation strength using
Signal-to-Noise Ratio (SNR) for Gaussian noise and an analogous
Signal-to-Watermark Ratio (SWR) for systematic watermarks, where higher
ratios indicate weaker perturbations.

Under the null hypothesis (both sets drawn from the same underlying
distribution), we expect MMD to remain near zero when perturbation is
minimal. As perturbation intensity increases (SNR or SWR decreases), we
expect the MMD to grow, reflecting the increasing distributional
divergence between clean and corrupted images.

Table \ref{tab:mnist_delta_mmd} shows that the test statistic is
essentially invariant to both noise and watermark perturbations across
the entire tested range (SNR and SWR from 25 down to 1). For all three
digits, the change in MMD (\(\Delta\)MMD) between clean and perturbed
images remains negligibly small (approximately \(-0.004\), with negative
values reflecting variance in the unbiased estimator when distributions
are identical). Even at extreme perturbation levels (SNR = 1, where
noise standard deviation equals signal standard deviation), the MMD
statistic shows no meaningful deviation from zero.

\begin{table}[htbp]
\centering
\caption{Metric Stability: Deviation in Squared MMD ($\Delta$MMD) under Input Perturbation}
\label{tab:mnist_delta_mmd}
\begin{tabular}{llcccccccc}
\toprule
& & \multicolumn{8}{c}{Perturbation Level (SNR or SWR)} \\
\cmidrule(lr){3-10}
Digit & Metric & 25 & 20 & 15 & 10 & 5 & 3 & 2 & 1 \\
\midrule
\multicolumn{10}{l}{\textit{Gaussian Noise (SNR)}} \\
1 & $\Delta$MMD & $-$0.005 & $-$0.005 & $-$0.005 & $-$0.005 & $-$0.004 & $-$0.004 & $-$0.004 & $-$0.004 \\
3 & $\Delta$MMD & $-$0.004 & $-$0.004 & $-$0.004 & $-$0.004 & $-$0.004 & $-$0.004 & $-$0.004 & $-$0.004 \\
8 & $\Delta$MMD & $-$0.004 & $-$0.004 & $-$0.004 & $-$0.004 & $-$0.004 & $-$0.004 & $-$0.004 & $-$0.004 \\
\midrule
\multicolumn{10}{l}{\textit{Watermark (SWR)}} \\
1 & $\Delta$MMD & $-$0.005 & $-$0.004 & $-$0.004 & $-$0.004 & $-$0.004 & $-$0.004 & $-$0.004 & $-$0.003 \\
3 & $\Delta$MMD & $-$0.004 & $-$0.004 & $-$0.004 & $-$0.004 & $-$0.004 & $-$0.004 & $-$0.004 & $-$0.004 \\
8 & $\Delta$MMD & $-$0.004 & $-$0.004 & $-$0.004 & $-$0.004 & $-$0.004 & $-$0.004 & $-$0.004 & $-$0.004 \\
\bottomrule
\end{tabular}
\begin{minipage}{0.95\linewidth}
\vspace{0.5em}
\footnotesize
Note: Change in unbiased squared MMD between clean and perturbed images ($n=200$). Negative values reflect the variance of the unbiased estimator when distributions are identical or near-identical. The negligible magnitude across all levels confirms geometric stability.
\end{minipage}
\end{table}

\subsection{Inferential Robustness}\label{inferential-robustness}

In a legal context, the dispositive question is not merely whether the
measured distance changes, but whether the evidentiary determination is
robust to methodological choices. Having established that the MMD metric
is geometrically stable, we now examine whether the \emph{statistical
conclusions}---the decision to reject or fail to reject the null
hypothesis---remain reliable under varying conditions. We conduct five
ablation studies testing sensitivity to dimensionality, input
perturbation, kernel choice, bandwidth, and feature representation.

\subsubsection{Dimensionality}\label{dimensionality}

We test whether the hypothesis test retains its statistical power when
the 84-dimensional LeNet embeddings are reduced via UMAP to
\(d \in \{10, 84\}\) dimensions. Using digit pair (0 vs.~7), we repeat
the permutation test after reducing the 84-dimensional embeddings via
UMAP to \(d \in \{10, 84\}\) dimensions. Across 100 trials at
\(\alpha = 0.01\), rejection rates are remarkably stable: at \(n=6\),
rejection rates reach 1.00 for both \(d=10\) and full dimensionality
(\(d=84\)); for \(n \geq 7\), rejection reaches 1.00 regardless of
compression level. These results confirm that while dimensionality
reduction introduces approximation error in the MMD statistic itself
(Section 4.3.1), it does not compromise the evidentiary conclusion about
distributional distinctiveness.

\subsubsection{Input Perturbation}\label{input-perturbation-1}

Complementing the metric stability analysis in Section 4.3.2, we examine
whether noise and watermarks induce spurious rejections of the null
hypothesis. We use a \emph{paired design}: the same images serve as both
the clean baseline and the source for perturbation, isolating the effect
of the perturbation itself by controlling for content variation. We fit
UMAP on the pooled clean and perturbed samples before partitioning to
ensure comparison in a common projection space. Under \(H_0\), clean and
perturbed versions of the same images should be indistinguishable;
rejection indicates the perturbation has created a detectable
distributional shift.

\begin{table}[htbp]
\centering
\caption{Inferential Robustness: \textit{p}-values under Input Perturbation}
\label{tab:mnist_perturbation_inference}
\begin{tabular}{llcccccccc}
\toprule
& & \multicolumn{8}{c}{Perturbation Level (SNR or SWR)} \\
\cmidrule(lr){3-10}
Digit & Metric & 25 & 20 & 15 & 10 & 5 & 3 & 2 & 1 \\
\midrule
\multicolumn{10}{l}{\textit{Gaussian Noise (SNR)}} \\
1 & \textit{p}-value & 1.00 & 1.00 & 1.00 & 1.00 & 1.00 & 1.00 & 1.00 & 0.97 \\
3 & \textit{p}-value & 1.00 & 1.00 & 1.00 & 1.00 & 1.00 & 1.00 & 1.00 & 1.00 \\
8 & \textit{p}-value & 1.00 & 1.00 & 1.00 & 1.00 & 1.00 & 1.00 & 1.00 & 1.00 \\
\midrule
\multicolumn{10}{l}{\textit{Watermark (SWR)}} \\
1 & \textit{p}-value & 1.00 & 1.00 & 1.00 & 1.00 & 1.00 & 1.00 & 0.99 & 0.85 \\
3 & \textit{p}-value & 1.00 & 1.00 & 1.00 & 1.00 & 1.00 & 1.00 & 1.00 & 1.00 \\
8 & \textit{p}-value & 1.00 & 1.00 & 1.00 & 1.00 & 1.00 & 1.00 & 1.00 & 1.00 \\
\bottomrule
\end{tabular}
\begin{minipage}{0.95\linewidth}
\vspace{0.5em}
\footnotesize
Note: \textit{p}-values for MMD test comparing clean digit images against perturbed versions ($n=200$ per digit, $\alpha=0.01$). High \textit{p}-values indicate a failure to reject the null hypothesis, confirming that the method does not produce false positives even in the presence of significant image artifacts.
\end{minipage}
\end{table}

Table \ref{tab:mnist_perturbation_inference} presents the perturbation
sensitivity analysis. Remarkably, \emph{p}-values remain well above the
significance threshold (\(\alpha = 0.01\)) for all digits across the
entire range of perturbation intensities tested (SNR and SWR from 25
down to 1). Even at SNR = 1---where noise standard deviation equals the
signal standard deviation---\emph{p}-values remain non-significant for
digits 3 and 8, with digit 1 showing only marginal reduction (p = 0.97).
This extreme robustness is achieved by computing SNR relative to each
digit's own signal variance, ensuring comparable perturbation strength
across digits, combined with the noise-augmented training
(\(\sigma = 1.0\)) that produces highly noise-invariant features without
sacrificing discriminative accuracy (validation accuracy 99.1\%).

The robustness extends to systematic perturbations as well: watermarks
at all tested SWR levels (from 25 down to 1) fail to induce spurious
rejections for digits 3 and 8, while digit 1 shows slight sensitivity
only at the most extreme level (SWR = 1, p = 0.85). This finding has
important implications for forensic applications: minor image
artifacts---whether from compression, noise, or watermarking---do not
compromise the evidentiary conclusion about distributional
distinctiveness. The method correctly fails to reject the null
hypothesis (that clean and perturbed images come from the same
distribution) across virtually all tested perturbation levels.

\subsubsection{Kernel Choice}\label{kernel-choice}

To assess whether the detected distinctiveness depends on the kernel's
ability to capture nonlinear structure, we compare the Gaussian RBF
kernel against a linear kernel, which is sensitive only to differences
in distribution means (centroid shift).

For digit pair (0 vs.~7) at \(\alpha = 0.01\) over 100 trials, both
kernels achieve comparable inferential performance. At \(n=5\),
rejection rates are 0.68 (RBF) and 0.86 (linear); at \(n=6\), they reach
0.98 (RBF) and 1.00 (linear); and by \(n \geq 7\), both achieve 1.00.
The similar performance of the linear kernel suggests that, for MNIST
digits in the learned embedding space, distributional distinctiveness is
consistent with substantial centroid shift between digit classes. We
retain the RBF kernel for subsequent analyses because of its theoretical
guarantees as a characteristic kernel and its greater sensitivity to
nonlinear structure in more complex domains (Section 6).

\subsubsection{Bandwidth Sensitivity}\label{bandwidth-sensitivity}

The Gaussian RBF kernel relies on a bandwidth parameter \(\sigma\),
which we set using the median heuristic. To ensure our conclusions are
not artifacts of this specific tuning, we repeat the analysis with
\(\sigma\) scaled by factors of \(0.5\times\), \(1.0\times\) (baseline),
and \(2.0\times\).

For digit pair (0 vs.~7) at \(\alpha = 0.01\) over 100 trials, rejection
rates are virtually identical across bandwidth settings: at \(n=5\),
rejection ranges from 0.68 to 0.75; at \(n=6\), all settings achieve
0.98--1.00; and by \(n \geq 7\), all reach 1.00. This ``plateau of
significance'' confirms that the detected distinctiveness is a genuine
property of the underlying distributions rather than a brittle artifact
of hyperparameter optimization.

\subsubsection{Representation}\label{representation}

Finally, we test whether the learned embedding is necessary for reliable
distinctiveness detection by comparing MMD tests conducted on raw pixel
vectors (\(28 \times 28 = 784\) dimensions) versus the 84-dimensional
CNN embeddings. Table \ref{tab:mnist_representation} presents rejection
rates for three digit pairs spanning a range of visual similarity.

\begin{table}[htbp]
\centering
\caption{Representation Ablation: Rejection Rates ($\alpha=0.01$, 100 trials) for Raw Pixels vs. Learned Embeddings}
\label{tab:mnist_representation}
\begin{tabular}{lcccc}
\toprule
Digit Pair & \multicolumn{2}{c}{Raw Pixels (784-dim)} & \multicolumn{2}{c}{CNN Embedding (84-dim)} \\
\cmidrule(lr){2-3} \cmidrule(lr){4-5}
 & $n=6$ & $n=10$ & $n=6$ & $n=10$ \\
\midrule
1 vs. 7 (visually distinct) & 0.93 & 1.00 & 0.99 & 1.00 \\
3 vs. 5 (visually similar) & 0.22 & 0.66 & 1.00 & 1.00 \\
4 vs. 9 (visually similar) & 0.08 & 0.37 & 0.99 & 1.00 \\
\bottomrule
\end{tabular}
\begin{minipage}{0.95\linewidth}
\vspace{0.5em}
\footnotesize
Note: Raw pixels are flattened image vectors. CNN embeddings are 84-dimensional features from the trained LeNet-5 penultimate layer. Rejection rate (RR) is the fraction of 100 Monte Carlo trials rejecting $H_0$ at $\alpha = 0.01$.
\end{minipage}
\end{table}

The results reveal a critical interaction between representation choice
and visual similarity. For structurally distinct digits (1 vs.~7), raw
pixel representations achieve high rejection rates (0.93 at \(n=6\),
reaching 1.00 by \(n=10\)), comparable to the learned embeddings.
However, for visually confusable pairs (3 vs.~5; 4 vs.~9), raw pixels
perform substantially worse, achieving only 0.08--0.22 rejection at
\(n=6\) and failing to exceed 0.66 even at \(n=10\). In contrast, the
CNN embeddings achieve \(\geq 0.99\) rejection at \(n=6\) for all pairs.
This demonstrates that learned representations are not merely cosmetic
improvements but are essential for detecting distinctiveness when the
signal resides in higher-order structural features rather than raw
spatial patterns. This finding motivates the use of semantic embeddings
(CLIP) for the AI Art analysis in Section 6, where the relevant
distinctions are stylistic and conceptual rather than pixel-level.

\section{Validation: Patent Abstracts and Textual
Distinctiveness}\label{validation-patent-abstracts-and-textual-distinctiveness}

While the MNIST study establishes the statistical validity of our
framework in the visual domain, a significant portion of intellectual
property---particularly patent claims and literary copyright---is
textual. To demonstrate the method's versatility, we extend our
validation to the domain of natural language. We apply the MMD framework
to patent abstracts, testing whether the metric can reliably distinguish
between clearly demarcated technical fields defined by the International
Patent Classification (IPC) system.

This experiment serves two purposes. First, it validates the use of
text-based semantic embeddings (SentenceTransformers) within our
distributional framework. Second, it confirms that the method can detect
``prior art'' boundaries: if the framework functions correctly, it
should identify high distributional similarity within a specific
technical field (e.g., Chemistry) and high distributional divergence
between unrelated fields (e.g., Chemistry vs.~Electricity).

\subsection{Dataset and Experimental
Design}\label{dataset-and-experimental-design}

We utilize the \emph{CCDV Patent Classification} dataset, a benchmark
corpus derived from USPTO and EPO patent documents. The dataset
categorizes patents according to top-level IPC sections. To ensure a
rigorous test of distinctiveness, we select three topologically distinct
technical fields:

\begin{enumerate}
\def\labelenumi{\arabic{enumi}.}
\tightlist
\item
  \textbf{Section A (Human Necessities):} Covering agriculture, food,
  and personal goods.
\item
  \textbf{Section C (Chemistry):} Covering metallurgy and chemical
  engineering.
\item
  \textbf{Section H (Electricity):} Covering electronic circuitry and
  power generation.
\end{enumerate}

We draw 1000 abstracts from each section, using two disjoint halves of
500 for split-half negative controls; cross-section comparisons use 500
per section. To prevent data leakage, we apply strict preprocessing to
remove boilerplate header lines containing IPC section identifiers from
the text. This ensures that the embedding model relies solely on the
semantic content of the technical description rather than superficial
formatting artifacts---a concern analogous to the watermark sensitivity
analysis in the MNIST study (Section 4.4.2).

As with the MNIST study, we employ a \emph{split-half negative control}.
For each section, we partition the data into two disjoint subsets (e.g.,
Chemistry\(_A\) vs.~Chemistry\(_B\)). We expect the MMD between these
subsets to be statistically indistinguishable from zero, while
comparisons between sections (e.g., Chemistry vs.~Electricity) should
yield statistically significant divergence.

\subsection{Semantic Text Embeddings}\label{semantic-text-embeddings}

To map these textual descriptions into a numerical vector space, we
employ a sentence transformer model, specifically
\texttt{GIST-small-Embedding-v0}
(\citeproc{ref-solatorio2024gistembed}{Solatorio 2024}). Unlike
traditional keyword-based methods (such as TF-IDF) which measure lexical
overlap, this transformer network maps sentences to a 384-dimensional
dense vector space such that semantically similar texts are
geometrically close. We select this model for its superior
signal-to-noise ratio in our MMD framework compared to alternatives such
as \texttt{all-MiniLM-L6-v2}.

\subsection{Results: Technical
Distinctiveness}\label{results-technical-distinctiveness}

Figure \ref{fig:patent_heatmap} presents the pairwise MMD statistics for
the patent sections. The results confirm the method's efficacy in the
textual domain. The diagonal elements (negative controls) show
negligible MMD values, and the null hypothesis of identical
distributions is not rejected in any split-half comparison
(\(p > 0.01\)). This confirms that the metric does not detect
differences where none exist.

\begin{figure}[htbp]
\centering
\includegraphics[width=0.6\linewidth]{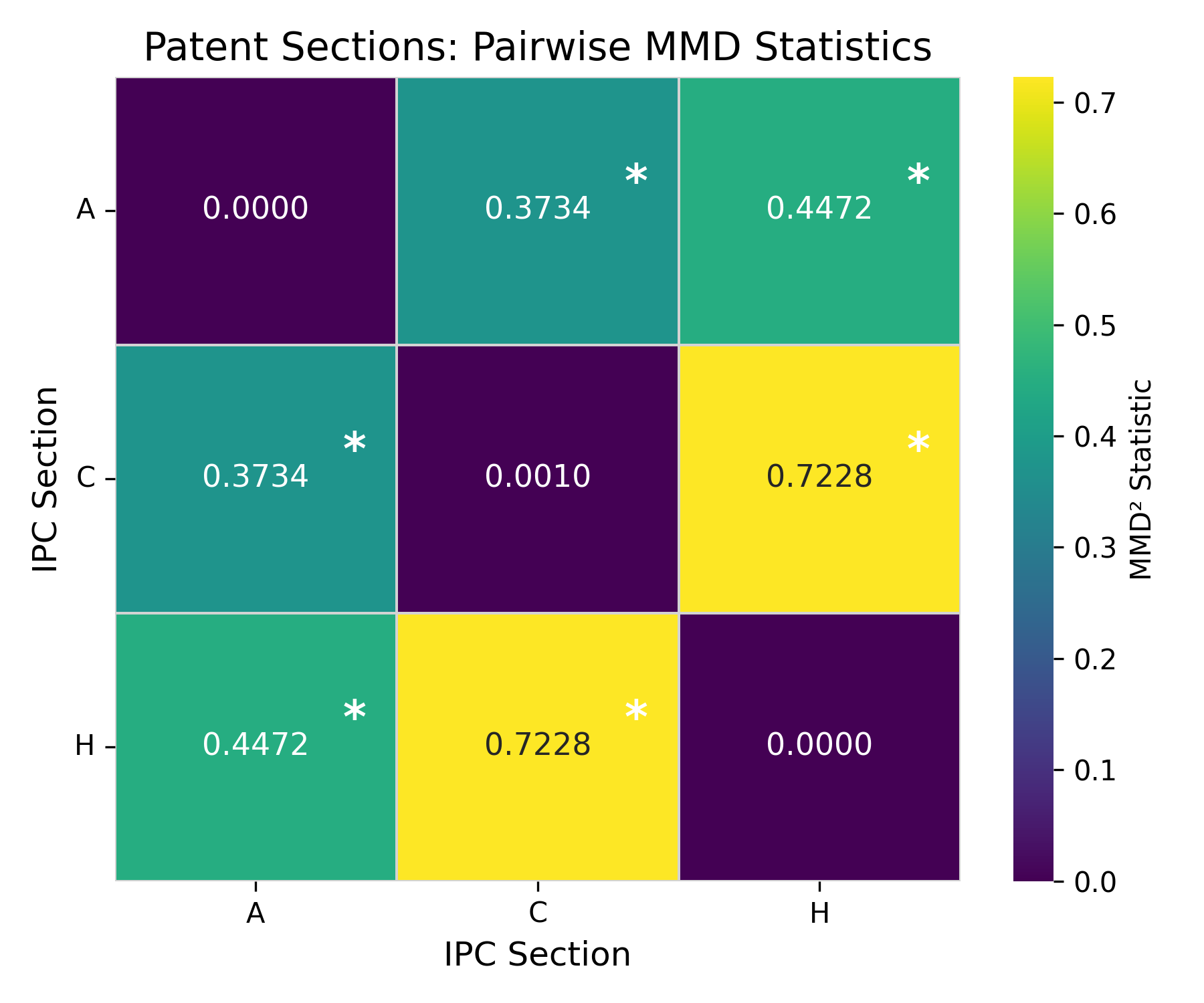}
\caption{Pairwise MMD Statistics for Patent IPC Sections}
\label{fig:patent_heatmap}
\begin{minipage}{0.95\linewidth}
\footnotesize
Note: Three IPC sections compared: A (Human Necessities), C (Chemistry), H (Electricity). Diagonal cells show split-half negative controls (disjoint samples from the same section); off-diagonal cells show cross-section comparisons. Asterisks indicate statistical significance ($p < 0.01$, permutation test with $R=500$).
\end{minipage}
\end{figure}

Conversely, all cross-section comparisons yield highly significant
differences (\(p < 0.01\)). The MMD statistic effectively captures the
semantic distance between technical fields. For example, the distance
between \emph{Chemistry} and \emph{Electricity} (distinct physical
sciences) is robustly detected. This suggests that the framework can
effectively map the ``topology of prior art,'' identifying whether a new
set of claims falls within the distribution of an existing field or
occupies a distinct region.

\subsubsection{Sample Efficiency in
Text}\label{sample-efficiency-in-text}

We further evaluate the data requirements for textual analysis. Figure
\ref{fig:patent_rejection_rate} plots the rejection rate of the null
hypothesis against sample size.

\begin{figure}[htbp]
\centering
\includegraphics[width=\linewidth]{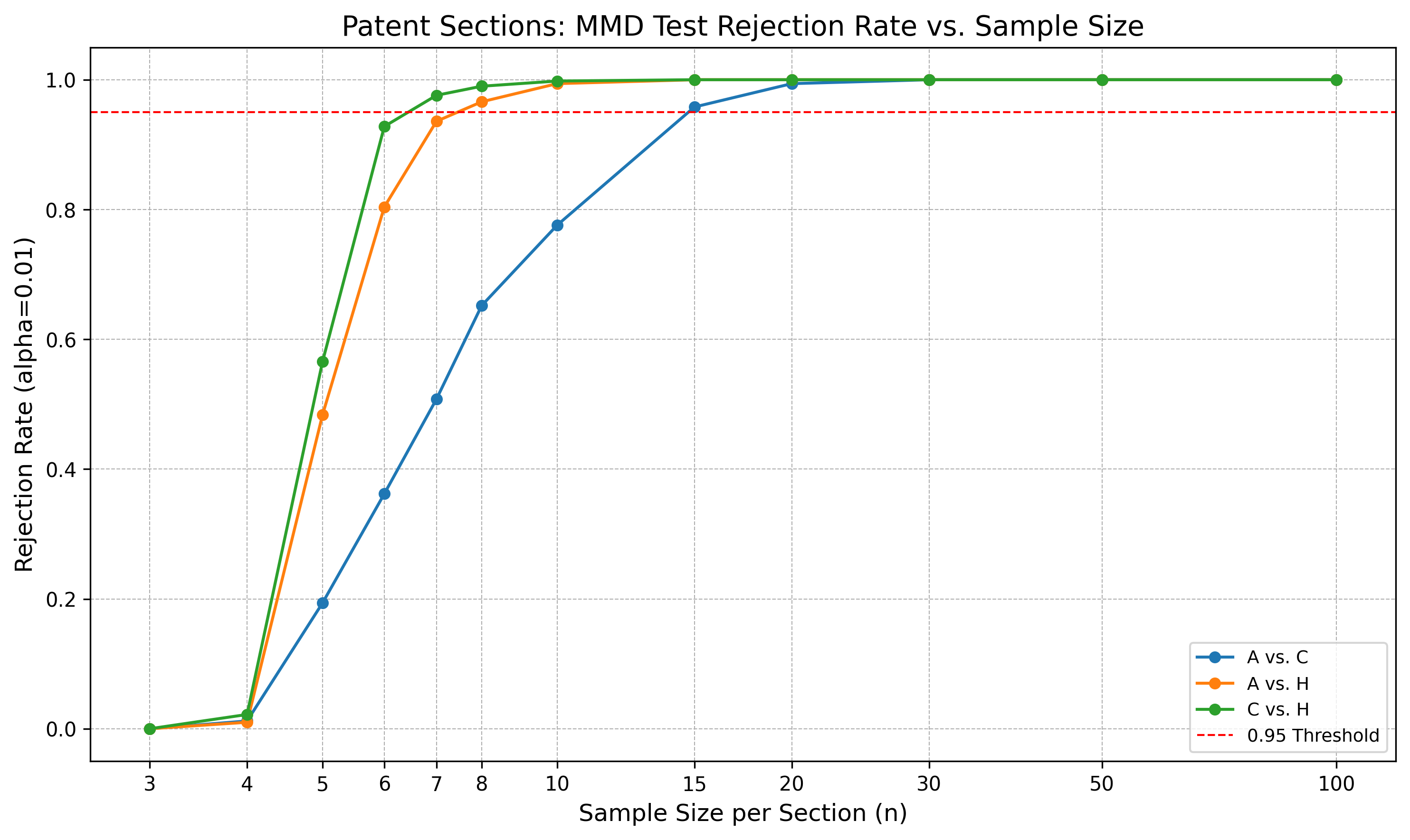}
\caption{Rejection Rate vs. Sample Size for Patent Section Comparisons}
\label{fig:patent_rejection_rate}
\begin{minipage}{0.95\linewidth}
\footnotesize
Note: Rejection rate (statistical power) for distinguishing IPC sections at $\alpha = 0.01$ across sample sizes. The most semantically distinct pair (Chemistry vs. Electricity) achieves $>95\%$ power at $n=7$; pairs with smaller effect sizes require $n=8$–$15$ abstracts per group.
\end{minipage}
\end{figure}

The results demonstrate reasonable sample efficiency that varies with
effect size. The most semantically distinct pair (Chemistry
vs.~Electricity, \(\text{MMD}^2=0.72\)) achieves 95\% rejection at
\(n=7\); pairs with smaller effect sizes (A vs.~H,
\(\text{MMD}^2=0.45\); A vs.~C, \(\text{MMD}^2=0.37\)) require
\(n=8\)--\(15\) abstracts per group. This finding is relevant for
litigation and examination contexts. While not as efficient as the
visual domain (where 5--8 samples suffice for MNIST), the ability to
reliably distinguish technical fields with 8--15 documents confirms the
method is deployable in patent disputes, where examiner searches
typically yield dozens of prior art references. Given the algorithmic
stability established in the MNIST validation (Section 4.4), we
attribute these significant MMD scores to genuine semantic differences
rather than artifacts of kernel choice or dimensionality.

\section{AI-Generated Art and the Perceptual
Paradox}\label{ai-generated-art-and-the-perceptual-paradox}

The preceding studies validate our methodology in controlled settings:
MNIST provides objective ground truth in the visual domain, and Patent
Abstracts demonstrate versatility in text. Creative art presents a more
complex challenge: it is subjective, stylistically diverse, and lacks
rigid class boundaries. Moreover, human evaluators distinguish
AI-generated art from human-created art with only approximately 58\%
accuracy---barely better than chance. Yet if AI models are genuinely
distinct creative processes rather than mere replicators of training
data, that distinctiveness should be statistically detectable, even if
it eludes direct observation. This section tests whether our MMD
framework can resolve this paradox: detecting distributional differences
even when human perception fails.

\subsection{The AI-ArtBench Dataset and
Categories}\label{the-ai-artbench-dataset-and-categories}

We employ the AI-ArtBench dataset
(\citeproc{ref-silva2024artbrain}{Silva et al. 2024}), a comprehensive
corpus designed to benchmark the detection and attribution of
AI-generated art. It comprises 185,015 artistic images spanning ten
distinct art styles (including Impressionism, Surrealism, and Art
Nouveau). The dataset includes both human-created artworks---60,000
images derived from the curated ArtBench-10 dataset
(\citeproc{ref-liao2022artbench}{Liao et al. 2022})---and 125,015
AI-generated images produced using text prompts based on the human
artworks. The AI images are generated by two prominent diffusion models:

\begin{itemize}
\tightlist
\item
  \textbf{Latent Diffusion (LD):} The original CompVis latent diffusion
  model, which serves as a foundational architecture for subsequent
  systems.
\item
  \textbf{Stable Diffusion (SD):} A widely deployed latent diffusion
  model (\citeproc{ref-rombach2022high}{Rombach et al. 2022}), noted for
  its high-resolution generation capabilities.\footnote{Silva et al.
    (\citeproc{ref-silva2024artbrain}{2024}) refer to this model as
    ``Standard Diffusion'' in their dataset and portions of their text,
    though they also use ``Stable Diffusion'' interchangeably (e.g.,
    Section 3.1). Their documentation points to the official Stable
    Diffusion repository, confirming the model identity.}
\end{itemize}

We categorize the images into three groups: Human (original human
artworks), AI (SD) (images generated by Stable Diffusion), and AI (LD)
(images generated by Latent Diffusion). This design allows us to test
distinctiveness along two axes: the divergence of machine from human,
and the divergence of one generative model from another.\footnote{We do
  not observe the full, proprietary training corpora for these diffusion
  models. Therefore, throughout, we treat the curated human artworks in
  AI-ArtBench as a proxy for samples from the relevant training
  distribution (e.g., treating the Impressionist subset of ArtBench-10
  as representative of the Impressionist data seen by the models during
  training).}

This dataset offers a crucial test case because human evaluators can
distinguish these AI-generated images from human art with only
approximately 58\% accuracy (\citeproc{ref-silva2024artbrain}{Silva et
al. 2024}). This near-collapse of perceptual distinctiveness exposes the
limits of the ``ordinary observer'' standard relied upon in copyright
and trademark law: if the human eye cannot reliably distinguish the
source, legal analysis requires a robust quantitative metric capable of
detecting underlying distributional differences that human intuition
misses.

In addition, it presents a puzzle: is the AI producing work that
\emph{appears} human because it is reproducing memorized training data
(pure regurgitation), or because it is interpolating learned patterns to
create genuinely novel yet semantically coherent outputs? The perceptual
data alone cannot distinguish these possibilities. Our MMD framework
allows us to contrast their distributional predictions. If the AI is
regurgitating, its output distribution should closely resemble a
human-art baseline. If it interpolates, the distributions should
diverge.

\subsection{Embedding with CLIP}\label{embedding-with-clip}

Assessing distinctiveness in the visual arts requires capturing complex
stylistic and semantic topologies. Constructing a domain-specific
feature extractor for such high-dimensional data requires
resource-intensive training on large-scale corpora. However, in actual
legal adjudication, courts and IP offices cannot be expected to train
bespoke feature extractors for every new dispute. Moreover, such
large-scale case-specific training corpora are unlikely to be available.
Consequently, a legally viable metric must be
``training-free''---capable of evaluating new works immediately using
general-purpose semantic knowledge, much like a human adjudicator.

To satisfy these constraints, we employ CLIP (Contrastive Language-Image
Pre-training) (\citeproc{ref-radford2021learning}{Radford et al. 2021}).
CLIP is a multimodal foundation model pre-trained on massive datasets of
image-text pairs to align visual and textual concepts. By mapping images
into a semantic vector space, CLIP captures high-level features---style,
subject matter, and composition---rather than mere pixel-level
correlations. This allows us to measure the semantic distance between
creative processes without requiring access to the AI's training data or
model weights.

Specifically, we use the ViT-H-14-quickgelu variant of CLIP, pre-trained
on the dfn5b dataset. We process each image through the encoder to
obtain a normalized 1024-dimensional embedding vector. This transforms
the raw pixel data into a semantic representation where geometric
distance corresponds to conceptual divergence. Relying on a pre-trained
foundation model ensures the framework is deployable for ``zero-shot''
legal analysis, allowing adjudicators to compare new image sources
immediately without technical overhead.

\subsection{MMD Analysis Procedure and
Setup}\label{mmd-analysis-procedure-and-setup-1}

We apply the MMD framework established in Section 4 to the AI-ArtBench
data. We perform comparisons at the category level---comparing Human
vs.~AI \emph{within} each artistic style---to prevent stylistic
confounds from masking or inflating distinctiveness estimates.
Therefore, we adapt our sampling strategy to account for stylistic
diversity. Specifically, to balance comparisons across artistic
movements, we employ \emph{stratified sampling}: we draw 250 images per
artistic style for each category (Human, AI-SD, AI-LD), yielding a total
dataset of 7,500 images. We pass each sampled image through the
pre-trained CLIP encoder to obtain its normalized 1024-dimensional
embedding vector and employ UMAP for dimensionality reduction to 64
dimensions. We compute the unbiased MMD statistic using the Gaussian RBF
kernel (with median heuristic bandwidth) and assess significance via the
permutation test (\(R=500\), \(\alpha=0.01\)).

\subsection{Core Results: Style-Stratified
Distinctiveness}\label{core-results-style-stratified-distinctiveness}

Figure \ref{fig:art_category_heatmaps} presents per-style MMD heatmaps
showing Human vs.~AI (SD) vs.~AI (LD) comparisons within each artistic
style. Across all styles, diagonal elements (negative controls) yield
MMD statistics near zero, with all 30 comparisons non-significant at
\(\alpha = 0.01\), consistent with correct Type I error control. In
contrast, Human vs.~AI comparisons yield statistically significant
differences (\(p < 0.01\)) across all ten styles, demonstrating that
distributional distinctiveness is a robust feature that persists
regardless of artistic movement.

\begin{figure}[htbp]
\centering
\includegraphics[width=\linewidth]{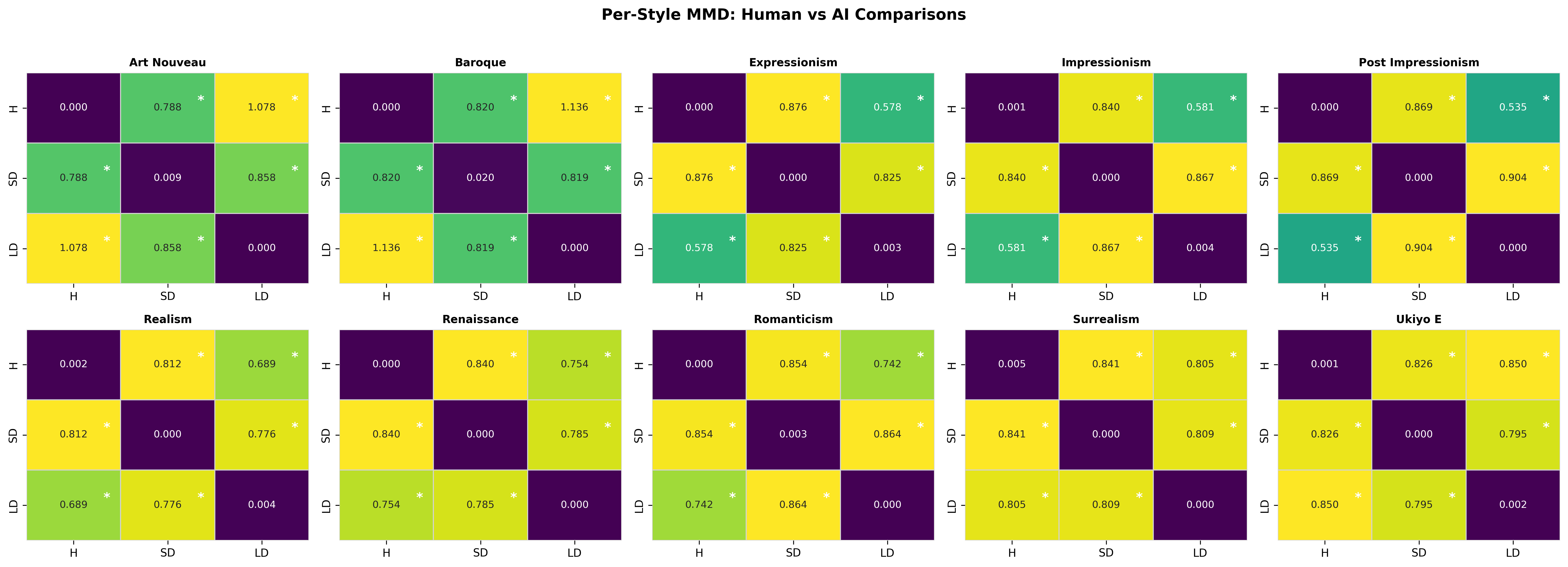}
\caption{Per-Style MMD Heatmaps: Human vs AI Comparisons Within Each Artistic Movement.}
\label{fig:art_category_heatmaps}
\begin{minipage}{0.95\linewidth}
\footnotesize
Note: Each panel shows Human (H), AI (SD), and AI (LD) comparisons within a single artistic style. Diagonal cells (negative controls) show near-zero MMD values; all 30 diagonal comparisons are non-significant at $\alpha = 0.01$, consistent with correct Type I error control. Off-diagonal cells marked with * indicate statistically significant differences (\(p < 0.01\), permutation test with \(R=500\)).
\end{minipage}
\end{figure}

Notably, the magnitude of distinctiveness varies substantially across
styles, ranging from \(\text{MMD}^2 = 0.788\) for Art Nouveau to
\(0.876\) for Expressionism. Figure
\ref{fig:art_category_mmd_comparison} ranks styles by their Human vs.~AI
(SD) MMD values, revealing a convergence spectrum: some styles (e.g.,
those with more constrained compositional conventions) show lower MMD
values, indicating that AI outputs more closely approximate the human
distribution, while others show higher divergence. This variation is
significant---it suggests that the strength of a ``distinctiveness''
argument may depend on the specific artistic domain under consideration.

\begin{figure}[htbp]
\centering
\includegraphics[width=\linewidth]{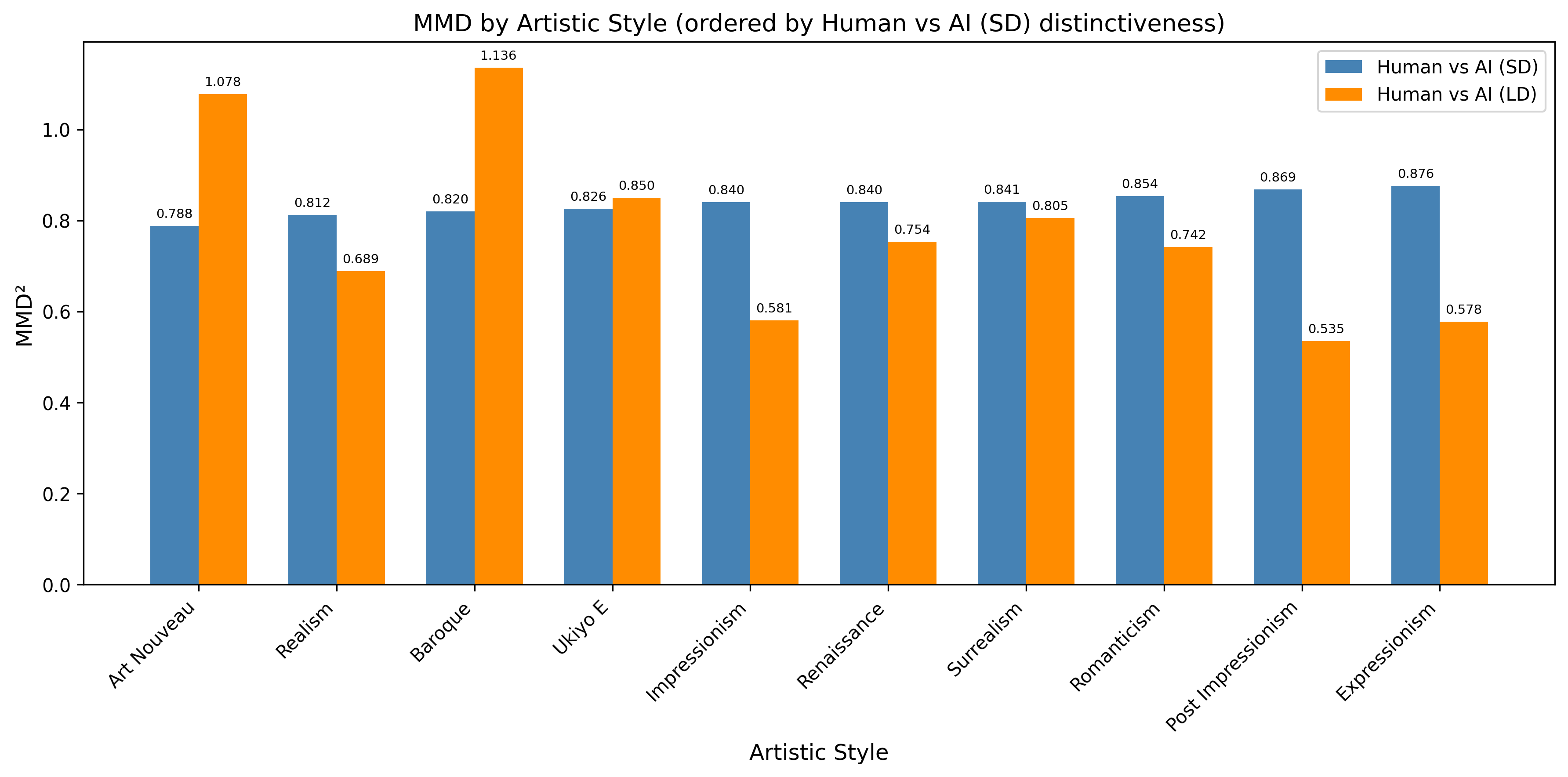}
\caption{MMD by Artistic Style (ordered by Human vs AI (SD) distinctiveness)}
\label{fig:art_category_mmd_comparison}
\begin{minipage}{0.95\linewidth}
\footnotesize
Note: Bars show $\text{MMD}^2$ for Human vs AI (SD) and Human vs AI (LD) within each style. Styles are ordered from lowest to highest Human vs AI (SD) distinctiveness, revealing the convergence spectrum across artistic movements.
\end{minipage}
\end{figure}

However, a raw MMD score is abstract without a baseline. To interpret
these magnitudes, we calibrate them against the distances between human
art movements calculated within the same embedding space. This
establishes an empirical reference scale based purely on the model's
topological constraints. Within this framework, closely related schools
(Impressionism vs.~Realism) yield a noise floor of
\(\text{MMD}^2 = 0.027\), while historically distinct traditions (Art
Nouveau vs.~Ukiyo-e) yield a ceiling of \(\text{MMD}^2 = 0.859\).

Against this scale, the AI outputs do not cluster near the human
baseline. The Human-AI (SD) distances (mean 0.837) fall at the upper
extreme of the spectrum, comparable to the distance between Art Nouveau
and Ukiyo-e. Even Latent Diffusion, the earliest model, achieves a mean
distance of 0.775---comparable to the separation between Baroque and
Surrealism (0.744). This calibration suggests that the detected
distinctiveness is not a statistical artifact of large sample sizes, but
a distributional divergence comparable in magnitude to the largest
stylistic shifts observed within the human reference set.

Figure \ref{fig:art_category_rejection_rates_sd} and Figure
\ref{fig:art_category_rejection_rates_ld} demonstrate the data
efficiency of our approach across styles. When comparing human art
against Latent Diffusion (LD) outputs, all styles achieve
\textgreater95\% rejection rates with sample sizes of \(n=6\) to
\(n=10\) images per source type. There is heterogeneity in styles:
``fast-converging'' styles (those with lower MMD) require slightly
larger samples to achieve high statistical power, while
``slow-converging'' styles (higher MMD) can be distinguished with fewer
samples. A similar analysis comparing human art against Standard
Diffusion (SD) outputs shows a faster convergence to high rejection
rates. Most styles converge with \(n=6\) images per source type and all
styles converge with \(n=7\) images per source type.

\begin{figure}[htbp]
\centering
\includegraphics[width=\linewidth]{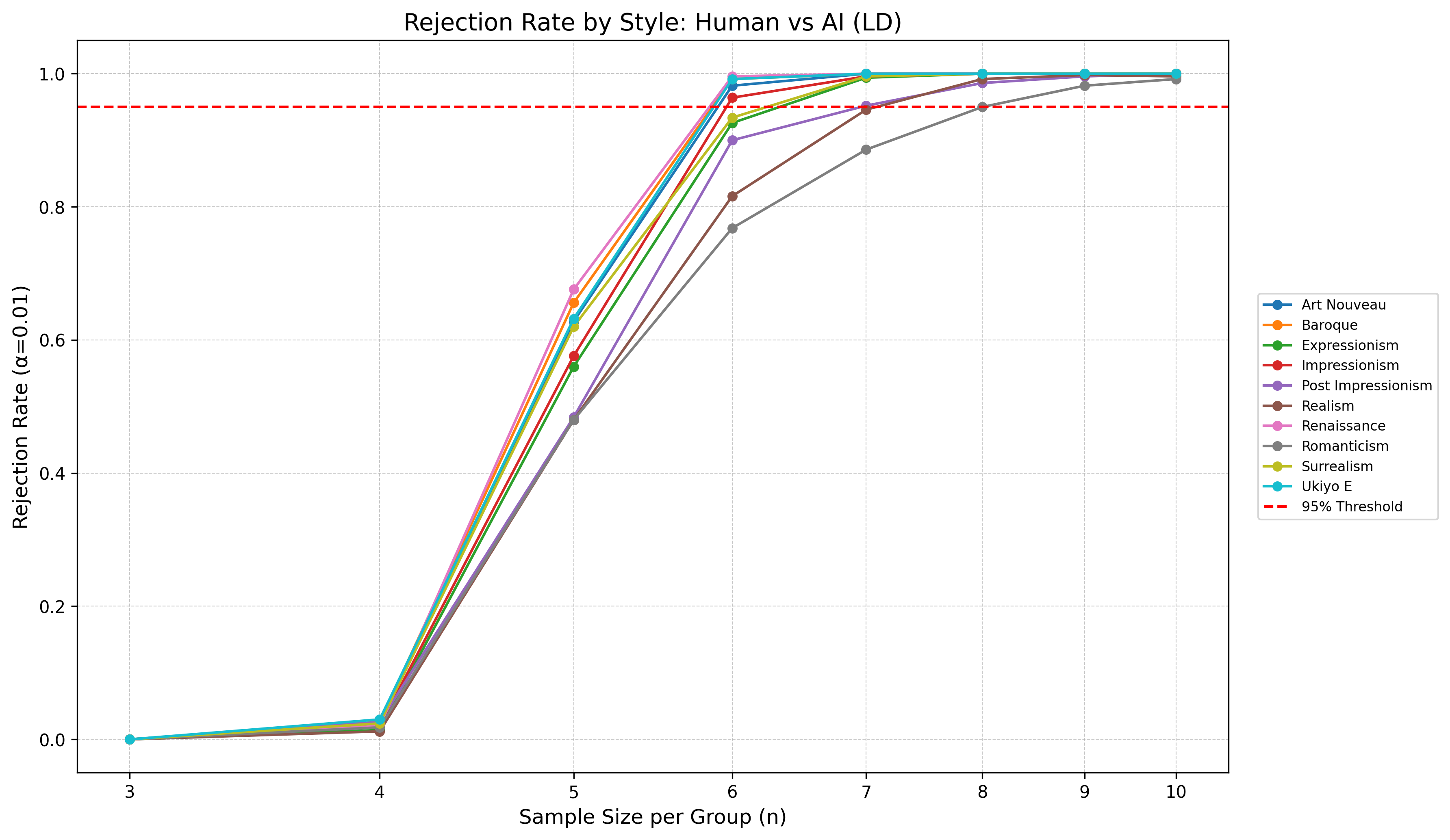}
\caption{Rejection Rate vs. Sample Size by Artistic Style (Human vs Latent Diffusion)}
\label{fig:art_category_rejection_rates_ld}
\begin{minipage}{0.95\linewidth}
\footnotesize
Note: Each line represents the rejection rate for a single artistic style at \(\alpha = 0.01\). The dashed line at 0.95 depicts the threshold for reliable detection.
\end{minipage}
\end{figure}

\begin{figure}[htbp]
\centering
\includegraphics[width=\linewidth]{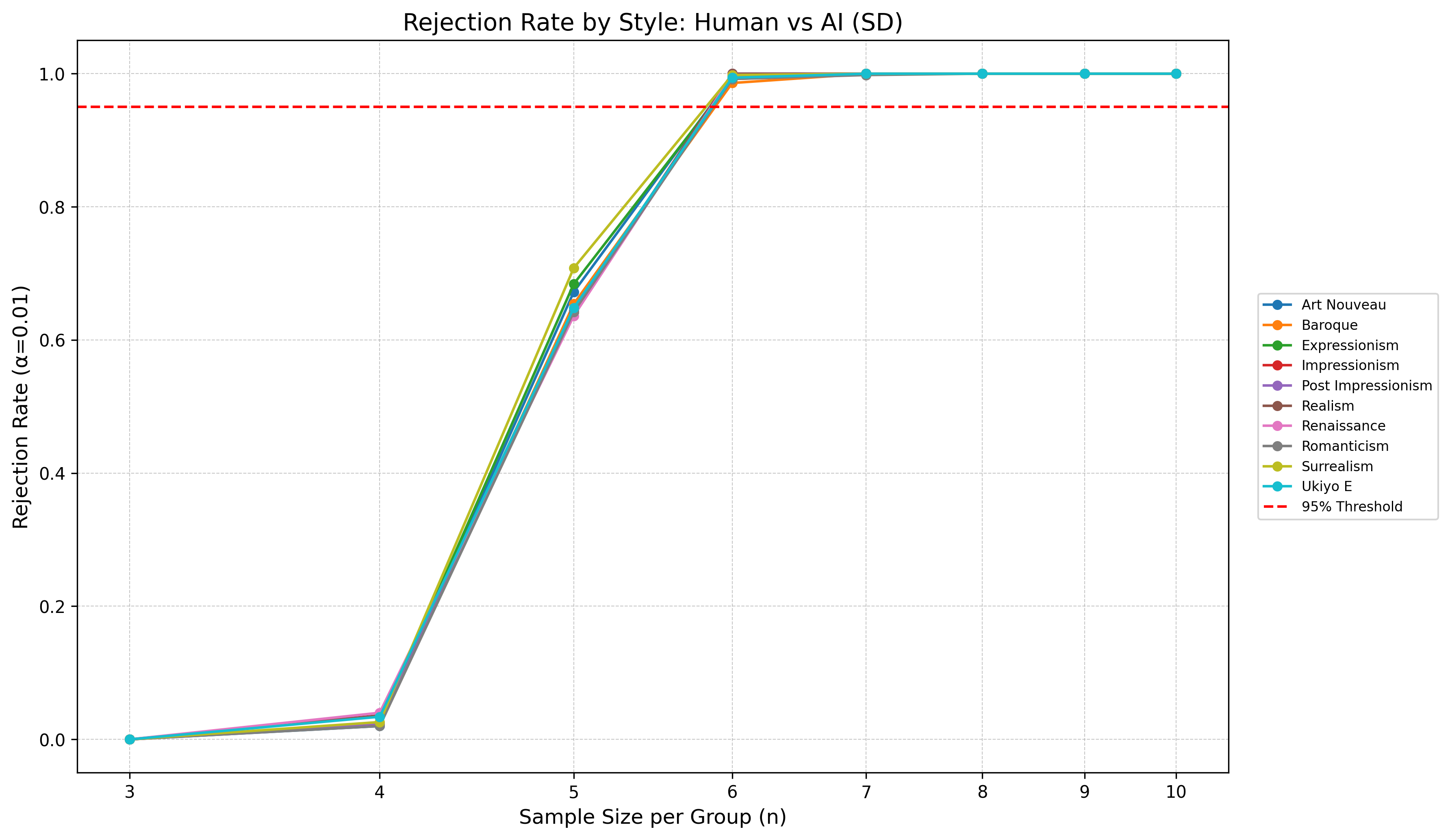}
\caption{Rejection Rate vs. Sample Size by Artistic Style (Human vs Standard Diffusion)}
\label{fig:art_category_rejection_rates_sd}
\begin{minipage}{0.95\linewidth}
\footnotesize
Note: Each line represents the rejection rate for a single artistic style at \(\alpha = 0.01\). The dashed line at 0.95 depicts the threshold for reliable detection.
\end{minipage}
\end{figure}

\subsection{Robustness of Core
Results}\label{robustness-of-core-results}

We evaluate the stability of these findings through the ablation suite
defined in Section 4.4, testing sensitivity to three algorithmic choices
under the researcher's control: (1) kernel choice, (2) bandwidth, and
(3) dimensionality. For each ablation, we conduct a rejection rate
analysis as in Section 6.4, varying the number of samples per group and
measuring the proportion of trials that correctly reject the null
hypothesis at \(\alpha = 0.01\). Table \ref{tab:art_ablation_summary}
summarizes the results.

\begin{table}[htbp]
\centering
\caption{Algorithmic Robustness: Sample Size for 95\% Rejection Rate Across Ablation Conditions}
\label{tab:art_ablation_summary}
\begin{tabular}{llcc}
\toprule
\textbf{Ablation} & \textbf{Scope} & \textbf{Human vs AI (SD)} & \textbf{Human vs AI (LD)} \\
\midrule
Kernel & Realism & 6 & 10 \\
       & Renaissance & 6 & 6 \\
       & Post-Impressionism & 6 & 8 \\
       & \textit{Overall} & 6 & 13 \\
\cmidrule(lr){1-4}
Bandwidth & Realism & 6 & 8 \\
          & Renaissance & 6 & 6 \\
          & Post-Impressionism & 6 & 8 \\
          & \textit{Overall} & 6 & 16 \\
\cmidrule(lr){1-4}
Dimensionality & Realism & 17 & 17 \\
               & Renaissance & 17 & 17 \\
               & Post-Impressionism & 17 & 17 \\
               & \textit{Overall} & 17 & 9 \\
\bottomrule
\end{tabular}
\begin{minipage}{0.95\linewidth}
\vspace{0.5em}
\footnotesize
Note: Each cell reports the sample size (per group) at which the rejection rate first reaches 0.95, taking the maximum across all tested parameter values within each ablation type (RBF/Linear for kernel; $0.5\times$/$1\times$/$2\times$ median for bandwidth; $d \in \{5, 10, 16, 32, 1024\}$ for dimensionality). The higher dimensionality thresholds for Realism, Renaissance, and Overall ($n=17$) reflect the constraint that UMAP reduction to $d=32$ dimensions requires at least 34 samples (17 per group).
\end{minipage}
\end{table}

Across all three algorithmic dimensions, the detected distinctiveness is
robust. For Human vs.~AI (SD), kernel choice and bandwidth show uniform
convergence by \(n=6\) across all styles, while dimensionality requires
larger samples (\(n=17\)) due to UMAP's constraint that target
dimensions must be less than the number of samples. For Human vs.~AI
(LD), thresholds are moderately higher---ranging from \(n=6\) to
\(n=17\) depending on ablation type and scope---reflecting the somewhat
smaller effect size for this earlier generative model (Section 6.5). The
Overall LD comparison requires \(n=13\) for kernel ablation and \(n=16\)
for bandwidth ablation. These results confirm that the finding of
distributional divergence is not brittle; it persists regardless of
specific parameter tuning, insulating the method from claims that the
results are artifacts of p-hacking or idiosyncratic kernel selection.

We turn finally to input perturbation (Table
\ref{tab:art_perturbation}), the most relevant stress test for
real-world forensic applications where images may be subject to
compression, noise, or watermarking. We compare clean human artworks
against perturbed versions of the same artworks; under the null
hypothesis, these should be indistinguishable. A robust semantic metric
should fail to reject the null hypothesis (i.e., find no difference)
unless the perturbation destroys the underlying artistic content.

\begin{table}[htbp]
\centering
\caption{Art Perturbation Robustness: \textit{p}-values Across Perturbation Levels}
\label{tab:art_perturbation}
\begin{tabular}{llcccccccc}
\toprule
& & \multicolumn{8}{c}{Perturbation Level (SNR or SWR)} \\
\cmidrule(lr){3-10}
Style & Metric & 25 & 20 & 15 & 10 & 5 & 3 & 2 & 1 \\
\midrule
\multicolumn{10}{l}{\textit{Gaussian Noise (SNR)}} \\
Realism & p-value & 1.00 & 1.00 & 1.00 & 0.97 & 0.57 & 0.07 & $<$0.01 & $<$0.01 \\
Renaissance & p-value & 0.94 & 0.97 & 0.98 & 0.69 & 0.06 & $<$0.01 & $<$0.01 & $<$0.01 \\
Post-Impr. & p-value & 1.00 & 1.00 & 1.00 & 0.99 & 0.86 & 0.40 & $<$0.01 & $<$0.01 \\
\midrule
\multicolumn{10}{l}{\textit{Watermark (SWR)}} \\
Realism & p-value & 1.00 & 1.00 & 1.00 & 1.00 & 1.00 & 0.99 & 0.86 & 0.31 \\
Renaissance & p-value & 1.00 & 1.00 & 1.00 & 1.00 & 1.00 & 1.00 & 0.96 & 0.72 \\
Post-Impr. & p-value & 1.00 & 1.00 & 1.00 & 1.00 & 1.00 & 1.00 & 0.98 & 0.68 \\
\bottomrule
\end{tabular}
\begin{minipage}{0.95\linewidth}
\vspace{0.5em}
\footnotesize
Note: \textit{p}-values for MMD test comparing clean vs. perturbed human artworks ($n=200$ per style, $\alpha=0.01$). SNR = Signal-to-Noise Ratio; SWR = Signal-to-Watermark Ratio. Higher values indicate weaker perturbation. Three representative styles from MMD terciles: realism (fast-converging), renaissance (median), post-impressionism (slow-converging).
\end{minipage}
\end{table}

At typical real-world degradation levels (SNR/SWR \(\geq\) 10),
\emph{p}-values remain well above the significance threshold
(\(\alpha = 0.01\)). This confirms that minor image artifacts do not
induce false positives. The test begins rejecting only at
moderate-to-substantial perturbation levels: for Gaussian noise,
renaissance reaches significance at SNR \(\leq\) 3, while
post-impressionism shows the greatest robustness (significant only at
SNR = 2). Watermarks show even stronger robustness, with \emph{p}-values
remaining non-significant across all tested SWR levels for all styles.
This establishes the critical forensic property: the metric is sensitive
to semantic divergence (the creative process) but invariant to
incidental degradation (the file quality).

\subsection{Memorization Checks}\label{memorization-checks}

While distributional distinctiveness characterizes the dominant creative
mode, it does not preclude rare item-level memorization. A generative
model that is predominantly interpolative may nonetheless produce
occasional outputs that closely resemble specific training examples. To
address this concern, we complement our distributional analysis with a
nearest-neighbor audit using metrics standard in the forensic
literature: semantic similarity (CLIP cosine), structural similarity
(SSIM), and perceptual similarity (LPIPS).

For each AI-generated image in our dataset (\(n = 2{,}500\) per model),
we identify its nearest neighbor in the human reference corpus
(\(n = 2{,}500\)) under three metrics:

\begin{enumerate}
\def\labelenumi{\arabic{enumi}.}
\tightlist
\item
  \textbf{CLIP cosine similarity}: Semantic proximity in the CLIP
  ViT-H-14 embedding space (1024 dimensions). Higher values indicate
  greater semantic overlap.
\item
  \textbf{SSIM (Structural Similarity Index)}: Pixel-level structural
  correspondence accounting for luminance, contrast, and structure
  (\citeproc{ref-wang2004image}{Wang et al. 2004}). Values range from -1
  to 1, with 1 indicating identical images.
\item
  \textbf{LPIPS (Learned Perceptual Image Patch Similarity)}: Perceptual
  distance computed from deep features
  (\citeproc{ref-zhang2018unreasonable}{Zhang et al. 2018}). Lower
  values indicate greater perceptual similarity.
\end{enumerate}

To establish a principled detection threshold, we calibrate the metrics
against a human baseline. For each human image, we compute its
nearest-neighbor similarity to other human images within the same style,
establishing the distribution of similarity expected among independent
works in a genre. We then set the detection threshold at the 99th
percentile of this human-human distribution. Any AI output flagged by
this audit thus exhibits similarity to a training example that exceeds
99\% of the similarities found between independent human
artworks---holding AI to the same standard as human artists. Under the
null hypothesis of no memorization, this threshold yields an expected
false-positive rate of 1\% (approximately 25 images per model).

\begin{table}[htbp]
\centering
\caption{Item-Level Memorization Audit Results}
\label{tab:memorization_audit}
\begin{tabular}{lcccccc}
\toprule
& \multicolumn{3}{c}{\textbf{Stable Diffusion}} & \multicolumn{3}{c}{\textbf{Latent Diffusion}} \\
\cmidrule(lr){2-4} \cmidrule(lr){5-7}
& CLIP & SSIM & LPIPS & CLIP & SSIM & LPIPS \\
\midrule
Flagged ($n$)        & 0 & 0 & 0  & 0 & 1  & 6 \\
Exceedance (\%)      & 0.00 & 0.00 & 0.00 & 0.00 & 0.04 & 0.24 \\
\bottomrule
\end{tabular}
\begin{minipage}{0.95\linewidth}
\vspace{0.5em}
\footnotesize
Note: Flagged images exceed the 99th percentile of human-human within-style nearest-neighbor similarity. Under the null hypothesis of no memorization, the expected false-positive rate is 1\% ($\approx$25 images per model). All observed exceedance rates fall well below this threshold, indicating no systematic memorization. Total samples: 2,500 per AI model.
\end{minipage}
\end{table}

Table \ref{tab:memorization_audit} reports the results. CLIP
similarity---the metric most sensitive to semantic content---flagged
zero images for both models, indicating no high-level conceptual
copying. SSIM and LPIPS, which capture lower-level structural and
perceptual features, flagged a small number of candidates: 0 images for
Stable Diffusion (0.00\%) and 7 images for Latent Diffusion (0.28\%).
Critically, both rates fall well below the 1\% false-positive rate
expected by chance, providing statistical evidence against systematic
memorization.

To adjudicate these flags, we conducted visual inspection of all 7
flagged cases (all from Latent Diffusion). Figures
\ref{fig:flagged_lpips_renaissance}, \ref{fig:flagged_lpips_baroque},
and \ref{fig:flagged_ssim_expressionism} present these comparisons. In
every instance, the flagged AI images depict entirely different subjects
from their matched human paintings. For example, Figure
\ref{fig:flagged_lpips_renaissance} shows five AI-generated images
flagged against a Correggio religious scene---images that share
Renaissance stylistic characteristics (warm color palettes, painterly
textures) but lack any compositional or subject-matter overlap. All
flags arose from incidental stylistic convergence rather than
reproductive memorization.

\begin{figure}[htbp]
\centering
\includegraphics[width=\textwidth]{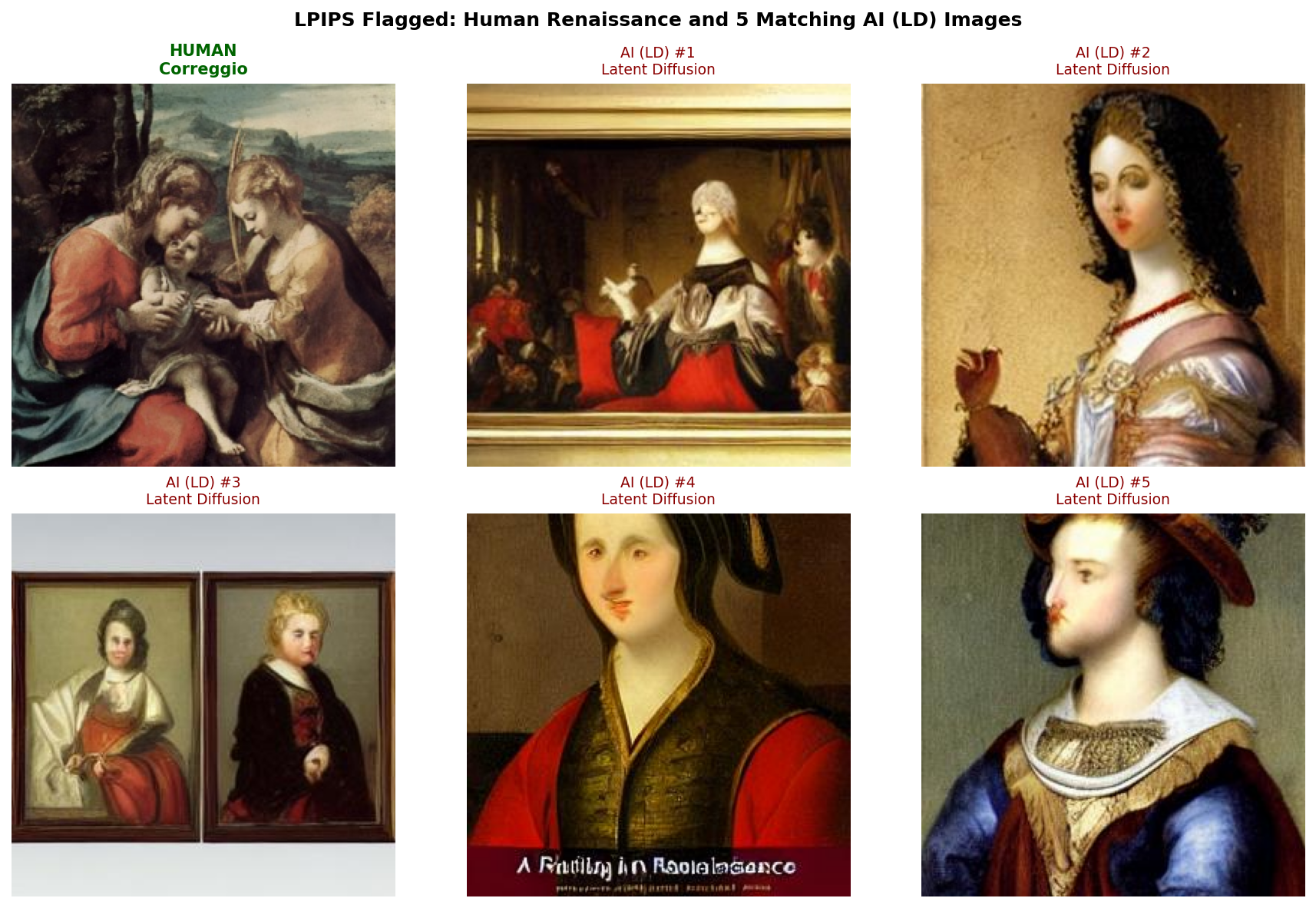}
\caption{LPIPS-Flagged Renaissance Comparisons}
\label{fig:flagged_lpips_renaissance}
\begin{minipage}{0.95\linewidth}
\footnotesize
Note: Human reference (Correggio, \textit{The Mystic Marriage of St.\ Catherine}, 1518) and 5 flagged Latent Diffusion outputs. The LPIPS metric detected shared warm color palettes and soft textural qualities characteristic of Renaissance painting. However, subjects differ entirely: the reference depicts a multi-figure religious scene while the AI outputs are predominantly single-figure portraits or framed scenes. Several AI images contain visible text artifacts (``Renaissance Art Style''), a hallmark of early diffusion models that further confirms non-reproductive generation.
\end{minipage}
\end{figure}

\begin{figure}[htbp]
\centering
\includegraphics[width=0.9\textwidth]{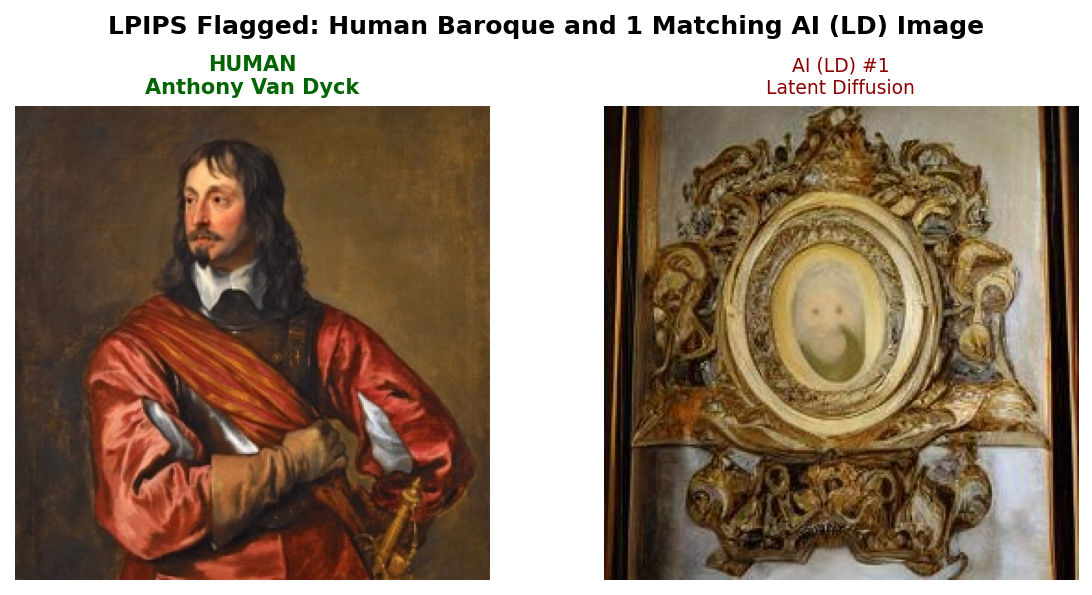}
\caption{LPIPS-Flagged Baroque Comparisons}
\label{fig:flagged_lpips_baroque}
\begin{minipage}{0.95\linewidth}
\footnotesize
Note: Human reference (Anthony van Dyck, \textit{Portrait of Sir John Mennes}, c.\ 1640) and 1 flagged Latent Diffusion output. The LPIPS metric detected shared warm red and gold color palettes characteristic of Baroque painting. However, the subjects differ substantially: a male portrait in armor versus an ornate framed religious scene. The similarity reflects period-typical color choices rather than content replication.
\end{minipage}
\end{figure}

\begin{figure}[htbp]
\centering
\includegraphics[width=0.8\textwidth]{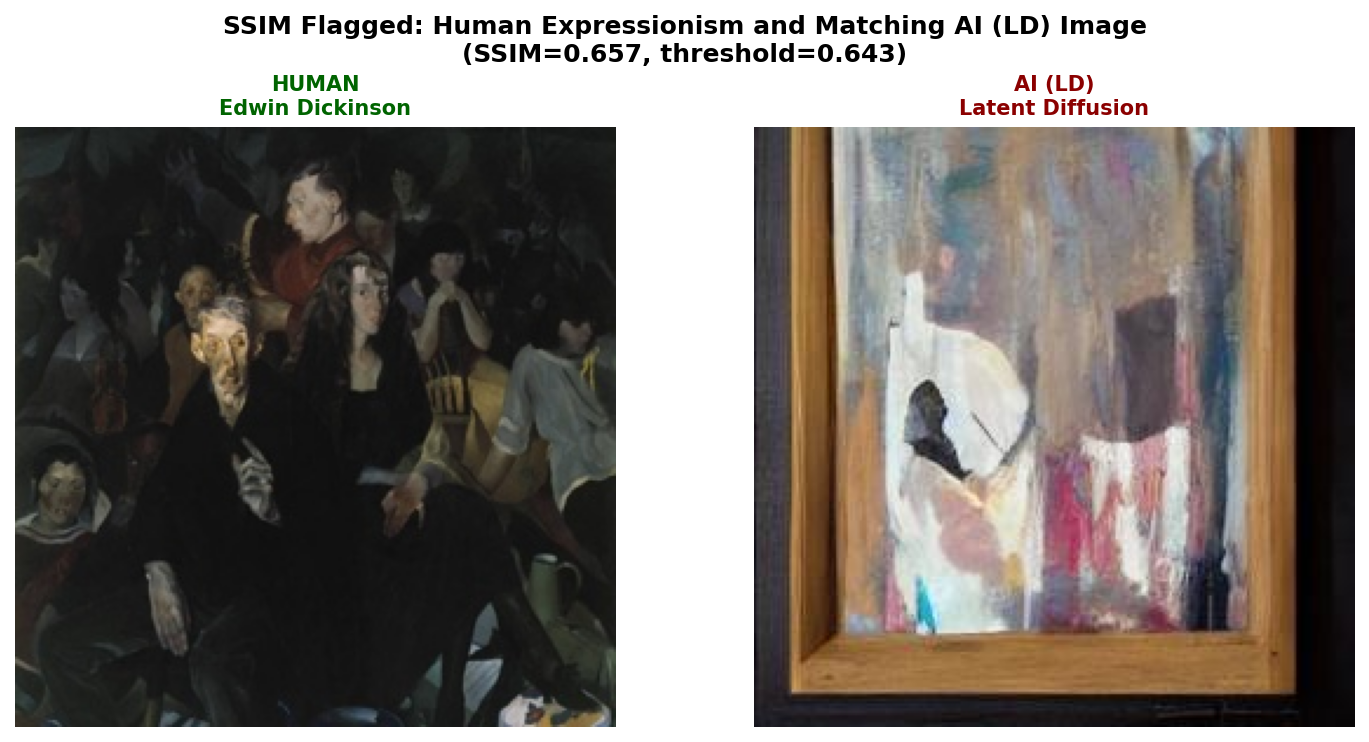}
\caption{SSIM-Flagged Expressionism Comparison}
\label{fig:flagged_ssim_expressionism}
\begin{minipage}{0.95\linewidth}
\footnotesize
Note: Human reference (Edwin Dickinson, \textit{An Anniversary}, 1921) and 1 flagged Latent Diffusion output. Despite strikingly different color palettes---the human work is dark and muted while the AI output is vibrant---the SSIM metric flagged this pair based on similar structural density and brushwork patterns. This illustrates how pixel-level structural metrics, which operate on luminance channels, can conflate textural similarity with copying even when semantic content and color are entirely distinct.
\end{minipage}
\end{figure}

In sum, the item-level memorization audit detected no credible instances
of training-data replication. The few statistical flags represent false
positives at threshold boundaries, confirmed by visual inspection to
reflect superficial stylistic similarity rather than copied content.
Combined with the distributional MMD analysis, these findings establish
that text-to-image models generate novel interpolations within the
learned manifold rather than regurgitating training examples.

\subsection{Evolution of Generative
Distinctiveness}\label{evolution-of-generative-distinctiveness}

As text-to-image models become more sophisticated, their outputs
converge semantically with real-world images and human art. But does
this fidelity arise because the outputs are moving closer to their
training samples, or because the models are interpolating more
effectively? Put differently: do better outputs reflect better
regurgitation or better learning?

To investigate this, we extend our analysis beyond the AI-ArtBench data
to include three subsequent generations of diffusion models. Our
expanded dataset spans five years of development: Latent Diffusion (Dec
2021), Stable Diffusion v1.4 (Aug 2022), Stable Diffusion XL (Jul 2023),
FLUX (Aug 2024), and FLUX-Krea (Jul 2025).\footnote{We deliberately
  employ models within a single architectural lineage rather than
  comparing across model families (e.g., Midjourney, DALL-E, Imagen) as
  otherwise we might conflate capability improvements with differences
  in architecture, training data, and design philosophy; by staying
  within the Stable Diffusion lineage, we isolate the effect of
  advancing capability on distributional distinctiveness. This choice
  also has practical relevance: Stability AI, the developer of Stable
  Diffusion, is currently a defendant in high-profile intellectual
  property litigation---including \emph{Andersen v. Stability AI Ltd.}
  and \emph{Getty Images (US), Inc.~v. Stability AI Ltd.}---making our
  findings directly relevant to the empirical questions at the heart of
  these disputes.}

\begin{figure}[htbp]
\centering
\includegraphics[width=\linewidth]{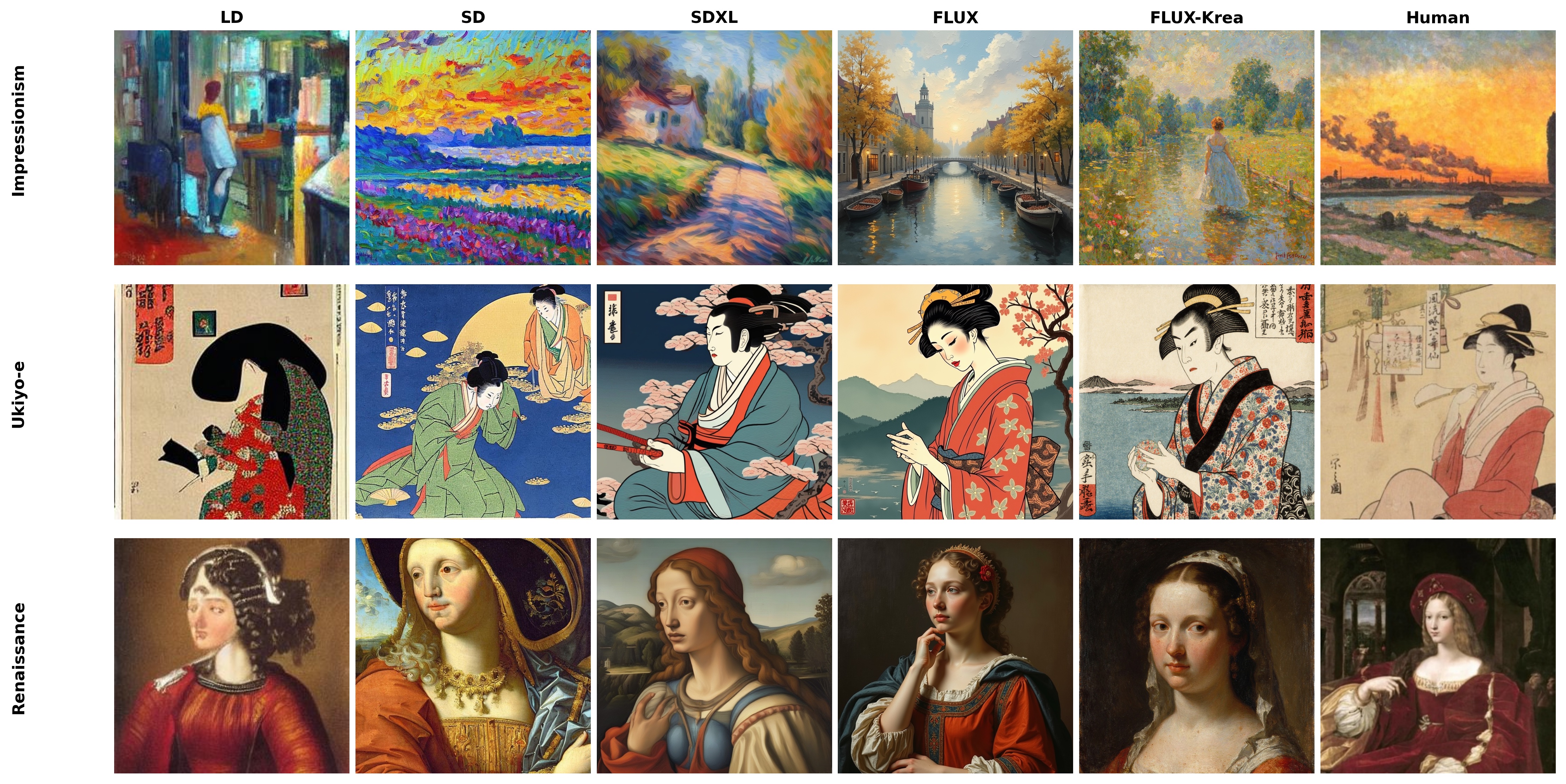}
\caption{Exemplar Images Across Model Generations and Artistic Styles}
\label{fig:art_evolution_grid}
\begin{minipage}{0.95\linewidth}
\footnotesize
Note: Each column represents a model generation, ordered chronologically from Latent Diffusion (LD) to FLUX-Krea, with human originals in the final column for comparison. Rows correspond to three artistic styles (Impressionism, Ukiyo-e, Renaissance) with consistent subject matter.
\end{minipage}
\end{figure}

Figure \ref{fig:art_evolution_grid} illustrates the visual progression
across model generations. Early models (LD, SD) exhibit visible
artifacts, such as blocky textures and inconsistent details. Later
generations (SDXL, FLUX) eliminate these flaws but tend toward
photorealism, producing outputs that often resemble photographs or 3D
renders more than paintings. FLUX-Krea, a fine-tuned variant designed
specifically to enhance artistic aesthetics, strikes a balance: it
recovers much of the stylistic fidelity while maintaining high
coherence.

\begin{figure}[htbp]
\centering
\includegraphics[width=\linewidth]{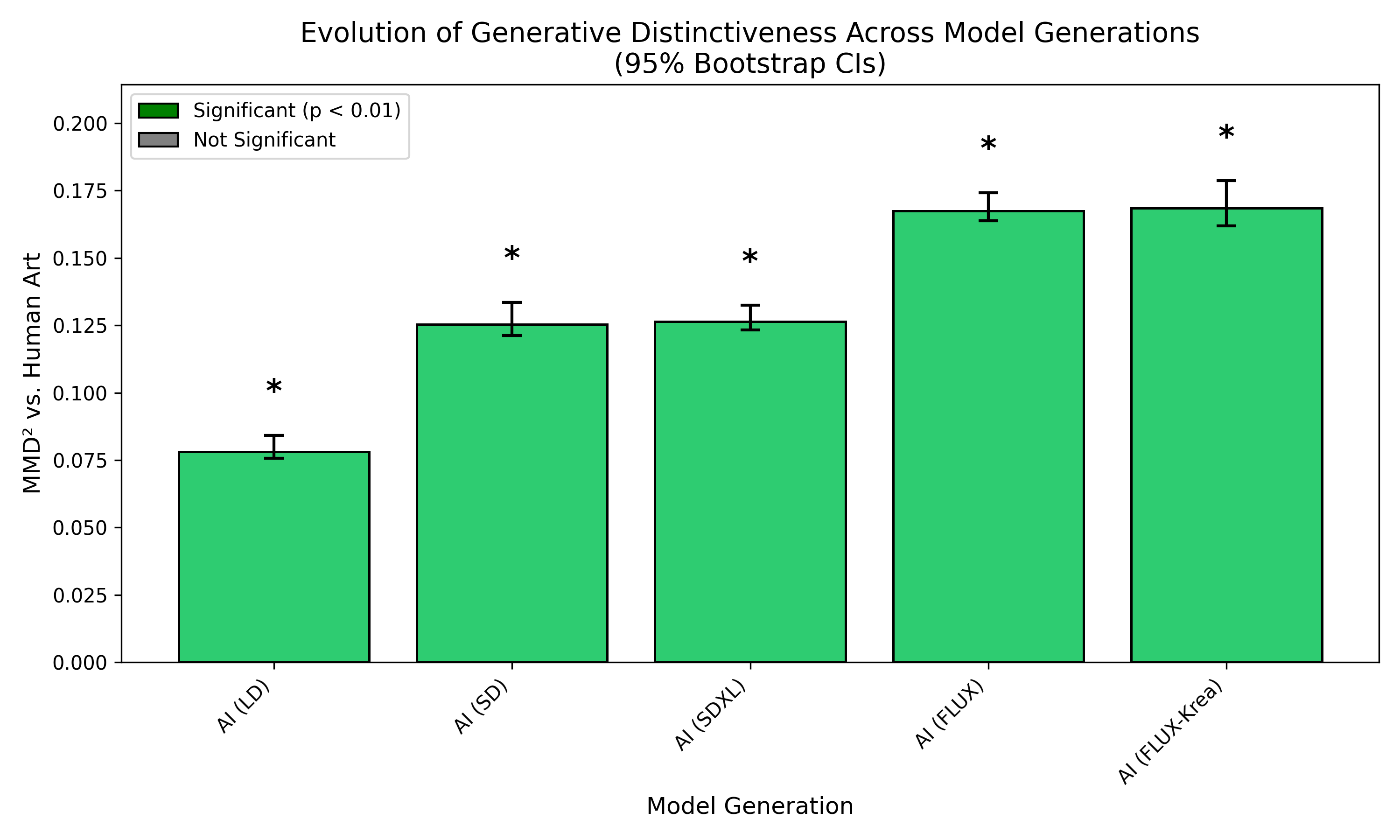}
\caption{Evolution of Human-AI Distributional Distinctiveness Across Model Generations (Pooled)}
\label{fig:art_evolution}
\begin{minipage}{0.95\linewidth}
\footnotesize
Note: $\text{MMD}^2$ between human artworks and outputs from each model generation, pooling across all ten artistic styles. Error bars indicate 95\% bootstrap confidence intervals (1,000 iterations). Models ordered chronologically: Latent Diffusion (Dec 2021), Stable Diffusion (Aug 2022), SDXL (Jul 2023), FLUX (Aug 2024), FLUX-Krea (Jul 2025). All comparisons significant at $p < 0.01$.
\end{minipage}
\end{figure}

Our results show that distributional distinctiveness does not vanish as
models advance---it \emph{increases}. Figure \ref{fig:art_evolution}
shows \(\text{MMD}^2\) values rising from 0.078 {[}0.076, 0.084{]} for
Latent Diffusion to 0.125 {[}0.121, 0.133{]} for Stable Diffusion to
0.126 {[}0.123, 0.132{]} for SDXL to 0.167 {[}0.164, 0.174{]} for FLUX
to 0.169 {[}0.162, 0.179{]} for FLUX-Krea---more than doubling despite
dramatic improvements in perceptual quality. The 95\% bootstrap
confidence intervals confirm that these differences are not artifacts of
sampling variability: the progression from LD to SD to FLUX represents
statistically reliable increases in distinctiveness, while FLUX and
FLUX-Krea are indistinguishable from each other (overlapping CIs).

\begin{figure}[htbp]
\centering
\includegraphics[width=\linewidth]{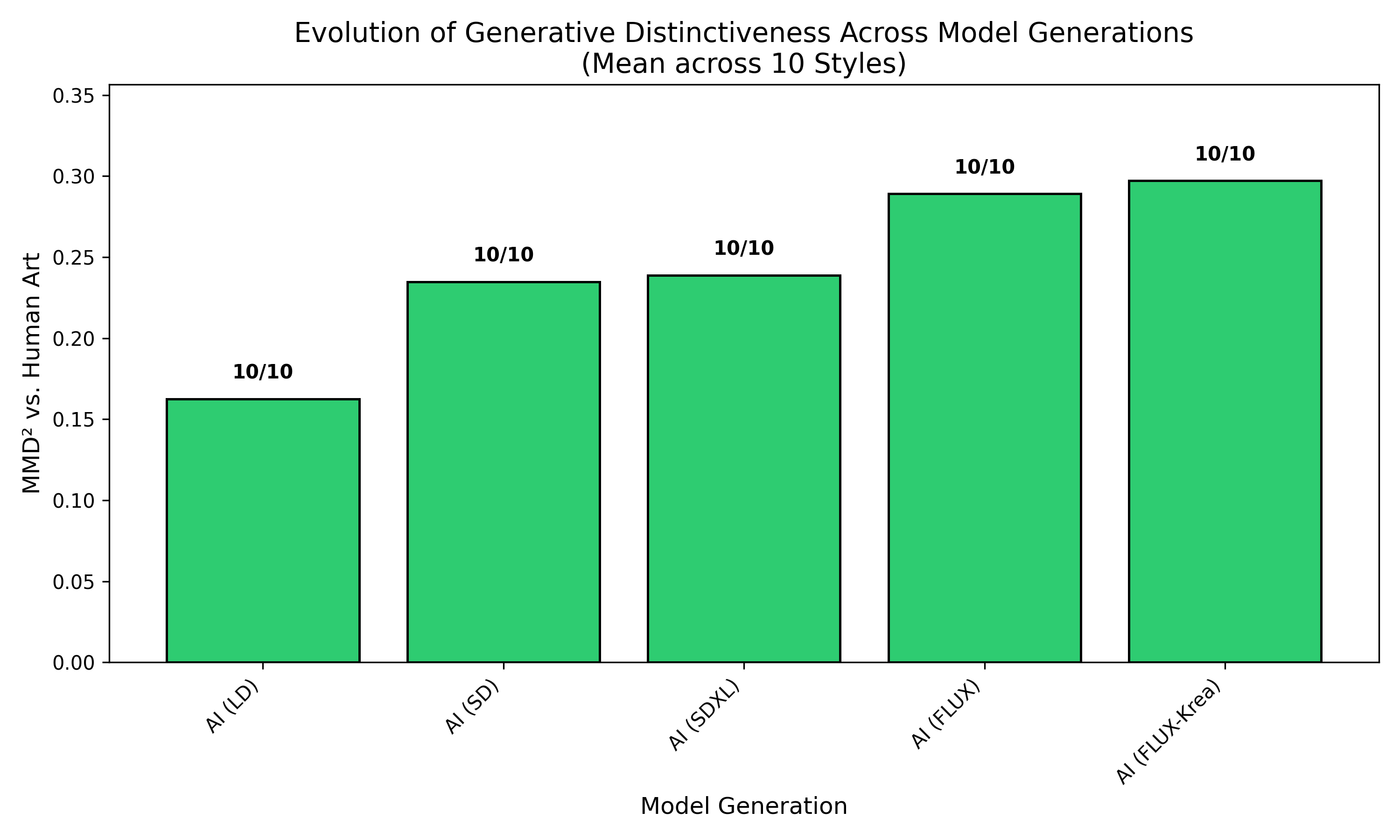}
\caption{Evolution of Human-AI Distributional Distinctiveness (Mean Across Styles)}
\label{fig:art_evolution_by_style}
\begin{minipage}{0.95\linewidth}
\footnotesize
Note: Mean $\text{MMD}^2$ computed separately within each of the ten artistic styles, then averaged. This within-style analysis controls for genre confounds and uses smaller per-style sample sizes, yielding higher absolute magnitudes than the pooled analysis. The monotonic increase across model generations is robust to this methodological variation.
\end{minipage}
\end{figure}

To confirm that this pattern is not driven by a subset of styles, we
also compute \(\text{MMD}^2\) within each of the ten artistic styles
separately and report the mean (Figure
\ref{fig:art_evolution_by_style}). The trajectory is qualitatively
identical: 0.16 (LD) \(\rightarrow\) 0.23 (SD) \(\rightarrow\) 0.24
(SDXL) \(\rightarrow\) 0.29 (FLUX) \(\rightarrow\) 0.30 (FLUX-Krea).
While absolute magnitudes differ from the pooled analysis (within-style
comparisons have smaller sample sizes, yielding higher variance and
higher point estimates), the monotonic increase is robust.

\subsubsection{Replication on a Single-Artist Corpus: Monet's Water
Lilies}\label{replication-on-a-single-artist-corpus-monets-water-lilies}

The AI-ArtBench analysis pools human artworks across many artists and
periods. To test whether our findings generalize to a more focused
setting---and to address the possibility that cross-artist heterogeneity
drives the results---we replicate the evolution analysis on a
single-artist corpus: Claude Monet's Water Lilies series.

We assembled a dataset of 200 authenticated Monet Water Lilies paintings
and generated 200 AI images from each of five model generations using
the prompt ``an oil painting of water lilies in the style of Claude
Monet.'' This design isolates the effect of generative modeling from
artistic diversity: all human works share a single artist, subject, and
style, while all AI works target the same stylistic goal.

\begin{figure}[htbp]
\centering
\includegraphics[width=0.6\linewidth]{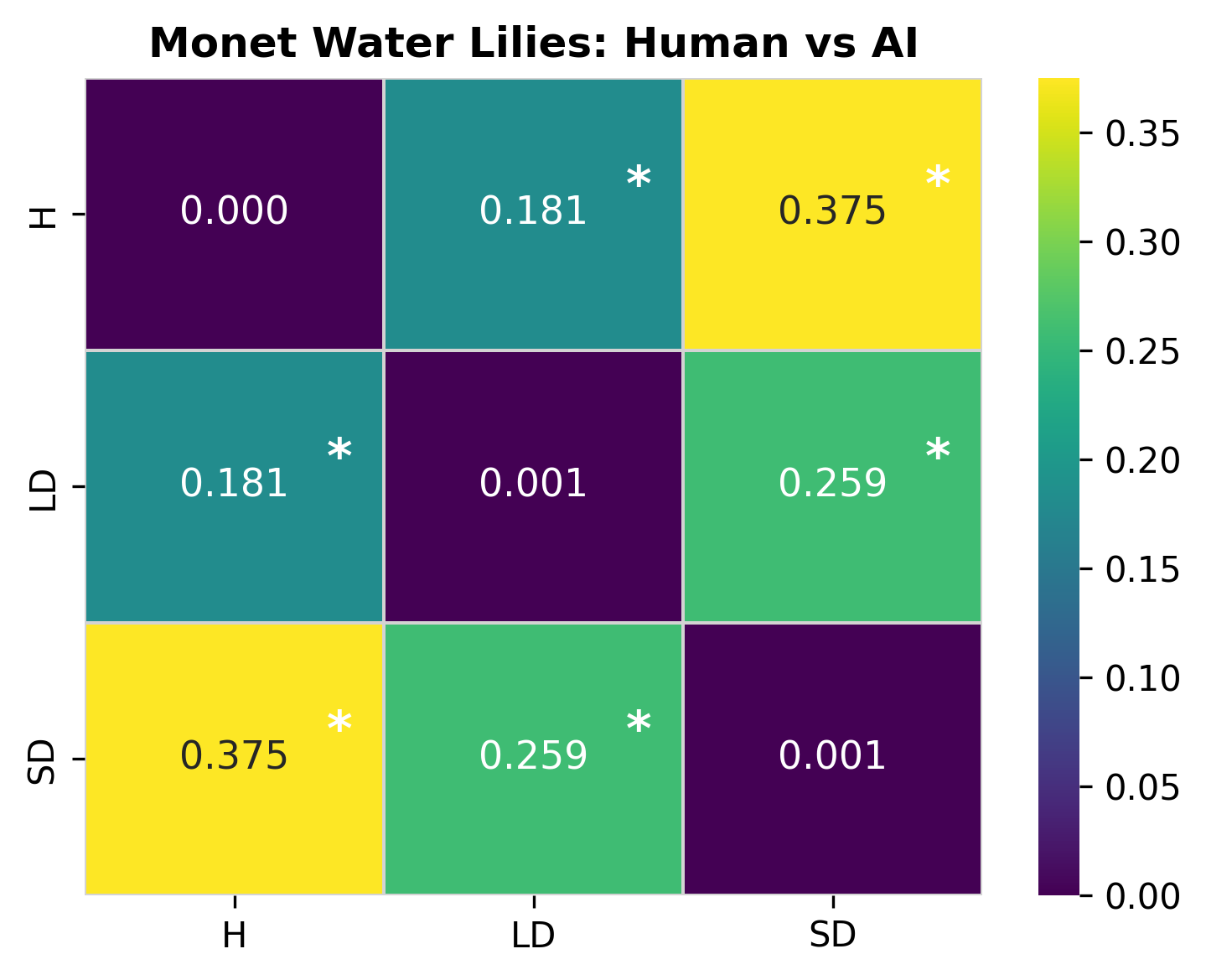}
\caption{Monet Water Lilies: MMD Heatmap with Negative Controls}
\label{fig:monet_heatmap}
\begin{minipage}{0.95\linewidth}
\footnotesize
Note: $\text{MMD}^2$ matrix comparing Human (Monet originals), AI (LD), and AI (SD). Diagonal elements (within-source comparisons) serve as negative controls; off-diagonal elements (cross-source comparisons) test for distributional distinctiveness. Asterisks indicate $p < 0.01$. The near-zero diagonal values confirm that the method does not detect spurious differences within homogeneous sources.
\end{minipage}
\end{figure}

Figure \ref{fig:monet_heatmap} presents the MMD heatmap for the Monet
corpus. Critically, the diagonal elements (within-source comparisons:
Human-A vs Human-B, AI-LD-A vs AI-LD-B, AI-SD-A vs AI-SD-B) yield
near-zero \(\text{MMD}^2\) values and are non-significant at
\(\alpha = 0.01\). This validates our negative controls: the method does
not detect differences where none exist. Meanwhile, all off-diagonal
(cross-source) comparisons are highly significant, with Human vs AI (SD)
yielding \(\text{MMD}^2 = 0.375\) and Human vs AI (LD) yielding
\(\text{MMD}^2 = 0.181\).

\begin{figure}[htbp]
\centering
\includegraphics[width=\linewidth]{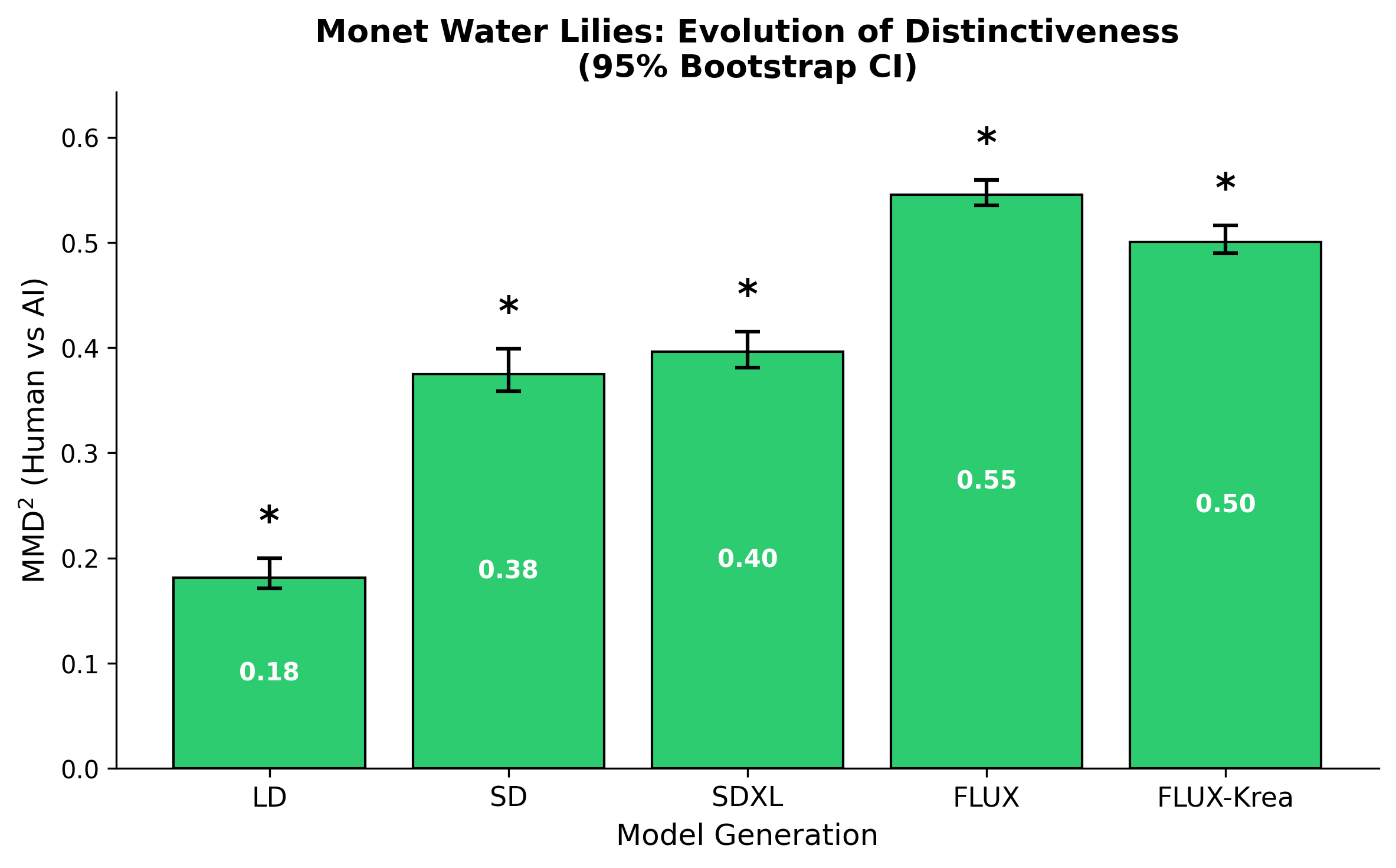}
\caption{Monet Water Lilies: Evolution of Distributional Distinctiveness}
\label{fig:monet_evolution}
\begin{minipage}{0.95\linewidth}
\footnotesize
Note: $\text{MMD}^2$ between Monet's Water Lilies paintings and AI-generated imitations across five model generations. Error bars indicate 95\% bootstrap confidence intervals (1,000 iterations). All comparisons significant at $p < 0.01$.
\end{minipage}
\end{figure}

The evolution analysis on this single-artist corpus confirms and
amplifies the main finding (Figure \ref{fig:monet_evolution}).
\(\text{MMD}^2\) rises from 0.181 {[}0.171, 0.200{]} for Latent
Diffusion to 0.375 {[}0.359, 0.399{]} for Stable Diffusion to 0.396
{[}0.381, 0.415{]} for SDXL to 0.545 {[}0.535, 0.559{]} for FLUX---more
than tripling. The confidence intervals do not overlap between adjacent
model generations (except SD and SDXL), confirming that the increases
are statistically reliable.

Notably, FLUX-Krea yields \(\text{MMD}^2 = 0.501\) {[}0.490, 0.516{]},
which is reliably \emph{lower} than vanilla FLUX (CIs do not overlap).
This finding is unique to the Monet corpus and was not observed in the
pooled AI-ArtBench analysis, where FLUX and FLUX-Krea showed comparable
distinctiveness. The difference likely reflects FLUX-Krea's design: it
was fine-tuned on curated aesthetic datasets to better capture artistic
styles.\footnote{\emph{See}
  \url{https://www.krea.ai/blog/flux-krea-open-source-release}.} When
the target is a specific artistic tradition (Monet's Impressionism),
this fine-tuning appears to reduce---but not eliminate---distributional
distinctiveness: the fine-tuned model remains highly distinguishable
from Monet's originals with \(\text{MMD}^2\) values exceeding those of
earlier model generations.

This Monet replication provides several important insights. First, the
increasing distinctiveness over model generations is not an artifact of
cross-artist heterogeneity in the human reference set; it persists when
the reference is a single artist's oeuvre. Second, the validated
negative controls (diagonal near-zero) confirm that the method's
specificity holds in a focused setting. Third, the FLUX-Krea finding
suggests that targeted fine-tuning can partially close the gap with
specific artistic traditions---a finding with potential implications for
copyright's substantial similarity analysis.

\subsection{The Perceptual Paradox}\label{the-perceptual-paradox}

Our central finding is that AI models exhibit ``interpolative
distinctiveness'': they produce outputs that are semantically novel yet
perceptually familiar. To probe the nature of this distinctiveness, we
compare MMD scores across three embedding spaces that differ
systematically in their relationship to AI-generated imagery:

\begin{enumerate}
\def\labelenumi{\arabic{enumi}.}
\item
  \textbf{CLIP (ViT-H-14):} A contrastive vision-language model trained
  on dfn5b, a curated dataset of image-text pairs scraped from the web
  prior to the widespread deployment of generative AI. This embedding
  space is \emph{naive} to AI-generated content---it represents a purely
  human-centric semantic topology.
\item
  \textbf{DreamSim:} A perceptual similarity metric explicitly trained
  on synthetic images generated by text-to-image models, then calibrated
  to predict human perceptual similarity judgments
  (\citeproc{ref-fu2023dreamsim}{Fu et al. 2023}). This embedding is
  \emph{aware} of AI-generated imagery but is optimized to align with
  human perception, placing it at an intermediate position between
  semantic and generative representations.
\item
  \textbf{Stable Diffusion VAE:} The variational autoencoder from Stable
  Diffusion, which compresses images into the \(64 \times 64 \times 4\)
  latent space where the diffusion process operates
  (\citeproc{ref-rombach2022high}{Rombach et al. 2022}). This is the
  \emph{native generative space}---the representation in which AI
  outputs are literally constructed. By design, successful generation
  requires AI outputs to occupy the same latent manifold as human art.
\end{enumerate}

These three embeddings form a spectrum from ``AI-naive'' (CLIP) through
``AI-aware but human-calibrated'' (DreamSim) to ``AI-native'' (VAE). If
distinctiveness is a genuine property of AI outputs rather than an
artifact of the measurement space, all three embeddings should detect
significant differences. However, the optimization objective of each
embedding suggests that the \emph{magnitude} should vary systematically.
Because diffusion models operate entirely within their VAE latent
space---and would fail to generate coherent outputs if they strayed from
it---we expect this embedding to show the greatest distributional
overlap (lowest MMD). Conversely, CLIP represents the semantic manifold
of human creation; if AI outputs are indeed
``interpolative''---occupying sparse regions of the creative space
between human concepts---they should appear as outliers to an AI-naive
observer, yielding the highest MMD. DreamSim, by penalizing differences
that humans ignore, should dampen this semantic signal, producing
intermediate distinctiveness.

\begin{figure}[htbp]
\centering
\includegraphics[width=\linewidth]{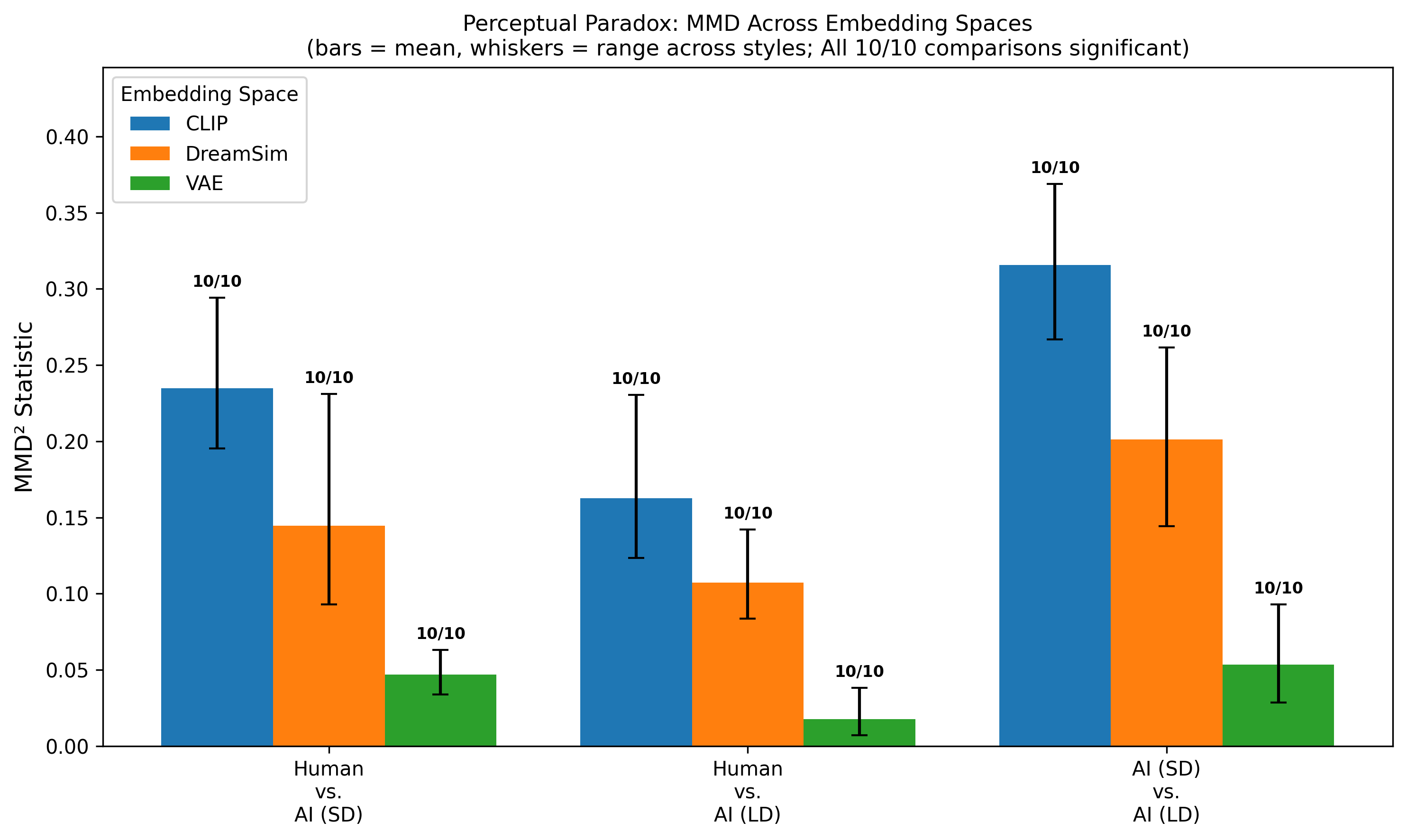}
\caption{MMD Across Embedding Spaces: The Embedding Awareness Spectrum}
\label{fig:art_embedding_comparison}
\begin{minipage}{0.95\linewidth}
\footnotesize
Note: Bars show mean $\text{MMD}^2$ across 10 artistic styles; error bars indicate the range (min–max) across styles. All 90 comparisons (3 embeddings \(\times\) 3 category pairs \(\times\) 10 styles) achieved statistical significance (\(p < 0.01\)), as indicated by the "10/10" labels. 
\end{minipage}
\end{figure}

The empirical results confirm this prediction with striking consistency
(Figure \ref{fig:art_embedding_comparison}). Averaging across all 10
artistic styles, the Human vs.~AI (SD) comparison yields mean
\(\text{MMD}^2 = 0.235\) for CLIP, \(0.145\) for DreamSim, and \(0.047\)
for VAE---a monotonic decrease along the AI-awareness spectrum. The same
ordering holds for Human vs.~AI (LD): CLIP (0.163) \textgreater{}
DreamSim (0.107) \textgreater{} VAE (0.018). Yet despite this variation
in magnitude, \emph{all} 90 comparisons achieve statistical significance
(\(p < 0.01\)), with most reaching \(p = 0.002\)---the minimum possible
with 500 permutations. This uniformity arises because the permutation
test operates \emph{within} each embedding space: it asks whether the
observed distance is extreme relative to chance variation in that same
representation, not whether it exceeds some absolute threshold. The
embedding choice thus affects the \emph{magnitude} of measured
distinctiveness but not whether it is \emph{detectable}. Even in VAE
space, where the generative model is designed to produce outputs that
fit the latent manifold, the distributional signature remains
statistically unambiguous across all styles.

This gradient of distances illuminates the nature of AI creativity. The
fact that distinctiveness is \emph{largest} in CLIP---an embedding
trained exclusively on human-created imagery---suggests that AI outputs
are genuine outliers in the semantic space of human creativity. They
occupy regions that human artists have not densely populated. The fact
that distinctiveness is \emph{smallest} in the VAE---the space where AI
outputs are constructed---reflects the design objective of generative
models: to produce samples that are distributionally similar to the
training data in latent space. DreamSim's intermediate position confirms
that when calibrated to human perception, AI outputs appear more similar
to human art than pure semantic analysis would suggest, yet remain
distinguishable.

For copyright's substantial similarity analysis, this finding cuts both
ways. The robust semantic distinctiveness (CLIP) supports claims of
independent creation---AI outputs are not mere copies but occupy
distinct regions of creative space. Yet the reduced distinctiveness in
perceptually-calibrated space (DreamSim) explains why human observers
struggle to distinguish AI from human art: the features humans attend
to---foreground objects, composition, color, and layout
(\citeproc{ref-fu2023dreamsim}{Fu et al. 2023})---are precisely those
where generative models excel at mimicry. The MMD framework reveals what
human perception misses: a systematic distributional signature that
persists across representational choices.

\section{Discussion}\label{discussion}

Our successful validation across three distinct domains---handwritten
digits, patent abstracts, and AI-generated art---demonstrates the
method's versatility across modalities. Substantively, these findings
provide evidence that generative models do not operate as mere
regurgitators of their training data. By demonstrating that AI-generated
art is statistically distinct in distribution from human-created art, we
identify a phenomenon of ``interpolative distinctiveness'' that is
difficult to reconcile with mechanical regurgitation. If generative
models were merely resampling training data, their output distributions
would be expected to converge to the human baseline; instead, they
diverge significantly. Moreover, the magnitude of this distinctiveness
is probative: it suggests the divergence represents a degree of semantic
separation comparable to the `transformative' leaps recognized between
distinct genres of human creativity.

Crucially, this statistical separation occurs even where human
perception fails: the same AI outputs that are distributionally distinct
are often perceptually indistinguishable from human art. This
conjunction is the crux of the evidence against the regurgitation
hypothesis. Distributional distinctiveness alone does not suffice; a
random noise generator would be statistically distinct from human art
yet devoid of expression. Perceptual indistinguishability alone does not
suffice; it could reflect the successful compression and reproduction of
training data. However, the observation that AI outputs are
simultaneously distributionally novel yet semantically human-like
suggests a mechanism of interpolative creativity. The models are not
merely retrieving existing data points but are recombining learned
patterns to generate outputs that occupy a distinct, yet semantically
coherent, topological region of the creative space.

The shift from item-level to process-level analysis resolves the
``infinite cardinality'' problem inherent in generative AI. While
traditional doctrines rely on pairwise comparisons mediated by human
proxies, such methods cannot scale to unbounded output spaces. By
measuring the distance between distributions, our framework allows
courts and IP offices to evaluate the generative process itself.
Moreover, because the method is training-free and highly
sample-efficient, it aligns with the evidentiary realities of
litigation, enabling rigorous assessment even when proprietary model
weights or massive training sets are inaccessible (Table
\ref{tab:sample_efficiency_summary}).

\begin{table}[htbp]
\centering
\caption{Summary of Data Efficiency: Minimum Sample Sizes for 95\% Statistical Power ($\alpha=0.01$)}
\label{tab:sample_efficiency_summary}
\begin{tabular}{llcc}
\toprule
\textbf{Domain} & \textbf{Comparison Type} & \textbf{Typical $\text{MMD}^2$} & \textbf{$N$ (per group)} \\
\midrule
MNIST (Images) & All digit pairs & $0.81$--$1.22$ & 5--6 \\
\midrule
\multirow{2}{*}{Patents (Text)} & Distinct Fields (C vs. H) & $0.72$ & 7 \\
 & Related Fields (A vs. C) & $0.37$ & 15 \\
\midrule
\multirow{2}{*}{AI Art (Images)} & High Divergence (Expressionism) & $0.88$ & 6 \\
 & Low Divergence (Art Nouveau) & $0.79$ & 6 \\
\bottomrule
\end{tabular}
\begin{minipage}{0.95\linewidth}
\vspace{0.5em}
\footnotesize
Note: $N$ = samples per group to reject the null with probability $>0.95$. Patent sections: A = Human Necessities, C = Chemistry, H = Electricity. MNIST digit pairs all exhibit high distinctiveness; the range reflects the full off-diagonal MMD matrix.
\end{minipage}
\end{table}

This formalization transforms distinctiveness from a subjective gestalt
impression into a form of testable empirical evidence characterized by
known error rates---key factors in admissibility analyses such as
\emph{Daubert}.\footnote{Daubert v. Merrell Dow Pharms., Inc., 509 U.S.
  579, 593--94 (1993) (establishing that the admissibility of scientific
  evidence depends on factors including testability and known error
  rates).} Our framework satisfies these criteria through permutation
testing, which provides exact Type I error control independent of the
underlying data distribution. This allows the fact-finder to select a
significance level (\(\alpha\)) that explicitly reflects the tolerance
for false positives. Rather than relying on opaque intuition, this
parameter allows for a transparent calibration of evidentiary certainty
that can be justified in relation to the applicable standard of proof.
By converting raw semantic distances into probabilistic inferences, the
method integrates computational rigor with the procedural demands of
legal adjudication.

These capabilities offer specific utility across intellectual property
doctrines. In trademark, the method can quantify the distinctiveness of
a brand's visual footprint, providing evidence regarding the distance of
a generative branding process from existing commercial
symbols.\footnote{\emph{See Abercrombie}, 537 F.2d at 9.} In copyright,
it provides objective evidence of ``probabilistic originality,''
supporting the ``independent creation'' prong of the \emph{Feist}
test\footnote{Feist Publ'ns, Inc.~v. Rural Tel. Serv. Co., 499 U.S. 340,
  345 (1991) (holding that originality requires independent creation and
  a ``modicum of creativity'').} by demonstrating that a work is not a
mechanical reproduction of the training data.\footnote{To the extent
  that copyright is extended to non-natural persons.} In patent, it
offers quantitative evidence relevant to non-obviousness, measuring the
``generative distance'' of the claimed generative process from the prior
art.\footnote{\emph{See Graham}, 383 U.S. at 17.} Here, we emphasize
that this metric operationalizes the \emph{empirical} inquiry---the
magnitude of difference---not the \emph{normative} legal conclusion.
While MMD provides the factual predicate, the selection of the relevant
reference class, the determination of the requisite degree of
difference, amongst other issues, remain matters of legal argumentation.
Nothing in the method requires style-level grouping; in a concrete
copyright dispute, the relevant reference class may be as narrow as a
plaintiff's catalog or an artist-period corpus, and our framework
applies at that granularity (see Section 6.6 and Table 8).

This distinction between item-level and process-level distinctiveness is
critical for remedies. A generative process that is distributionally
distinct (high MMD) but produces rare instances of memorization presents
a different legal harm than a process that is distributionally
indistinguishable from its training data. The former suggests a tool
with specific defects---analogous to a printing press that occasionally
produces a defective copy---counseling toward targeted damages for the
specific infringing outputs. The latter suggests a market substitute
that systematically competes with the original works, supporting broader
injunctive relief or disgorgement. Thus, while the memorization audit
identifies \emph{whether} infringement occurs, the MMD magnitude informs
the \emph{extent} and \emph{nature} of the remedy. This bifurcated
framework aligns with copyright's remedial structure, which
distinguishes between actual damages (tied to specific harm) and
statutory damages or profits (tied to systemic conduct).

We acknowledge important limitations. The method's sensitivity depends
on embedding quality; while CLIP is robust for visual art, other domains
may require specialized feature extractors. Furthermore, the validity of
the inference relies on representative sampling. In adversarial
litigation, parties may attempt to ``cherry-pick'' prompts or outputs to
skew the distribution; consequently, the utility of this metric in court
depends on the enforcement of rigorous sampling protocols during
discovery to ensure the analyzed portfolios accurately reflect the
generative process. Similarly, embeddings such as CLIP inherit biases
from their training corpora. CLIP was trained predominantly on Western
internet imagery, potentially underrepresenting artistic traditions from
other cultures. While our Perceptual Paradox analysis (Section 6.7)
demonstrates that conclusions are robust across three architecturally
distinct embeddings---CLIP, DreamSim, and the Stable Diffusion
VAE---this does not eliminate concerns about cultural representation.
Researchers applying this framework should select embeddings that
counter biases relevant to their application context; when cultural
representation is paramount, embeddings trained on more diverse corpora
may be appropriate. We recommend triangulating across multiple
complementary embeddings to ensure robustness. Finally, while a
significant MMD score establishes that a model's dominant mode is not
regurgitation, it is a macro-level metric. It does not---and is not
designed to---rule out the possibility that specific individual outputs
may reproduce training examples. Therefore, we advocate for a bifurcated
evidentiary framework: MMD to adjudicate the nature of the generative
process, and targeted audits to detect specific instances of potential
infringement.

These limitations suggest directions for future research. Establishing
domain-specific MMD thresholds that correspond to specific legal
standards of proof would enhance practical applicability. As human-AI
collaboration becomes the norm, methods should be developed to
disentangle machine contributions from human creative input, as the
current method evaluates only the composite output. Our rejection of the
regurgitation hypothesis relies on the assumption that the human
artworks in our dataset are representative of the relevant prior art and
the employed embeddings do not yield artefacts correlated with the
central inquiry. Future work should investigate the extent to which
various biases, such as cultural and representational biases, affect MMD
statistics, and develop principled strategies for embedding selection
and bias mitigation in legal applications. Comparing AI outputs directly
to other examples of training data and broader samples of human
creativity would help confirm the robustness of our key findings of a
perceptual paradox in AI outputs.

As the provenance of creative works grows increasingly uncertain,
intellectual property law must evolve from intuition-based tests toward
principled quantitative frameworks. By grounding assessments in
distributional evidence, our methodology offers a step toward ensuring
that the law reflects technological realities while preserving its
fundamental purpose of incentivizing innovation and creative expression.

\newpage

\section*{Bibliography}\label{bibliography}
\addcontentsline{toc}{section}{Bibliography}

\singlespacing

\phantomsection\label{refs}

\newpage

\setcounter{section}{0}
\renewcommand{\thesection}{\Alph{section}}
\renewcommand{\thesubsection}{\thesection.\arabic{subsection}}
\renewcommand{\thesubsubsection}{\thesubsection.\arabic{subsubsection}}
\renewcommand{\theparagraph}{\thesubsubsection.\arabic{paragraph}}
\renewcommand{\thesubparagraph}{\theparagraph.\arabic{subparagraph}}
\renewcommand{\thetable}{A\arabic{table}}
\renewcommand{\thefigure}{A\arabic{figure}}

\section{Appendix: Implementation
Details}\label{appendix-implementation-details}

This appendix documents a Python implementation of the distributional
distinctiveness framework developed in this paper. The implementation is
designed for reproducibility, extensibility, and clarity. It can serve
as a foundation for researchers adapting the methodology to new domains
(e.g., trademark visual distinctiveness, literary style analysis, music
generation) or practitioners deploying it in legal contexts.

\textbf{Repository:} During peer review, the complete codebase is
available at: \url{https://osf.io/\%5BINSERT_PLACEHOLDER_LINK\%5D}. Upon
acceptance, this will be permanently archived with a DOI.

\subsection{Architecture Overview}\label{architecture-overview}

The codebase follows a modular architecture with a domain-agnostic
statistical engine (domain-agnostic) separated from domain-specific
pipelines (data loading, embedding extraction, visualization):

\begin{verbatim}
Module 1: Shared Statistical Core
    ├── MMD computation (unbiased estimator)
    ├── Permutation test (non-parametric significance)
    └── Ablation engines (kernel, dimensionality, bandwidth, stability)

Module 2: MNIST Pipeline (validation)
Module 3: Patent Pipeline (text domain)
Module 4: AI Art Pipeline (image domain)
    ├── Category analysis (per-style MMD)
    ├── Evolution analysis (model generation trajectory)
    ├── Perceptual paradox (CLIP vs DreamSim vs VAE)
    ├── Memorization audit (item-level nearest-neighbor analysis)
    └── Robustness suite (ablations, perturbations)
Module 5: Exposition Summaries (formatted output)
Module 6: Main Execution (orchestration, configuration)
\end{verbatim}

Extending to a new domain requires: (1) a data loader returning samples
with labels, (2) an embedding extractor producing numerical vectors, and
(3) calls to the shared statistical functions. The core MMD and
permutation test logic remains unchanged.

\subsection{Dependencies}\label{dependencies}

\begin{verbatim}
# Core (required)
numpy, scipy, scikit-learn, umap-learn, matplotlib, seaborn, tqdm, joblib, pillow

# Deep Learning
torch              # All studies (LeNet for MNIST, CLIP for AI Art)
torchvision        # Data augmentation transforms, MNIST data loading
open_clip_torch    # CLIP embeddings for images
sentence_transformers  # SBERT embeddings for text

# Data
datasets           # Hugging Face datasets (Patent study)

# Optional (for extended AI Art analyses)
diffusers          # VAE latent extraction
dreamsim           # Perceptual embeddings
scikit-image       # SSIM computation for memorization audit
lpips              # Learned perceptual similarity for memorization audit

# Performance
numexpr            # Fast array expression evaluation
psutil             # Memory checking for precomputation decisions
\end{verbatim}

\subsection{Module 1: Shared Statistical
Core}\label{module-1-shared-statistical-core}

This module implements the theoretical framework from Section 3. All
functions operate on numerical vectors (embeddings) and are agnostic to
the underlying data type.

\subsubsection{Core MMD Functions}\label{core-mmd-functions}

\begin{itemize}
\tightlist
\item
  \texttt{\_compute\_sigma\_median\_heuristic(x,\ y)\ →\ float}:
  Computes the bandwidth parameter \(\sigma\) for the RBF kernel using
  the median heuristic: \(\sigma\) = median of all pairwise Euclidean
  distances in the combined sample. This data-driven approach
  automatically scales the kernel's sensitivity to the data's intrinsic
  dimensionality, avoiding manual tuning. The median heuristic is
  preferred over cross-validation because it is (a) deterministic, (b)
  computationally cheap, and (c) provides reasonable performance across
  diverse datasets without risk of overfitting to a specific comparison.
\item
  \texttt{mmd\_squared\_unbiased(x,\ y,\ kernel=\textquotesingle{}rbf\textquotesingle{},\ gamma=None)\ →\ float}:
  Computes the unbiased squared MMD estimator that has expected value
  zero under the null hypothesis regardless of sample size, enabling
  valid comparison across different n.~Uses
  \texttt{scipy.spatial.distance.pdist}/\texttt{cdist} for efficient
  pairwise distance computation (computing only unique pairs) combined
  with \texttt{numexpr} for fast kernel evaluation. Supports both RBF
  kernel (default, captures all moments) and linear kernel (captures
  only mean difference).
\item
  \texttt{permutation\_test(x,\ y,\ R,\ alpha,\ ...)\ →\ (p\_value,\ reject,\ lb,\ ub)}:
  Implements Algorithm 1, the non-parametric permutation test for MMD
  significance. Under \(H_0\) (identical distributions), pooling and
  reshuffling samples should produce MMD values similar to the observed
  value; if the observed MMD is extreme relative to the permutation
  distribution, we reject \(H_0\). \texttt{R} is then number of
  permutation iterations. With R=500, minimum attainable p-value is
  1/501 \(\sim\) 0.002, sufficient for \(\alpha\)=0.01. If
  \texttt{precompute} is True, the function computes the full kernel
  matrix once before permutations. This dramatically improves
  performance for large samples but requires O(n²) memory. The function
  automatically falls back to on-the-fly computation if memory is
  insufficient. \texttt{n\_jobs} parallelizes permutation iterations
  using \texttt{joblib} with thread-based backend (avoids memory copying
  overhead). The function uses the conservative Monte Carlo estimator
  (r+1)/(R+1) where r is the count of permuted statistics \(\geq\)
  observed. This avoids zero \emph{p}-values and is the standard
  approach in permutation testing literature.
\end{itemize}

\subsubsection{Shared Ablation
Functions}\label{shared-ablation-functions}

These functions implement domain-agnostic robustness analyses, ensuring
methodological consistency across studies.

\begin{itemize}
\tightlist
\item
  \texttt{compute\_stability\_analysis(x,\ y,\ dims\_list,\ n\_trials,\ sample\_size)\ →\ dict}:
  Implements Section 4.3: evaluates how UMAP compression error
  propagates to the MMD statistic. For each target dimension d, computes
  MMD in both full and reduced spaces across multiple trials, returning
  the mean and standard deviation of the absolute deviation. UMAP is fit
  on the pooled sample (X \(\cup\) Y), not separately on each group.
  Fitting separately would artificially inflate differences by allowing
  each group to find its own optimal subspace.
\item
  \texttt{compute\_kernel\_ablation(x,\ y,\ sample\_sizes,\ n\_trials,\ R,\ alpha,\ n\_jobs)\ →\ dict}:
  Compares RBF vs.~Linear kernel rejection rates. The Linear kernel MMD
  measures only centroid distance; the RBF kernel captures higher-order
  distributional structure. If both kernels reject at similar rates,
  distinctiveness is driven by mean shift; if RBF substantially
  outperforms Linear, the distributions differ in variance or shape.
\item
  \texttt{compute\_dimensionality\_ablation(x,\ y,\ dims\_list,\ sample\_sizes,\ n\_trials,\ R,\ alpha,\ n\_jobs)\ →\ dict}:
  Tests whether the statistical conclusion (reject/fail-to-reject) is
  stable under aggressive dimensionality reduction. Uses paired
  sampling: for each (n, trial), the same sample is tested across all
  dimensions, enabling clean comparison. UMAP requires the target
  dimension to be strictly less than the pooled sample size
  (\(d < 2n\)). When this constraint is violated, the function returns
  NaN rather than failing, allowing partial results. For example,
  reducing to \(d=32\) dimensions requires at least \(n=17\) samples per
  group (34 total).
\item
  \texttt{compute\_bandwidth\_ablation(x,\ y,\ sigma\_multipliers,\ sample\_sizes,\ ...)\ →\ dict}:
  Tests robustness to the bandwidth hyperparameter by scaling the
  median-heuristic \(\sigma\) by specified multipliers (e.g., 0.5×,
  1.0×, 2.0×). A ``plateau of significance'' (high rejection across all
  multipliers) indicates the result is not an artifact of specific
  tuning. The baseline \(\sigma\) is estimated from a larger fixed
  sample (default n=200) independent of the test sample size, reducing
  noise in the bandwidth estimate.
\end{itemize}

\subsection{Module 2: MNIST Study}\label{module-2-mnist-study}

Validates the framework using MNIST digits where ground truth is known
(e.g., digit 3 \(\neq\) digit 5). For robustness analyses, the digit
pair with the lowest off-diagonal \(\text{MMD}^2\) (hardest to
distinguish) is automatically selected at runtime, providing the most
conservative test. This module demonstrates the methodology in a
controlled setting before applying it to legally relevant domains. The
implementation uses PyTorch for neural network training and embedding
extraction.

\subsubsection{Model Architecture}\label{model-architecture}

The \texttt{LeNet5} class implements the classic convolutional
architecture as a PyTorch \texttt{nn.Module}:

\begin{verbatim}
Input (1×28×28) → Conv(6, 5×5) → AvgPool(2×2) → Conv(16, 5×5) → AvgPool(2×2) →
FC(256→120) → Dropout(0.1) → FC(120→84) [embedding] → Dropout(0.1) → FC(84→10)
\end{verbatim}

The 84-dimensional embedding layer (fc2) serves as the feature
representation for MMD analysis. Embeddings are extracted via the
model's \texttt{forward(x,\ extract\_embeddings=True)} method, which
returns activations before the final classification layer.

\subsubsection{Data Augmentation}\label{data-augmentation}

The \texttt{AugmentedDataset} class wraps PyTorch's \texttt{Dataset} to
apply transforms per-sample rather than to batched tensors. This is
necessary because torchvision applies identical random transforms to all
images in a batch; per-sample application ensures independent
randomization. Augmentation uses a single \texttt{RandomAffine}
transform combining rotation (±10°), translation (±10\%), and scaling
(0.9×--1.1×) in one interpolation step, reducing artifacts compared to
chained transforms.

\subsubsection{Core Functions}\label{core-functions}

\begin{itemize}
\tightlist
\item
  \texttt{mnist\_load\_and\_prepare\_data()}: Loads MNIST via
  torchvision, normalizes pixels to {[}0,1{]}, reshapes for CNN input.
\item
  \texttt{mnist\_build\_lenet5\_model()}: Factory function returning a
  LeNet5 \texttt{nn.Module} instance with 84-dimensional embedding
  layer.
\item
  \texttt{mnist\_train\_model(...)}: Trains with data augmentation,
  early stopping, and checkpointing. Uses Adam optimizer with
  cross-entropy loss. Automatic device selection prioritizes CUDA
  (NVIDIA GPU), then MPS (Apple Silicon), then CPU.
\item
  \texttt{mnist\_extract\_embeddings(...)}: Extracts 84-dim vectors
  (float32 precision for GPU/MPS compatibility) from the trained model's
  penultimate layer with the model in evaluation mode (dropout
  disabled).
\item
  \texttt{mnist\_compute\_mmd\_matrix(...)}: Computes 10×10 MMD matrix
  for all digit pairs. Diagonal entries use split-half negative control:
  each digit class is split into two disjoint halves, and MMD is
  computed between them. Under \(H_0\), these should yield
  non-significant results.
\item
  \texttt{mnist\_compute\_rejection\_rates(...)}: Estimates statistical
  power across sample sizes via Monte Carlo trials.
\end{itemize}

\subsubsection{Visualization}\label{visualization}

\begin{itemize}
\tightlist
\item
  \texttt{mnist\_plot\_mmd\_heatmap(...)}: Heatmap of MMD values with
  significance markers. Negative values (possible due to unbiased
  estimator variance) are clamped to zero for visualization only.
\item
  \texttt{mnist\_plot\_rejection\_rates(...)}: Line plot of rejection
  rate vs.~sample size.
\item
  \texttt{mnist\_plot\_stability\_curve(...)}: Plots approximation error
  vs.~UMAP dimensions (Section 4.3).
\end{itemize}

\subsubsection{Robustness Suite}\label{robustness-suite}

\begin{itemize}
\tightlist
\item
  \texttt{mnist\_run\_robustness\_suite(...)}: Orchestrates all
  ablations (stability, kernel, dimensionality, representation,
  bandwidth, perturbation) and saves results. Generates
  \texttt{mnist\_perturbation\_results.pkl} (serialized perturbation
  data) and \texttt{mnist\_ablation\_tables.txt} (formatted tables for
  all ablation results).
\item
  \texttt{mnist\_compute\_representation\_ablation(...)}: Compares raw
  784-dim pixels vs.~learned 84-dim embeddings. For structurally simple
  objects like digits, raw pixels should work comparably---validating
  that MMD functions correctly even without sophisticated embeddings.
\item
  \texttt{mnist\_perturbation\_analysis(...)}: Tests robustness to
  Gaussian noise (parameterized by SNR) and grid watermarks
  (parameterized by SWR) using a paired sample design (the same images
  are used for clean and perturbed conditions, projected jointly via
  UMAP).
\item
  \texttt{mnist\_perturbation\_table\_combined(...)}: Generates a
  consolidated perturbation table showing \(\Delta\)MMD and
  \emph{p}-values across multiple digits (used in manuscript Table 2).
\item
  \texttt{mnist\_save\_ablation\_tables(...)}: Prints and saves
  formatted tables for all ablation results (kernel, bandwidth,
  dimensionality, representation, perturbation). Tables are both printed
  to stdout for monitoring and saved to
  \texttt{mnist\_ablation\_tables.txt} for reference during manuscript
  writing.
\end{itemize}

\subsection{Module 3: Patent Study}\label{module-3-patent-study}

Validates the framework on text using patent abstracts from different
IPC sections. Demonstrates that the methodology generalizes beyond
images to semantic text embeddings.

\subsubsection{Core Functions}\label{core-functions-1}

\begin{itemize}
\tightlist
\item
  \texttt{patent\_load\_dataset(...)}: Loads from
  \texttt{ccdv/patent-classification} via Hugging Face. Filters for
  sections A (Human Necessities), C (Chemistry), H (Electricity). Cleans
  text using regex to remove explicit class labels that could cause
  leakage. Collects 2× samples to enable split-half negative controls.
\item
  \texttt{patent\_extract\_embeddings(...)}: Extracts 384-dim embeddings
  using SentenceTransformers (\texttt{GIST-small-Embedding-v0}, with
  \texttt{all-MiniLM-L6-v2} as fallback). Embeddings are L2-normalized
  for RBF kernel stability.
\item
  \texttt{patent\_compute\_mmd\_matrix(...)} and
  \texttt{patent\_compute\_rejection\_rates(...)}: Standard MMD analysis
  with split-half negative controls on diagonal.
\end{itemize}

\subsubsection{Visualization and Output}\label{visualization-and-output}

\begin{itemize}
\tightlist
\item
  \texttt{patent\_plot\_mmd\_heatmap(...)} and
  \texttt{patent\_plot\_rejection\_rates(...)}: Analogous to MNIST
  visualizations.
\item
  \texttt{patent\_save\_summary\_tables(...)}: Prints and saves
  formatted summary tables (MMD matrix, \emph{p}-values, rejection
  rates) to \texttt{patent\_summary\_tables.txt}.
\end{itemize}

\subsection{Module 4: AI Art Study}\label{module-4-ai-art-study}

The primary legal application: comparing human-created art to
AI-generated art across multiple models and artistic styles.

\subsubsection{Data Loading and
Embedding}\label{data-loading-and-embedding}

Internal helper functions (\texttt{\_slugify},
\texttt{\_extract\_style\_from\_original\_class},
\texttt{\_filter\_embeddings\_by\_style}) handle filename normalization
and style-based filtering.

\begin{itemize}
\tightlist
\item
  \texttt{art\_load\_dataset(...)}: Loads images from an
  AI-ArtBench-like directory structure with \emph{stratified sampling}
  by artistic style. The \texttt{max\_images\_per\_style} parameter
  (default 250) ensures equal representation of each style within each
  category, yielding 2,500 images per category (250 × 10 styles). The
  \texttt{categories\_map} parameter maps target categories to folder
  prefixes. Supports optional manifest presence check for evolution
  analysis.
\item
  \texttt{art\_extract\_clip\_embeddings(...)}: Extracts 1024-dim
  semantic embeddings using CLIP (ViT-H-14-quickgelu, dfn5b pretrained).
  CLIP captures high-level semantic and stylistic features suitable for
  art comparison.
\end{itemize}

\subsubsection{Category-Level Analysis (Primary
Pipeline)}\label{category-level-analysis-primary-pipeline}

The primary AI Art analysis operates at the style level, computing Human
vs.~AI distinctiveness within each artistic movement:

\begin{itemize}
\tightlist
\item
  \texttt{art\_compute\_category\_mmd\_analysis(...)}: Iterates through
  styles, computes 3×3 MMD matrix (Human × AI(SD) × AI(LD)) per style.
\item
  \texttt{art\_compute\_category\_rejection\_rates(...)}: Computes
  per-style rejection rate curves.
\item
  \texttt{art\_compute\_category\_summary(...)}: Ranks styles by MMD
  (convergence ranking), identifies fast/median/slow converging styles
  via terciles, computes threshold sample sizes for 95\% power.
\end{itemize}

\subsubsection{Category-Level
Visualization}\label{category-level-visualization}

\begin{itemize}
\tightlist
\item
  \texttt{art\_plot\_category\_heatmaps(...)}: Grid of per-style 3×3
  heatmaps (2×5 layout for 10 styles).
\item
  \texttt{art\_plot\_category\_rejection\_rates(...)}: Rejection rate
  vs.~sample size, one line per style.
\item
  \texttt{art\_plot\_category\_mmd\_comparison(...)}: Grouped bar chart
  comparing Human-AI(SD) and Human-AI(LD) MMD across styles.
\item
  \texttt{art\_print\_category\_exposition\_summary(...)}: Formatted
  text summary of category-level findings.
\end{itemize}

\subsubsection{Evolution Analysis (Section
6.5)}\label{evolution-analysis-section-6.5}

\begin{itemize}
\tightlist
\item
  \texttt{art\_compute\_evolution\_analysis(...)}: Computes MMD between
  a human baseline and each model generation (LD → SD → SDXL → FLUX →
  FLUX-Krea). The human baseline sample is fixed to ensure that changes
  in MMD reflect model evolution, not sampling variation in the human
  reference. SDXL and FLUX images were generated using
  \texttt{generate\_images.py} with consistent prompts matching
  AI-ArtBench style categories. For MMD testing, comparisons use up to
  500 samples per group (controlled by \texttt{sample\_cap}).
\item
  \texttt{art\_plot\_evolution\_curve(...)}: Line plot of MMD vs.~model
  generation.
\item
  \texttt{art\_plot\_evolution\_grid(...)}: Creates a 3×6 grid of
  exemplar images showing visual evolution across model generations
  (Human → LD → SD → SDXL → FLUX → FLUX-Krea) for three representative
  styles. Images are hand-picked and stored in
  \texttt{art\_results/evolution\_grid\_images/}.
\end{itemize}

\subsubsection{Perceptual Paradox Analysis (Section
6.7)}\label{perceptual-paradox-analysis-section-6.7}

\begin{itemize}
\tightlist
\item
  \texttt{art\_extract\_vae\_embeddings(...)}: Extracts 16,384-dim VAE
  latents (64×64×4 flattened) from \texttt{stabilityai/sd-vae-ft-mse},
  the Stable Diffusion encoder fine-tuned on reconstruction loss. Tests
  whether distinctiveness is an artifact of the model's compression
  scheme.
\item
  \texttt{art\_extract\_dreamsim\_embeddings(...)}: Extracts embeddings
  from DreamSim, trained to match human perceptual judgments. Tests the
  ``perceptual paradox'': distinct in semantic space (CLIP) but
  indistinguishable in perceptual space (DreamSim).
\item
  \texttt{art\_compute\_perceptual\_paradox\_by\_style(...)}: The
  primary perceptual paradox analysis. Computes MMD for each of 10
  artistic styles across 3 embedding types (CLIP, DreamSim, VAE) and 3
  comparisons (Human vs AI(SD), Human vs AI(LD), AI(SD) vs AI(LD)),
  yielding 90 MMD values total. Returns per-style results and aggregate
  statistics (mean, std, min, max) with significance counts across
  styles.
\item
  \texttt{art\_plot\_perceptual\_paradox(...)}: Grouped bar chart
  visualizing the perceptual paradox results. Shows 3 comparison groups
  × 3 embedding bars per group, with error bars indicating min-max range
  across styles and ``n/10'' significance counts above each bar.
  Hypothesis: CLIP (AI-naive) \textgreater{} DreamSim (AI-aware)
  \textgreater{} VAE (AI-native).
\item
  \texttt{art\_compute\_embedding\_comparison(...)}: Legacy function for
  single-style embedding comparison. Superseded by
  \texttt{art\_compute\_perceptual\_paradox\_by\_style()} for aggregate
  analysis.
\item
  \texttt{art\_plot\_embedding\_comparison(...)}: Grouped bar chart for
  single-style embedding comparison.
\end{itemize}

\subsubsection{Style Comparison (Magnitude
Contextualization)}\label{style-comparison-magnitude-contextualization}

\begin{itemize}
\tightlist
\item
  \texttt{art\_compute\_style\_comparison(...)}: Computes MMD between
  pairs of human art movements (e.g., Impressionism vs.~Realism).
  Calibrates the ``legal distance'': is Human-AI divergence larger or
  smaller than the gap between major human movements? Uses exact style
  matching via \texttt{\_extract\_style\_from\_original\_class()} to
  avoid substring contamination (e.g., ``realism'' incorrectly matching
  ``surrealism'').
\end{itemize}

\subsubsection{Robustness and
Perturbation}\label{robustness-and-perturbation}

\begin{itemize}
\tightlist
\item
  \texttt{art\_run\_robustness\_suite(...,\ target\_style=None,\ embeddings\_raw=None,\ comparison\_pairs=None)}:
  Orchestrates all ablations. If \texttt{target\_style} is provided,
  filters to that style before analysis. The \texttt{embeddings\_raw}
  parameter accepts raw 1024-dim CLIP embeddings for dimensionality
  ablation (which tests reduction from high dimensions). The
  \texttt{comparison\_pairs} parameter accepts a list of tuples, e.g.,
  \texttt{{[}(\textquotesingle{}Human\textquotesingle{},\ \textquotesingle{}AI\ (SD)\textquotesingle{}),\ (\textquotesingle{}Human\textquotesingle{},\ \textquotesingle{}AI\ (LD)\textquotesingle{}){]}},
  enabling ablations across multiple comparisons in a single call.
  Results are saved with pair-specific suffixes (e.g.,
  \texttt{kernel\_ablation\_results\_h\_sd.pkl},
  \texttt{kernel\_ablation\_results\_h\_ld.pkl}).
\item
  \texttt{\_comparison\_pair\_suffix(pair)}: Helper function that
  converts comparison tuples to filename suffixes, e.g.,
  \texttt{(\textquotesingle{}Human\textquotesingle{},\ \textquotesingle{}AI\ (SD)\textquotesingle{})}
  → \texttt{\textquotesingle{}\_h\_sd\textquotesingle{}}.
\item
  \texttt{art\_perturbation\_analysis(...)}: Tests robustness to noise
  and watermarks, analogous to MNIST. Uses a paired sample design with
  joint UMAP projection; sample size controlled by \texttt{n\_samples}
  parameter.
\item
  \texttt{art\_perturbation\_table(...)}: Generates formatted p-value
  tables for perturbation results (used in manuscript Table 4).
\item
  \texttt{art\_print\_robustness\_summary(...)}: Formatted text summary
  of robustness findings.
\item
  \texttt{art\_save\_ablation\_tables(...)}: Prints and saves formatted
  ablation tables (kernel, bandwidth, dimensionality, embedding
  comparison, perturbation) to \texttt{art\_ablation\_tables.txt}.
\end{itemize}

\subsubsection{Memorization Audit (Item-Level
Analysis)}\label{memorization-audit-item-level-analysis}

The distributional MMD test answers whether the \emph{process} is
distinct; the memorization audit answers whether specific \emph{outputs}
exhibit suspicious similarity to training data. This bifurcated approach
is legally necessary: a generative process can be distributionally novel
(\(Q \neq P\)) while still occasionally regurgitating near-copies.

\emph{Key insight:} Distributional tests average over the output space,
potentially masking a small fraction of memorized outputs. If an AI
model is 99\% creative and 1\% regurgitative, MMD will indicate
``distinct,'' but copyright law still cares about that 1\%.

\begin{itemize}
\tightlist
\item
  \texttt{art\_compute\_memorization\_audit(embeddings\_raw,\ images,\ categories,\ original\_classes,\ config,\ output\_dir)}:
  The primary memorization detection function. For each AI-generated
  image, finds its nearest neighbor among human images \emph{within the
  same artistic style} using three complementary metrics:

  \begin{itemize}
  \tightlist
  \item
    \textbf{CLIP cosine similarity}: Semantic proximity in the 1024-dim
    embedding space. High similarity suggests the AI output captures the
    same semantic content as a human work.
  \item
    \textbf{SSIM} (Structural Similarity Index): Pixel-level structural
    correspondence measuring luminance, contrast, and structure. Detects
    visual copying that may not register semantically.
  \item
    \textbf{LPIPS} (Learned Perceptual Image Patch Similarity): Deep
    perceptual distance trained to match human similarity judgments.
    Bridges semantic and structural measures.
  \end{itemize}
\end{itemize}

\emph{Within-style comparisons:} Comparisons are constrained to the same
artistic style (e.g., AI impressionist images compared only to human
impressionist images). This is critical because cross-style comparisons
would conflate style distance with potential copying---an AI Baroque
image should not be flagged simply for being distant from human
Impressionism.

\emph{Threshold establishment:} The human-human baseline establishes
what ``normal'' within-style similarity looks like. For each human
image, we find its nearest neighbor among \emph{other} human images of
the same style. The 99th percentile of this distribution becomes the
detection threshold for CLIP and SSIM (where higher = more similar). For
LPIPS (where lower = more similar), we use the 1st percentile.

\emph{Exceedance rate:} The key output metric. What percentage of AI
outputs have nearest-neighbor similarity exceeding the human-human
baseline threshold? An exceedance rate near 1\% is expected by
construction (matching the threshold percentile); rates substantially
higher suggest systematic memorization.

\begin{itemize}
\item
  \texttt{art\_plot\_memorization\_audit(audit\_results)}: Generates a
  3-panel histogram visualization showing the distribution of
  nearest-neighbor similarities for Human-Human (baseline),
  AI(SD)-Human, and AI(LD)-Human, with threshold lines marked. Useful
  for understanding the full distribution of similarity scores.
\item
  \texttt{art\_plot\_memorization\_exceedance\_chart(audit\_results)}:
  The primary visualization for the memorization audit. Generates a
  grouped bar chart showing exceedance rates across metrics (CLIP, SSIM,
  LPIPS) and models (SD, LD). Bar heights represent the mean exceedance
  rate across styles; whiskers show the min-max range. A horizontal
  dashed line at 1\% marks the expected baseline (since the threshold is
  the 99th percentile of human-human similarity, exactly 1\% of human
  works ``exceed'' by construction). Rates at or below this line
  indicate the AI produces high-similarity outputs no more frequently
  than human artists.
\item
  \texttt{art\_save\_memorization\_audit\_table(audit\_results,\ filename)}:
  Saves a formatted table summarizing exceedance rates across all three
  metrics for both AI model types. This table provides the empirical
  basis for the ``bifurcated inquiry'' discussed in Section 2.
\end{itemize}

\emph{Legal significance:} The combination of high MMD (process
distinctiveness) and low exceedance rates (rare memorization) suggests a
tool with occasional defects rather than a systematic copying machine.
This distinction is relevant for remedies: targeted damages for specific
infringements vs.~broad injunctive relief for market substitution.

\subsection{Module 5: Exposition
Summaries}\label{module-5-exposition-summaries}

Functions that extract and format key results for academic writing:

\begin{itemize}
\tightlist
\item
  \texttt{print\_mnist\_exposition\_summary(...)}: Sample size
  thresholds, negative control statistics, significance rates.
\item
  \texttt{print\_patent\_exposition\_summary(...)}: Cross-section
  comparisons, negative control verification.
\item
  \texttt{art\_print\_category\_exposition\_summary(...)}: Per-style MMD
  values, convergence rankings, threshold sample sizes.
\item
  \texttt{print\_mnist\_robustness\_exposition\_summary(...)}:
  Consolidates all MNIST ablation findings.
\item
  \texttt{print\_art\_robustness\_exposition\_summary(...)}:
  Consolidates evolution and robustness findings for AI Art.
\end{itemize}

\subsection{Module 6: Main Execution}\label{module-6-main-execution}

\subsubsection{Configuration}\label{configuration}

All parameters are centralized in the configuration section preceding
the \texttt{if\ \_\_name\_\_\ ==\ "\_\_main\_\_"} block. Key parameters
include:

\begin{itemize}
\tightlist
\item
  \texttt{SAMPLE\_CAP\ =\ 500}: Maximum samples per class/category for
  MMD computations
\item
  \texttt{PATENT\_N\_SAMPLES\ =\ 500}: Samples per IPC section for
  patent study
\item
  \texttt{ART\_MAX\_IMAGES\_PER\_STYLE\ =\ 250}: Images per artistic
  style (stratified sampling)
\item
  \texttt{R\ =\ 500}: Permutation iterations for all hypothesis tests
\item
  \texttt{ALPHA\ =\ 0.01}: Significance level
\item
  UMAP: 64 dimensions, cosine metric
\end{itemize}

Extended parameters are organized into configuration dictionaries
(\texttt{MNIST\_ROBUSTNESS\_CONFIG}, \texttt{ART\_ROBUSTNESS\_CONFIG},
\texttt{ART\_PERTURBATION\_CONFIG},
\texttt{MEMORIZATION\_AUDIT\_CONFIG}, etc.) for clarity.

The memorization audit is controlled by:

\begin{itemize}
\tightlist
\item
  \texttt{RUN\_MEMORIZATION\_AUDIT\ =\ True}: Toggle to enable/disable
  the item-level analysis
\item
  \texttt{MEMORIZATION\_AUDIT\_CONFIG}: Dictionary containing
  \texttt{threshold\_percentile} (default 99) and \texttt{device} (e.g.,
  `mps' for Apple Silicon, `cuda' for NVIDIA)
\end{itemize}

\subsubsection{Execution Flow}\label{execution-flow}

\begin{enumerate}
\def\labelenumi{\arabic{enumi}.}
\tightlist
\item
  Seed Management: \texttt{set\_all\_seeds(seed)} initializes NumPy,
  PyTorch (including CUDA if available), and Python random seeds for
  reproducibility. For deterministic GPU operations, the function also
  sets \texttt{torch.backends.cudnn.deterministic\ =\ True}.
\item
  Directory Setup: Creates \texttt{mnist\_results/},
  \texttt{patent\_results/}, \texttt{art\_results/}.
\item
  MNIST Study: Train LeNet → Extract embeddings → MMD matrix → Rejection
  rates → Auto-select hardest pair → Robustness suite → Save \& plot.
\item
  Memory Cleanup: After MNIST completes, the model is deleted and GPU
  memory is released via \texttt{torch.cuda.empty\_cache()} (NVIDIA) or
  \texttt{torch.mps.empty\_cache()} (Apple Silicon). This ensures
  subsequent studies have maximum available memory.
\item
  Patent Study: Load \& embed → MMD matrix → Rejection rates → Save \&
  plot.
\item
  AI Art Study: Load \& embed → Category analysis → Evolution analysis →
  Perceptual paradox → Memorization audit → Robustness (on
  representative styles) → Perturbation → Save \& plot.
\item
  Exposition Summaries: Print formatted results for each study.
\end{enumerate}

\subsubsection{Automated Style
Selection}\label{automated-style-selection}

To avoid running expensive ablations on all 10 styles, the script
automatically selects three representative styles based on category
analysis results:

\begin{itemize}
\tightlist
\item
  One from the \emph{fast-converging tercile} (low MMD)
\item
  One from the \emph{median tercile}
\item
  One from the \emph{slow-converging tercile} (high MMD)
\end{itemize}

Robustness and perturbation analyses run only on these three styles.

\subsubsection{Output Formats}\label{output-formats}

\begin{itemize}
\tightlist
\item
  \texttt{.npy}: Pure numpy arrays (embeddings, MMD matrices, p-value
  matrices)
\item
  \texttt{.pkl}: Python dictionaries (rejection rates, ablation results,
  robustness suites)
\item
  \texttt{.npz}: Compressed numpy archives (patent embeddings
  dictionary)
\item
  \texttt{.png}: All visualizations (heatmaps, line plots, bar charts)
\item
  \texttt{.txt}: Formatted tables for reference during manuscript
  writing (e.g., \texttt{mnist\_ablation\_tables.txt},
  \texttt{patent\_summary\_tables.txt},
  \texttt{art\_ablation\_tables.txt},
  \texttt{art\_memorization\_audit\_table.txt})
\end{itemize}

\subsubsection{Embedding Caching}\label{embedding-caching}

CLIP embedding extraction is computationally expensive
(\textasciitilde30 minutes on GPU). To accelerate re-runs, the AI Art
study caches embeddings with validation:

\begin{itemize}
\tightlist
\item
  \textbf{Cache files}: \texttt{art\_clip\_embeddings\_raw.npy}
  (1024-dim), \texttt{art\_clip\_embeddings.npy} (64-dim UMAP-reduced),
  \texttt{art\_categories.npy}, \texttt{art\_original\_classes.npy}
\item
  \textbf{Validation}: On load, the cache is validated by checking (1)
  sample count matches loaded images and (2) category sets match. If
  validation fails, embeddings are re-extracted.
\item
  \textbf{Cache invalidation}: Changing
  \texttt{RUN\_EVOLUTION\_ANALYSIS} or modifying the image dataset
  invalidates the cache automatically.
\end{itemize}

This reduces ablation-only re-runs from \textasciitilde2 hours to
\textasciitilde30 minutes.

\subsection{Extending to New Domains}\label{extending-to-new-domains}

To apply this framework to a new domain (e.g., trademark logos, music
samples, literary texts):

\begin{enumerate}
\def\labelenumi{\arabic{enumi}.}
\item
  \emph{Write a data loader} returning \texttt{(samples,\ labels)} where
  samples can be any format your embedding extractor accepts.
\item
  \emph{Write an embedding extractor} that converts samples to
  fixed-length numerical vectors. Choose an embedding model appropriate
  to your domain:

  \begin{itemize}
  \tightlist
  \item
    Images: CLIP, ResNet, domain-specific models
  \item
    Text: SBERT, legal-specific transformers
  \item
    Audio: wav2vec, audio spectrograms + CNN
  \end{itemize}
\item
  \emph{Call the shared statistical functions}:

  \begin{itemize}
  \tightlist
  \item
    \texttt{mmd\_squared\_unbiased(emb\_A,\ emb\_B)} for distributional
    distance
  \item
    \texttt{permutation\_test(emb\_A,\ emb\_B,\ ...)} for significance
    testing
  \item
    \texttt{compute\_stability\_analysis(...)} and
    \texttt{compute\_kernel\_ablation(...)} for robustness checks
  \end{itemize}
\item
  \emph{Adapt visualization functions} or use the outputs directly for
  custom plots.
\end{enumerate}

The statistical core (Module 1) requires no modification. Domain
expertise is concentrated in the data loader and embedding extractor.

\end{document}